\documentclass[12pt]{article}

\usepackage{amsmath}
\usepackage{epsfig, amssymb, amsfonts}
\usepackage{amsthm}
\usepackage{amstext}
\usepackage{amsopn}
\usepackage{enumerate}
\usepackage{mathrsfs}
\usepackage{subfigure}
\usepackage{graphicx}
\usepackage[left=2.3cm,top=2cm,right=2.3cm,bottom=2cm, nohead,foot=1cm]{geometry}
\numberwithin{equation}{section}

\newtheorem{theorem}{Theorem}[section]
\newtheorem{lemma}[theorem]{Lemma}
\newtheorem{assumptions}[theorem]{Assumptions List}

\newtheorem{proposition}[theorem]{Proposition}

\newtheorem{corollary}[theorem]{Corollary}

\newcommand{\R}{{\mathbb R}}

\newcommand{\p}{{p}}

\newcommand{\Z}{{\mathbb Z}}

\title{Diffusive behavior for randomly kicked
Newtonian particles in a spatially periodic medium }

\author{\textbf{Jeremy Clark}\footnote{jclark@mappi.helsinki.fi}\\  Department of Mathematics, 
 University of Helsinki \\ 
Helsinki 00014,  Finland \vspace{.3cm} \\ 
 \textbf{Christian Maes}\footnote{christian.maes@fys.kuleuven.be } \\ Instituut voor Theoretische Fysica, 
K.U.Leuven \\ 
Celestijnenlaan 200D, Belgium }

\begin{document}
\maketitle

\begin{abstract}
We prove a central limit theorem for the momentum distribution of
a particle undergoing an unbiased spatially periodic random
forcing at exponentially distributed times without friction. The
start
 is a linear Boltzmann equation for the phase space density,
where the average energy of the particle grows linearly in time.
Rescaling time, the momentum converges to a Brownian motion, and
the position is its time-integral showing superdiffusive scaling
with time $t^{3/2}$. The analysis has two parts: (1) to show that
the particle spends most of its time at high energy, where the
spatial environment is practically invisible; (2) to treat the low
energy incursions where the motion is dominated by the
deterministic force, with potential drift but where symmetry
arguments cancel the ballistic behavior.

\end{abstract}

\section{Introduction}

Recent times show a renewed great interest in obtaining diffusive
behavior from microscopically defined dynamics.  The motivation is
much older, to derive, as Fourier, Navier or Boltzmann first did
in their ways and times, irreversible and dissipative behavior
starting from the reversible microscopic laws.  The limiting
behavior is often associated to a conserved quantity like energy
in classical mechanics and the challenge is then to express the
(energy) current in terms of gradients of the (energy) density
itself.
 Obviously, for the sharpness of the limit, some
scaling must be done, combined with typicality arguments on the
level of the initial or boundary conditions.  For example, more
recently the search for a rigorous derivation of Fourier's law of
heat conduction was relaunched in \cite{fourier} and many attempts
and models have been taken up after that.\\

More modestly one starts half-way with an effective description on
the level of single-particle dynamics.  The one-particle phase
space density then really refers to a cloud of weakly interacting
particles brought in contact with some environment, and the
conserved quantity is simply the number of particles. In the
present paper we study the diffusive scaling limit of a massive
particle in a one-dimensional periodic potential to which we add
random forcing. The latter is not derived from first principles,
but has various physical motivations. Heuristically, the forcing
corresponds to the random collisions with an effectively infinite
temperature granular bath. The granular structure is in the
discrete kicks the particle undergoes at random exponentially
distributed times. Depending on its present position, the
distribution of the momentum kicks differs. Together with the
potential that specifies the spatial inhomogeneity and makes the
problem nontrivial.\\

 The
  spatial heterogeneity brings us to a second
  motivation of the present work:
  the study of active particles where flights of ballistic motion are interrupted by spatially
depending re-orientations, or the statistical characterization of
particle trajectories in active heterogeneous fluids, see e.g.
\cite{Weihs}.  Our present work adds a rigorous result
establishing under what conditions diffusive behavior for the
momentum and superdiffusive behavior for the position get
realized.\\

Finally, the present contribution fits in the long tradition of
proofs of the central limit theorem (invariance principle, and its
modifications) for additive functionals, position as the integral
of the momenta and momenta as the integral of the forces. Much of
all that for studies of interacting particle systems started from
the pioneering work in \cite{kv}.  In the spirit of the present
work the papers \cite{masi,gold} added very important symmetry
considerations, making it possible to apply the work to strongly
dependent variables.  The fact that these arguments avoided the
use of mixing assumptions or strong enough decay of
time-correlations appears like an important lesson for today's
pursuit of diffusive behavior mentioned at the very beginning of
this introduction.  Indeed, one often emphasizes that strong
enough chaoticity assumptions are needed in the mathematical
control of the transition from the microscopic reversible laws to
macroscopic irreversible behavior.  One then refers for example to
the problem of obtaining regular transport properties via
well-controlled Green-Kubo expressions where some temporal decay
certainly seems necessary.  These Green-Kubo relations are however
needed only for very special observables, and not for all possible
even microscopically defined quantities.  It is therefore very
welcome if symmetry considerations can help to establish diffusive
behavior for certain classes of functionals that share symmetry
properties with the relevant observables of statistical mechanics.
The present paper is not starting from the microscopic classical
mechanical world, but it does deal with the
problem of exploiting symmetry to cancel ballistic behavior.\\
Other more recent work on the central limit theorem that shares
important ambitions with the present study includes
\cite{olla,jara,bovier,hairer}.  For the study of fluctuations in
Markov processes with an overview of central limit results, we
refer to the recent book \cite{landim}.\\

 The next section
  introduces the model, the results and the main strategy of the proof.  The momentum variable is not
  autonomous since it is coupled to the position of the particle.  Its changes come from two sources, the
  momentum jumps by the external Poisson noise
  and the acceleration due to the presence of the potential.
  That is translated in the structure of the argument.  The idea is to obtain a martingale central limit theorem
  for the momentum jumps while the effect of the potential should vanish in the long time limit.  Section \ref{SecEnergy} establishes that, most of the time, the particle's energy
    grows linearly with time.     That is sufficient to show in Section \ref{SecReflection} that
    the absolute value of the momentum process converges in
    distribution to the absolute value of a Brownian motion.  Next, in Section \ref{SecEstimates},
     follow the estimates characterizing the motion at high energy, where the
      (bounded) potential has very little effect.  The low energy motion is
      discussed under section \ref{SecDrift}.  There the drift due to the
      potential gets controlled by symmetry arguments.  The combination of high
      and low energy estimates yields the main result of Section \ref{SecMain}.

\section{Main result}

\subsection{Informal description}
Consider a one-dimensional classical particle whose position and
momentum $(X_{t}, K_{t})$ evolve deterministically with
Hamiltonian $H(x,\,k)=\frac{1}{2}k^{2}+V(x)$ for some bounded
periodic potential $0\leq V(x)\leq \bar{V}$ except at Poisson
times at which the particle may receive a momentum kick from the
environment. That is, independent of its current momentum $k$ and
at the rate $j_x(v)$ when its current position is $x$ the particle
receives a momentum jump $v$.
 On the level of the phase space densities, the
dynamics we consider is then governed by the linear Boltzmann
equation
\begin{align}\label{Boltzmann}
\frac{d}{dt}P_{t}(x,k)= -k\, \frac{\partial P_{t} }{\partial
x}(x,k)+\frac{dV}{dx}(x)\frac{\partial P_{t}}{\partial k}(x,k) +
\int_{\R} dv\,j_{x}(v)\,(P_{t}(x,k-v)-P_{t}(x,k))
\end{align}
for the phase space probability density $P_{t}(x,k)\in
L^{1}(\R^{2})$ for the particle at each time $t\geq 0$.
 The rates $j_x$ have the same periodicity as the
potential $V$.  It is however
probabilistically simpler to imagine a universal Poisson clock
having rate $\mathcal{R}>0$, such that when the alarm rings, a
(biased) coin is tossed to decide whether or not a momentum kick
will occur. The probability of the coin and the distribution of
the momentum jump $v$ are respectively $0\leq \kappa(a) \leq 1$
and $\mathcal{P}_{a}(v)$, where $a=X_{t}\,\textup{mod}\,1$  is the
position (modulo the period $1$ of the potential $V$) at the
Poisson time $t$.  We assume that the momentum jumps are symmetric
$\mathcal{P}_{a}(-v)=\mathcal{P}_{a}(v)$ and that there is a
uniform lower bound for the coin probabilities $0< \nu \leq
\kappa(a)$.  Then, in \eqref{Boltzmann},
$j_{x}(v)=\mathcal{R}\,\kappa(x)\mathcal{P}_{x}(v)$.

 Our main result is to show, under certain technical conditions,
 that the normalized variables
 $(t^{-\frac{3}{2}}X_{s\,t},t^{-\frac{1}{2}}K_{s\,t})$, $s\geq 0$ approach
 the process $(\int_{0}^{s}dr\,\mathbf{B}_{r},\mathbf{B}_{s})$ in distribution
 where $\mathbf{B}_{s}$ is Brownian motion whose diffusion
  constant $\sigma$ depends on the spatial average of the
  periodic noise $\sigma=\int_{0}^{1}da\,\int_{\R}dv\,j_{a}(v)\,v^{2}$.

There is clearly no energy relaxation in~(\ref{Boltzmann}), since,
no matter where you start, the time-derivative of the expected energy satisfies
\begin{multline}
 \frac{d}{dt}\mathbb{E}\big[E_{t}\big]=
  \frac{d}{dt}\int_{\R^{2}}dx\,dk\, (\frac{k^{2}}{2}+V(x))\, P_{t}(x,k) = \frac{1}{2}\int_{\R^{2}}dx\,
   dk\, \big(\int_{\R} dv\,j_{x}(v)\,v^{2}\big)\,P_{t}(x,k)
\end{multline}
and thus the mean energy grows linearly as
$$\mathbb{E}\big[E_{0}\big]+
   \frac{t}{2}\inf_{a\in [0,\,1]}\int_{\R} dv\,j_{a}(v)\,v^{2} \leq \mathbb{E}
   \big[E_{t}\big] \leq  \mathbb{E}\big[E_{0}\big]+    \frac{t}{2}\sup_{a\in [0,\,1]}\int_{\R} dv\,j_{a}(v)\,v^{2}.  $$

Moreover, as we will show, by time $t$ not only the average
but also the typical energy is of order $t$. Since the
potential $V(x)$ is bounded, the absolute value of the momentum is
then $|k|\propto t^{\frac{1}{2}}$.  As a consequence of having
high momentum, the particle will pass through one period of the
potential (periodic cell) much faster than the time scale of the
Poisson clock governing the noise.  The particle then effectively
``feels'' a spatial average of the noise in which the averaged
distribution of a jump $\tilde{P}(v)$ and the averaged Poisson
rate are $\tilde{\mathcal{R}}$
$$\tilde{P}(v)=
\frac{\int_{0}^{1}da\,\kappa(a)\,\mathcal{P}_{a}(v)}{\int_{0}^{1}da\,\kappa(a)}\quad
  \text{ and }\quad \tilde{\mathcal{R}}=\mathcal{R}\int_{0}^{1} da\,\kappa(a)
$$
 Moreover, at very high momentum, the force field by the potential can only
 displace the momentum by relatively negligible values.
The effective dynamics at high energy is thus
\begin{align}\label{TIBoltzmann}
\frac{d}{dt}P_{t}(x,k)= -k\, \frac{\partial P_{t} }{\partial x}(x,k)+
 \tilde{\mathcal{R}}\int_{\R} dv\,\tilde{P}(v)\,(P_{t}(x,k-v)-P_{t}(x,k))
\end{align}
This dynamics is translation invariant and so the momentum process
has now become a Markov process.
Showing that $(t^{-{\frac{3}{2}}}X_{st},\,t^{-\frac{1}{2}}K_{st})$
converges to
$(\int_{0}^{s}dr\,\mathbf{B}_{r},\,\mathbf{B}_{s})$ is then
straightforward.

Another way of expressing this result on the level of
single-time marginals is to consider the rescaled density
$t^2\,P_{t}(t^{3/2}\,x,t^{1/2}\,k)$ at time $t$.  Its limit
$t\uparrow \infty$ is Gaussian $P_{\infty}(x,k)$,
 $$P_{\infty}(x,k)=  \frac{\sqrt{3}}{\pi\,\sigma}
  e^{-\frac{6}{\sigma}(x-\frac{k}{2})^{2}-\frac{1}{2\sigma}k^{2} }   $$
and $\sigma=\int_{0}^{1}dx\int_{\R}
dv\,j_{x}(v)\,v^{2}$.  The coupling between position $x$ and momentum $k$ results from the correlation between
the Brownian motion (for the momentum) and its time-integral (for the position).

\subsection{Strategy of proof}
The process $(X_t,K_t)$ is Markovian over the set of
right-continuous paths (having left limits) from $t\in\R^+$ to $\R^2$, bounded over
finite time intervals. Since the position variable is an integral of the
momentum, the proof that
$(t^{-\frac{3}{2}}X_{st},\,t^{-\frac{1}{2}}K_{st})$ converges to
$(\int_{0}^{s}dr\,\mathbf{B}_{r},\mathbf{B}_{s})$ for a Brownian
motion $\mathbf{B}_{s}$ is implied by showing that the momentum
component converges to a Brownian motion.   The momentum process
can be written as
\begin{align}\label{spl}
K_{st} = K_{0}+M_{st}+\int_{0}^{st}dr\,\frac{dV}{dx}(X_{r}
\big)
\end{align}
where $M_{t}$ is the martingale of jumps, $M_t = \sum_s^t v_s$
over the jump times $0\leq s\leq t$ in the Poisson process at rate
$\mathcal{R}$ and $v_s$ is the actual momentum kick.  On the other
hand, $\int_{0}^{t}dr\,\frac{dV}{dx}(X_{r})$ is the net drift due
to the conservative force up to time $t$.

The analysis splits into two semi-independent parts corresponding
to the two last terms in \eqref{spl}. First we show that the
momentum jump part $t^{-\frac{1}{2}}M_{st}$ converges to a
Brownian motion.  That requires establishing a martingale central
limit theorem. Because of the inhomogeneity in the momentum
jumps we need to prove that there is asymptotic regularity in the variances
of the momentum jumps (quadratic variation process). That is realized because the particle spends
most of its time at high energy where translation invariance is recovered.  We call that
 the high energy analysis.\\
Secondly, for the low energy analysis we show that the drift process
$t^{-\frac{1}{2}}\int_{0}^{st}dr\,\frac{dV}{dx}(X_{r})$ makes a
vanishing contribution for large times; in other words the variance of the time integral of the drift converges to zero.
Indeed note again that the final process, obtained after scaling, does
not depend on the potential $V$. The superdiffusion of the position makes the position process almost deterministic  and the random kicks become rare compared to the fast
movement of the particle.  Mathematically, in section 6 we use that periods of low energy are well separated by
times of high energy.  There is therefore some independence between the low energy incursions.  Moreover symmetry arguments constrain the gained momentum in each such incursion to have zero expectation.

\subsection{Main theorem}
Our main mathematical result is a central limit theorem for the momentum process.
We give here the precise statement.

\begin{assumptions}\label{AssumpOne}\text{ }

\begin{enumerate}[I.]
\item There exists
 $ 0< r_{1} $ such that for all $a\in [0,\,1]$,
   $r_{1}\leq   \int_{\R}dv\,j_{a}(v)\,v^{2}.$
\item There is $\rho>0$ such that for all $a\in [0,\,1]$,
  $\int_{\R}dv\,\mathcal{P}_{a}(v)\,v^{4}\leq \rho $.
\item $j_{a}(v)=j_{a}(-v)$ \item $\bar{V} > V(x)\geq 0$ is bounded
and has a bounded derivative.
\end{enumerate}

\end{assumptions}

The fist three assumptions are on the rate of momentum jumps. They
should be symmetric, allow spreading but still have a fourth
momentum.  For the Hamiltonian part, both the potential and the
force is bounded.  The assumptions of List~\ref{AssumpOne}  are
designed to be the minimal assumptions for
Section~\ref{SecEnergy}, and most results from later sections
require both List~\ref{AssumpOne} and some of the assumptions from
List~\ref{AssumpTwo}.

\begin{assumptions}\label{AssumpTwo}\text{ }

\begin{enumerate}[i.]
\item There exists $\mathcal{C}$ and $\eta>0$ such that for all
$a\in[0,\,1]$ and $v,w\in \R$ with $|v|-|w|\geq 0$
$$  \mathcal{P}_{a}(w)\leq \mathcal{C}\,e^{-\eta (|v|-|w|)}\mathcal{P}_{a}(v).      $$

\item    There exists a $\mu$ such that for all $a\in [0,\,1]$,
$$\sup_{v\in \R}\Big[\big(\mathcal{P}_{a}(v)\big)^{-1}(1+|v|)^{-1} \sup_{|w-v|\leq 1}\big|\frac{d\mathcal{P}_{a}}{dv}(w)\big| \Big]\leq \mu. $$.

\item There exists a reflection $R$ on the torus such that
$V(R(x))=V(x)$ and $j_{R(a)}(v)=j_{a}(v)$ for $a\in [0,\,1]$ and
$v\in \R$.

\end{enumerate}

\end{assumptions}

Condition (\textit{i}) implies that the Laplace
transform of $\mathcal{P}_{a}$ is finite in a neighborhood around
zero and thus that the fourth moment as in (\textit{II}) and all
other moments are finite.

In later sections, $r_{1}$, $r_{2}$, $\nu$ will be defined as
$$r_{1}= \inf_{a\in [0,1]} \int_{\R}dv\,j_{a}(v)\,v^{2}, \quad r_{2}=\sup_{a\in [0,1]} \int_{\R}dv\,j_{a}(v)\,v^{2}, \quad  \nu=\inf_{a\in [0,1]} \kappa_{a}.  $$
Since $j_{a}(v)=\mathcal{R}\kappa_{a}\mathcal{P}_{a}(v)$, the condition (\textit{I}) and (\textit{II}) with Jensen's inequality imply that $ 0< r_{1}\mathcal{R}^{-1}\rho^{-\frac{1}{2}}\leq \nu $.  Also by (\textit{II}) and Jensen's inequality, $r_{2}<\mathcal{R}\rho^{\frac{1}{2}}<\infty$.

(\textit{III}) and (\text{iii}) are the symmetries that we assume
for the dynamics. (\textit{III}) says that for every point $a\in
[0,1]$ in the periodic cell, the rate of kicks by a momentum $v$
occurs with the same rate as kicks by a momentum $-v$. (\text{iv})
specifies that in addition to the periodicity of the dynamics,
there is also a spatial reflection symmetry.  The combination of
these symmetries forms a  ``momentum time-reversal symmetry''
which is used in Section~\ref{Torus}.

(\textit{ii}) is a technical assumption so that the values of the derivative $|\frac{d\mathcal{P}_{a}}{dv}(w)\big|$ for $w$ in a neighborhood around $v$ cannot be to large compared to the value $\mathcal{P}_{a}(v)$.  The constraint becomes more flexible at large $|v|$, where the ratio is allowed to increase as $|v|$.  This condition effectively forbids densities with tails that vanish faster than a Gaussian density $G(v)$ in which
$$\frac{ \big| \frac{dG}{dv}(v)\big| }{G(v)}\propto |v|.$$
The support of $\mathcal{P}_{a}(v)$ for each $a$ cannot be finite for instance.
Avoiding this decay of $\mathcal{P}(v)$ as $|v|\rightarrow \infty$ is not essential to the analysis, but generalizing the condition (\textit{ii}) (for instance by replacing $(1+|v|)$ by $(1+|v|)^{m}$) requires making other conditions more complicated.

\begin{theorem}[Main result]\label{MainThm}
Assume List~\ref{AssumpOne}, List~\ref{AssumpTwo}, and that the
initial joint  phase space distribution $P_{0}(x,\,k) \in
L^1(\R^2)$ has finite second moments.  In the limit $t\rightarrow
\infty$, $(t^{-\frac{3}{2}}X_{st},t^{-\frac{1}{2}}K_{st})$
converges in distribution to
$(\int_{0}^{s}dr\,\mathbf{B}_{r},\mathbf{B}_{s})$ where
$\mathbf{B}_{s}$ is Brownian motion with diffusion constant
$\sigma=\int_{0}^{1}da\int_{\R} dv\,j_{a}(v)\,v^{2}$.
\end{theorem}

\section{A martingale central limit theorem}\label{SecEnergy}
In this section, we prove that the typical energy for the particle
is on the order of $t$. That implies a regularity in the momentum
process, at least concerning its absolute value and for the
quadratic variation of the momentum jumps.  The net result is a
martingale central limit theorem for the martingale part in
\eqref{spl}.  Note that we always assume the natural filtration
${\cal F}_t$ specifying the Markov process up to time $t$.

\begin{theorem} \label{MartBrown}
Assume List~\ref{AssumpOne} and (\textit{i})-(\textit{ii}) of List~\ref{AssumpTwo} and that the initial joint distribution $P_{0}(x,\,k)$ has finite second moments.   Then $t^{-\frac{1}{2}}M_{st}$ converges in distribution to a Brownian motion $\mathbf{B}_{s}$ with diffusion constant $\sigma$.
\end{theorem}

The proof follows in Section \ref{SecMain}.  It will be built on  the lemma's below and in the next section.

The following lemma relies on a martingale central limit theorem~\cite{MCLT} and on having some bounds for the time-change of that central limit theorem.   The result applies to more general class of martingales, but we develop it here to the martingale process
 $$M_{t}^{\prime}= \sum_{n=1}^{\mathcal{N}_{t}}w_{n}S(K_{t_{n}^{-}}),   $$
where $t_{n}$ are the Poisson time for the underlying Poisson clock $\mathcal{N}_{t}$ with rate $\mathcal{R}$, $S(K_{t^{-}})$ is the left limit up to time $t$ for the sign of the momentum, and $w_{n}=M_{t_{n}}-M_{t_{n}^{-}}$.  Note that $w_{n}$ is zero if at the Poisson time $t_{n}$ there is no momentum jump, and it is equal to the momentum jump if it does happen.

\begin{lemma}\label{TimeChange}
For $T_{\epsilon}^{(t)}=\int_{0}^{1}dr\, \chi(|t^{-\frac{1}{2}}M_{rt}'|> \epsilon ) $,  there exists a $C>0$ such that for all $\epsilon,\,\delta>0$
$$ \liminf_{t\rightarrow \infty}\textup{Pr}\big[ T_{\epsilon}^{(t)} \geq 1- \delta\big]\geq 1- C\frac{r_{2}^{\frac{1}{2}}}{r_{1}}\frac{\epsilon}{\delta}.   $$
 The same result remains true with $|t^{-\frac{1}{2}}M_{st}'|$ replaced by $t^{-\frac{1}{2}}M_{s\,t}' - \inf_{0\leq u \leq s} t^{-\frac{1}{2}}M_{u\,t}'$.

\end{lemma}

\begin{proof}
We start from the lower bound:
\begin{multline}\label{SubmergeTime}
\textup{Pr}\big[ T_{\epsilon}^{(t)} \leq 1-\delta\big]=1- \textup{Pr}\Big[ \int_{0}^{1}dr\,\chi\big(t^{-\frac{1}{2}} |M_{st}'|\leq \epsilon \big) \geq \delta\Big]\\  \geq 1- \frac{1}{\delta}\mathbb{E}\big[\int_{0}^{1}ds\, \chi(|t^{-\frac{1}{2}}M_{s\,t}'|\leq \epsilon )     \big]=1- \frac{1}{\delta}\int_{0}^{1}ds\, \textup{Pr}\big[|t^{-\frac{1}{2}}M_{s\,t}'|\leq \epsilon     \big]
\end{multline}

Define $\tilde{B}_{u}^{(t)}=t^{-\frac{1}{2}}M_{\tau_{u}t}^{\prime}$, where
$\tau_{u}$ is the hitting time
$$\tau_{u}=\inf \left \{ s\geq 0\, \Big|\, \frac{1}{r_{1}t}\langle M^{\prime}\rangle_{st} \, \geq u \,  \right \},
$$
and $\langle M^{\prime}\rangle_{t}$ is the predictable quadratic variation of $M^{\prime}$ up to time $t$.  In our situation, $\langle M^{\prime}\rangle_{t}$ has the form
$$\langle M^{\prime}\rangle_{t}= \int_{0}^{t}dr\,\int_{\R}dv\, j_{X_{r}}(v)\,v^{2}.  $$
By the martingale central limit theorem,  $\tilde{B}_{u}^{(t)}$ converges to a Brownian motion in the uniform metric.  The Lindberg condition is guaranteed by the boundedness of the fourth moment of single momentum jumps in (\text{II}) of List~\ref{AssumpOne} and the fact that the jump times occur according to a Poisson clock (having rate $\mathcal{R}$).

 We  also have  $t^{-\frac{1}{2}}M_{st}^{\prime}=\tilde{B}_{R_{s}}^{(t)}$, where $R_{s}= \frac{1}{r_{1}t}\langle M^{\prime}\rangle_{st} ,$
since $\tau_{u}$ and $R_{s}$ are inverses of one another. By
(\text{II})  of List~\ref{AssumpOne} and the discussion following
it, $ \int_{\R}dv\,j_{a}(v)\,v^{2}$ ranges between the values
$0<r_{1}\leq r_{2}$ for $a\in[0,1]$.  It follows that,
\begin{align}\label{Tandori}
u\,\frac{r_{1}}{r_{2}}\leq \tau_{u}\leq u \quad \text{and} \quad  s\leq R_{s}\leq s\frac{r_{2}}{r_{1}}.
\end{align}

By~(\ref{Tandori}), $R_{s}$ has the range $s\leq R_{s}\leq
\frac{r_{2}}{r_{1}}s $, and thus
$$\int_{0}^{1}ds\, \chi \big(| t^{-\frac{1}{2}} M_{st}^{\prime}|\leq \epsilon      \big)= \int_{0}^{1}ds\, \chi\big(|\tilde{B}_{R_{s}}^{(t)}|\leq \epsilon    \big) \leq \int_{0}^{\frac{r_{2}}{r_{1}} }du\, \chi \big(|\tilde{B}_{u}^{(t)}|\leq \epsilon \big).   $$

Taking the expectation of the right-hand side and using the fact that $\tilde{B}_{s}^{(t)}$ approaches a Brownian motion, we have that
\begin{align*}
\mathbb{E}\big[\int_{0}^{\frac{r_{2}}{r_{1}} }du\, \chi \big(|\tilde{B}_{u}^{(t)}|\leq \epsilon \big) \big]=\int_{0}^{\frac{r_{2}}{r_{1}} }du\,\textup{Pr}\big[|\tilde{B}_{u}^{(t)}|\leq \epsilon \big]\longrightarrow \int_{0}^{\frac{r_{2}}{r_{1}} }du\,\int_{-\epsilon}^{\epsilon}dx\,\frac{e^{-\frac{x^{2}}{2r_{1}s} }}{\sqrt{2\pi r_{1}s }}
\end{align*}
By a change of variables, the right side is bounded by a constant multiple of $\epsilon\, r_{1}^{-1}r_{2}^{\frac{1}{2}}$. With~(\ref{SubmergeTime}) this proves the result.

To generalize the result to $H_{s}^{(t)}=t^{-\frac{1}{2}}M_{s\,t}' - \inf_{0\leq a \leq s} t^{-\frac{1}{2}}M_{a\,t}^{\prime}$, we make the same time-change $\tau$, to define a process   $Z_{u}^{(t)}=H_{\tau_{u}}^{(t)}$.  Since $\tilde{B}_{u}^{(t)}=t^{-\frac{1}{2}}M_{\tau_{u}}'$
converges to a Brownian motion, it will follow that $Z_{u}^{(t)}$ converges to the absolute value of a Brownian motion.  Indeed, the function $f:L^{\infty}([0,\,1])\rightarrow L^{\infty}([0,\,1] )$ defined by $f(x_{s})=x_{s}+\textup{sup}_{0\leq r\leq s}-x_{r}$ (read the supremum as an essential supremum) satisfies $\|f(x_{s})-f(y_{s})\|_{\infty}\leq 2\|x_{s}-y_{s}\|_{\infty}$.  Thus the convergence of $\tilde{B}_{u}^{(t)}$ to a Brownian motion $\mathbf{B}_{u}$ implies that $f\big(\tilde{B}_{u}^{(t)}\big)$ converges in distribution to $f(\mathbf{B}_{u})$.  However, by the basic result for Brownian motion~\cite{Karat},  $ f(\mathbf{B}_{u})$ is equal in distribution to $|\mathbf{B}_{u}|$.  Thus we can apply the same reasoning as above to get the result.

\end{proof}

The lemma below will be used in the proof of Lemma~\ref{EIL}.  Its proof follows from basic calculus but requires consideration of several cases.

\begin{lemma}\label{StickyCases}
For $\bar{V}=\sup_{x\in\R}V(x)<\infty$, there exists a $c>0$ such that
\begin{multline}\label{State}
\Big| 2^{\frac{1}{2}}\,w\,S(k)-\Big[\big(\frac{1}{2} |k+w|^{2}+V \big)^{\frac{1}{2}}-\big(\frac{1}{2}|k-w|^{2}+V\big)^{\frac{1}{2}}  \Big]  \Big|^{2} \\ \leq  2\,w^{2}1_{|w|>J}+ c\,J\Big[ \big(\frac{1}{2} |k+w|^{2}+V\big)^{\frac{1}{2}}+\big(\frac{1}{2}|k-w|^{2}+V \big)^{\frac{1}{2}}-2\big(\frac{1}{2} k^{2}+V\big)^{\frac{1}{2}}  \Big]
\end{multline}
for all $k,w\in \R$, $0\leq V\leq \bar{V}$, and $J\geq V^{\frac{1}{2}}$.

\end{lemma}

\begin{proof}
When $V=0$, we use the identity
$$\big|S(k)w-\frac{1}{2}\big(|k+w|-|k-w|\big)\big|= \frac{1}{2}\big(|k+w|+|k-w|-2|k| \big). $$
Hence, for $c=2$,
\begin{align}\label{VZeroCase}
\big| 2^{\frac{1}{2}}S(k)w-2^{-\frac{1}{2} }\big(|k+w|-|k-w|\big)\big|^{2}\leq 2\,v^{2}\chi(|w|> J)+\frac{c\,J}{2}\big(|k+w|+|k-w|-2|k| \big),
\end{align}
since $\big|S(k)w-\frac{1}{2}\big(|k+w|-|k-w|\big)\big|\leq |w|$.

 When $V\neq 0$, we define
\begin{eqnarray*}
M(w,k)&=&2\,w\,S(k)-\big( |k+w|^{2}+1 \big)^{\frac{1}{2}}+\big( |k-w|^{2}+1\big)^{\frac{1}{2}}   \\
N(w,k)&=& \big( |k+w|^{2}+1\big)^{\frac{1}{2}}+\big(|k-w|^{2}+1 \big)^{\frac{1}{2}}-2\big( k^{2}+1\big)^{\frac{1}{2}}
\end{eqnarray*}
 We can divide the inequality~\eqref{State} by $V$, and let $2^{-\frac{1}{2}}V^{-\frac{1}{2}}k\rightarrow k$, $2^{-\frac{1}{2}}V^{-\frac{1}{2}}w\rightarrow v$, and $J  \,V^{-\frac{1}{2}}\rightarrow J$ so that the inequality we must show is
\begin{align}\label{Normalized}
|M(w,k)|^{2} \leq  4\,w^{2}1_{|w|>J}+ c\,J\,N(w,k)
\end{align}
for some $c$.  For $|w|>J$ the inequality is trivial, since the
$M(w,k)$ is smaller than $2\,|w|$.  Hence we can take  $J=|w|\vee
1$.   If we show that there is such a $c^{\prime}$ such
that~(\ref{Normalized}) holds for $(w,k)\in\R^{2}-D$ for the
compact domain $D=([-L,-L^{-1}]\cup [L^{-1}, L])\times [-L,L]$ for
some $L>1$, then, since $M(w,k)$ is bounded inside $D$ and
$N(w,k)$ is bounded away from zero inside of $D$, it follows that
there exits a $c$ verifying~\eqref{Normalized} for all
$(w,k)\subset \R^{2}$. The proof that there exits such a
$c^{\prime}$ for all $(w,k)\in \R^{2}-D$ for some $L\gg 1$ is
partitioned into the following main cases
\begin{enumerate}
\item $|w|\ll 1,$

\item $|w|\gg 1$ and $|k|\geq |w|+|w|^{\frac{1}{2}},$

\item $|w|\gg 1$ and $|k+w|\leq |w|^{\frac{1}{2}}$ or $|k-w|\leq |w|^{\frac{1}{2}}$,

\item $|w|\gg 1$ and $|k|\leq |w|^{\frac{1}{2}}.$

\end{enumerate}
We go through each case and we show that $|M(w,k)|^{2}$ is smaller
than some constant multiple of $(|w|\vee1)\,N(w,k)$. \vspace{.2cm}

\noindent Case (1):  The origin can only be between $k-w$ and $k+w$ when $|k|\leq |w|  \ll 1$, in which case
$$|M(w,k)|^{2}\leq w^{2}<2w^{2}=4\big[\frac{1}{2}(k+w)^{2}+\frac{1}{2}(k-w)^{2}-k^{2}\big]\approx 4 N(w,k).$$
When the origin is not between $k-w$  and $k+w$ then by Taylor's formula for $\big(|k+w|^{2}+1 \big)^{\frac{1}{2}}$ up to first order and using that $k+a$ has the same sign as $k$ for $|a|\leq |w|$,
\begin{align}\label{CapM}
M(w,k)= S(k) \int_{-w}^{w}     da \frac{|k+a|-\big( |k+a|^{2}+1 \big)^{\frac{1}{2}} }{\big( |k+a|^{2}+1 \big)^{\frac{1}{2}} }.
\end{align}
Using Taylor's formula to second order,
\begin{align}\label{CapN}
N(w,k)=\int_{-|w|}^{|w|}da\,\frac{|w|-|a|}{ \big( |k+a|^{2}+1 \big)^{\frac{3}{2}} }.
\end{align}
 for all $w$ and $k$.  From these formulas we can take upper and lower bounds for $|M(w,k)|^{2}$ and $N(w,k)$ respectively:
$$|M(w,k)|^{2}\leq |w|^{2}\frac{\big|\, |k|-|w|-\big((|k|-|w|)^{2}+1\big)^{\frac{1}{2}}\big|^{2}   }{ (|k|-|w|)^{2}+1}\leq |w|^{2}\frac{ 1\wedge (|k|-|w|)^{-2} }{ (|k|-|w|)^{2}+1          }   $$
$$N(w,k)\geq \frac{1}{2}|w|^{2}\frac{1}{\big( (|k|+|w|)^{2}+1\big)^{\frac{3}{2}} },           $$
where we have used the monotonicity of the functions involved, and the second inequality for $|M(w,k)|^{2}$ includes the inequality
\begin{align}\label{Facile}
\big| |a|-(a^{2}+1)^{\frac{1}{2}}\big|\leq 1\wedge |a|^{-1} .
\end{align}
For large $|k|$, $|M(w,k)|^{2}$ has order $|k|^{-4}$  and $N(w,k)$ has order $|k|^{-3}$.  Thus we can find a constant $c^{\prime}$ so that~(\ref{Normalized}) holds for all $(w,k)$ with $|w|<L^{-1}$ for some $L>1$. \vspace{.2cm}

\noindent Case (2):  When $|k|\geq |w|+|w|^{\frac{1}{2}}$, then we can apply~(\ref{CapM}), ~(\ref{CapN}), and the same estimates as above to show that the asymptotics for $|M(w,k)|^{2}$ has lower order than $|v|\,N(v,k)$. \vspace{.2cm}

\noindent Case (3): It is convenient to follow the pattern in~(\ref{VZeroCase}) and write $M(w,k)$ as $M=M_{1}+M_{2}$, where
\begin{eqnarray*}
M_{1}(w,k) &=& 2\,v\,s(k)+|k+w|-|k-w|\\  M_{2}(w,k)& =& \big( |k+w|^{2}+1 \big)^{\frac{1}{2}}-|k+w|-\big( |k-w|^{2}+1 \big)^{\frac{1}{2}}+|k-w| .
\end{eqnarray*}
Notice that $M_{1}(w,k)$ will only be non-zero when $|k|\leq |w|$,
and that it can be written as
$M_{1}(w,k)=2S(k)S(w)\big(|w|-|k|\big)\chi(|w|\geq |k|)$. Applying
the inequality $|x+y|^{2}\leq 2\,x^{2}+2\,y^{2}$:
\begin{align}
|M(w,k)|^{2}\leq  2\,M_{1}^{2}(w,k) + 2\,M_{2}^{2}(w,k).
\end{align}

We start by bounding $M_{1}$, since it will be used for the
$M_{2}$ term. For $|w|\gg 1$ we claim that
\begin{align}\label{Key}
\big(|w|-|k|\big)\chi(|w|\geq |k|)< N(w,k),
\end{align}
which implies a bound  for $M_{1}$, since
$$ M_{1}=4\big(|w|-|k|\big)^{2}\chi(|w|\geq |k|)\leq  4|w|(|w|-|k|)\chi(|w|\geq |k|)\leq 4|w|M_{2}(w,k).$$  The expression
$\big(|w|-|k|\big)\chi(|w|\geq |k|)$ has its maximum at $k=0$ and
decreases linearly to zero at $|k|=|w|$.  At $k=0$,
$\big(|w|-|k|\big)\chi(|w|\geq |k|)=|w|<2|w|\approx N(0,k) $.
However, $N(w,k)$ decreases at a slightly less than linear rate
and $N(w,w)\approx 1$.

For the $M_{2}^{2}(w,k)$ term, we will finally use the case condition that either $|k+w|\leq |w|^{\frac{1}{2}}$ or $|k-w|\leq |w|^{\frac{1}{2}}$.  Without loss of generality, let it be that $|k-w|\leq |w|^{\frac{1}{2}}$.  Then $|k+w|\gg 1$ and so by~(\ref{Facile}) $\big| \big(|k+w|^{2}+1 \big)^{\frac{1}{2}}-|k+w|\,\big|\sim |k+w|^{-1}\sim |w|^{-1} $.  Thus
\begin{multline}\label{Hiccup}
\Big( \big(|k+w|^{2}+1 \big)^{\frac{1}{2}}-|k+w|- \big(|k-w|^{2}+1 \big)^{\frac{1}{2}}+|k-w| \Big)^{2}\\  = \big[(|k-w|^{2}+1)^{\frac{1}{2}}-|k-w|\big]^{2}+\mathit{O}(|w|^{-1})  ,
\end{multline}
which is $\mathit{O}(1)$. On the other hand,
$$|w|\,N(w,k)\approx |w|\,\big[(|k-w|^{2}+1)^{\frac{1}{2}}+|k+w|-|k|+\mathit{O}(|w|^{-1}) ]=\mathit{O}( |w|^{2}),   $$
since $(|k-w|^{2}+1)^{\frac{1}{2}}=\mathit{O}(|w|^{\frac{1}{2}})$. Thus $|w|\,N(w,k)$ has higher order than~(\ref{Hiccup}). \vspace{.2cm}\\

\noindent Case (4):  Notice that
$$ M_{2}(w,k) =\int_{-w}^{w}da \frac{\,k+a-S(k+a)\big( |k+a|^{2}+1 \big)^{\frac{1}{2}} }{\big( |k+a|^{2}+1 \big)^{\frac{1}{2}}}
$$
where we have expanded the term $\big(|k+w|^{2}+1
\big)^{\frac{1}{2}}$ appearing in $M_{2}$ to first order as in
case (1) and we have written $|k+w|-|k-w| =
\int_{-w}^{w}da\,S(k+a)$. Now,
\begin{multline}
\Big|  \int_{-w}^{w}da \frac{\,k+a-S(k+a)\big( |k+a|^{2}+1 \big)^{\frac{1}{2}} }{\big( |k+a|^{2}+1 \big)^{\frac{1}{2}} } \Big|\\ \leq  2\big(|w|-|k|\big)+    \Big|  \int_{-w}^{w}da\, \chi(|k+a|\geq |w|)  \,     \frac{\,k+a-S(k+a)\big( |k+a|^{2}+1 \big)^{\frac{1}{2}} }{\big( |k+a|^{2}+1 \big)^{\frac{1}{2}} } \Big| \\ \leq 2\big(|w|-|k|\big)+  2\int_{|w|^{\frac{1}{2}}}^{|w|}da\,|a|^{-2}= 2\big(|w|-|k|\big)+\mathit{O}(|w|^{-\frac{1}{2}})  ,
\end{multline}
where in the first inequality  we have used that the integrand on
the left is bounded by $1$ over the interval $a\in[k-2|w|,k+
2|w|]$.  The second inequality follows since $|k+a|^{-2}$ bounds
the integrand, and $\int_{|w|^{\frac{1}{2}}}^{|w|}da\,|a|^{-2}$ is
the smallest that $\int_{S}da |a|^{-2}$ can be for a domain $S$
with diameter $2|w|$ and which is bounded away from the origin by
$|w|^{\frac{1}{2}}$.  Finally, since $|w|-|k|\leq |w|$ and
by~\eqref{Key}
\begin{align*}
\Big|  \int_{-w}^{w}da \frac{\,k+a-S(k+a)\big( |k+a|^{2}+1 \big)^{\frac{1}{2}} }{\big(  |k+a|^{2}+1 \big)^{\frac{1}{2}} } \Big|\leq \big| 2\big(|w|-|k|\big)+\mathit{O}(|w|^{-\frac{1}{2}})  \big|^{2}< 5|w|N(w,k).
\end{align*}
For $k$ such that $|w|-|k|\approx 0$, $\mathit{O}(|w|^{-\frac{1}{2}})$ may become the larger term, but it is still smaller than $N(w,k)$ which will be $\approx 1$ when $|k|$ is near $|w|$.

\end{proof}

We now apply Lemma~\ref{TimeChange} to obtain a similar inequality with $|M_{s}^{\prime}|$ replaced by $E_{s}^{\frac{1}{2} }$.  Our analysis is based on the fact that $E_{t}^{\frac{1}{2}}=\big(\frac{1}{2}  K_{t}^{2}+V(X_{t})\big)^{\frac{1}{2}}$ is a submartingale.  The submartingale property of $E_{t}^{\frac{1}{2}}$ follows since  $f(k)=\big(\frac{1}{2}k^{2}+V\big)^{\frac{1}{2}}$ is a convex function.   Since the jumps occurs with symmetric probabilities ${\cal P}_a(v)={\cal P}_a(-v)$, we can write $E_{t}^{\frac{1}{2}}-E_{0}^{\frac{1}{2}}=\mathbf{M}_{t}+\mathbf{A}_{t}$ in terms of a martingale part $\mathbf{M}_{t}$ and a stochastically increasing part $\mathbf{A}_{t}$ as
\begin{align}\label{Martingale}
\mathbf{M}_{t}=\frac{1}{2}\sum_{n=1}^{\mathcal{N}_{t}}\Big[ \big(\frac{1}{2} |K_{t_{n}^{-}}+w_{n}|^{2}+V(X_{t_{n}})\big)^{\frac{1}{2}}-\big(\frac{1}{2}|K_{t_{n}^{-}}-w_{n}|^{2}+V(X_{t_{n}})\big)^{\frac{1}{2}}\Big],
\end{align}
\begin{multline}
\mathbf{A}_{t}=\frac{1}{2}\sum_{n=1}^{\mathcal{N}_{t}}\Big[\big(\frac{1}{2} |K_{t_{n}^{-}}+w_{n}|^{2}+V(X_{t_{n}})\big)^{\frac{1}{2}}+\big(\frac{1}{2}|K_{t_{n}^{-}}-w_{n}|^{2}+V(X_{t_{n}})\big)^{\frac{1}{2}}\\-2\big(\frac{1}{2} |K_{t_{n}^{-}}|^{2}+V(X_{t_{n}})\big)^{\frac{1}{2}} \Big],
\end{multline}
where  $t_{n},w_{n}$ for $n=1,\dots, \mathcal{N}_{t}$ are the
Poisson times and their corresponding momentum jumps (when they
occur), respectively.  The processes $\mathbf{M}_{t}$, $\mathbf{A}_{t}$ form a Doob-Meyer decomposition for $E_{t}^{\frac{1}{2}}-E_{0}^{\frac{1}{2}}$, however $\mathbf{A}_{t}$ is not a predictable process so the decomposition is not in the unique sense.

\begin{lemma}[Energy Lemma]\label{EIL}
Assume List~\ref{AssumpOne}.  Define $T_{\epsilon,\,V}^{(t)}=\int_{0}^{1}ds\, \chi(|t^{-\frac{1}{2}}E_{st}^{\frac{1}{2}}|>\epsilon )$.  There exists a constant $C$ such that
$$ \liminf_{t\rightarrow \infty}\textup{Pr}\big[ T_{\epsilon,\,V}^{(t)}\geq 1- \delta \big]\geq 1-C\frac{r_{2}^{\frac{1}{2}} }{r_{1}}\frac{\epsilon}{\delta}.   $$

\end{lemma}

\begin{proof}

Define the martingale $M_{t}^{\prime}$ and the increasing process
$A_{t}^{\prime}$ as
 $$M_{t}^{\prime}= \sum_{n=1}^{\mathcal{N}_{t}}w_{n}S(K_{t_{n}^{-}})\text{ and } A_{t}^{\prime}=\sup_{0\leq s\leq t}-M_{s}^{\prime},$$
where $S(K_{t^{-}})$ is the left limit up to time $t$ for the sign of the momentum.
Also define $G_{s}^{(t)}$ as the difference $G_{s}^{(t)}=\sqrt{2}t^{-\frac{1}{2}}E_{st}-t^{-\frac{1}{2}}M_{st}^{\prime}-t^{-\frac{1}{2}}A^{\prime}_{st}$.  In general, we have that
\begin{align}\label{Kmart}
   \textup{Pr}\big[T_{\epsilon}^{(t)}>1-\delta,\, \sup_{0\leq s\leq 1} G_{s}^{(t)}\leq(\sqrt{2}-1)\epsilon  \big]\leq      \textup{Pr}\big[T_{\epsilon,V}^{(t)}>1-\delta \big],
\end{align}
where $T_{\epsilon}^{(t)}$ is defined as in Lemma~\ref{TimeChange} for the process $t^{-\frac{1}{2}}M_{st}^{\prime}+t^{-\frac{1}{2}}A_{st}^{\prime}$.  We will prove below that
$\Pr\big[\sup_{0\leq s\leq 1} G_{s}^{(t)}>(\sqrt{2}-1)\epsilon \big]=\mathit{O}(t^{-\frac{1}{4}}) $.  In that case, by applying the inclusion-exclusion principle to the left side of~(\ref{Kmart}), then~(\ref{Kmart}) can be written
$$
   \textup{Pr}\big[T_{\epsilon}^{(t)}>1-\delta  \big]+\mathit{O}(t^{-\frac{1}{4}}) \leq      \textup{Pr}\big[T_{\epsilon,\,V}^{(t)}>1-\delta \big].
$$
 We can then apply Lemma~\ref{TimeChange} to the left-side to  complete the proof.

Now we work towards establishing $\Pr\big[\sup_{0\leq s\leq 1} G_{s}^{(t)}>(\sqrt{2}-1)\epsilon \big]=\mathit{O}(t^{-\frac{1}{4}}) $.  Consider the martingale $M_{t}^{\prime}-\mathbf{M}_{t}$.  The square of the jumps of $M_{t}^{\prime}-\mathbf{M}_{t}$ can be bounded by the jumps of $\mathbf{A}_{t}$ plus an extra term through the inequality from Lemma~\ref{StickyCases}:
\begin{multline}\label{QuadVar}
 \Big| \sqrt{2}\,w\,S(k)-\Big[\big(\frac{1}{2} |k+w|^{2}+V(x) \big)^{\frac{1}{2}}-\big(\frac{1}{2}|k-w|^{2}+V(x)\big)^{\frac{1}{2}}  \Big]  \Big|^{2} \\ \leq 2\,v^{2}1_{|w|>J}+ c\,J\,\Big[ \big(\frac{1}{2} |k+w|^{2}+V(x)\big)^{\frac{1}{2}}+\big(\frac{1}{2}|k-w|^{2}+V(x) \big)^{\frac{1}{2}}-2\big(\frac{1}{2} k^{2}+V(x)\big)^{\frac{1}{2}}  \Big].
 \end{multline}
for all $J>1$, $k,\,w,\,x\in \R$.  Define the process $Q_{r}=\sum_{n=1}^{\mathcal{N}_{r}}w_{n}^{2}\chi(|w_{n}|\geq t^{\frac{1}{4}})$.

By~(\ref{QuadVar}), the quadratic variation process $[ M^{\prime}-\mathbf{M}]_{t}$ has the bound
$$[ M^{\prime}-\mathbf{M}]_{t} \leq c\,t^{\frac{1}{4}} \mathbf{A}_{t}+2Q_{r}.$$
This will allow us to bound $\mathbb{E}\big[
(M_{t}^{\prime}-\mathbf{M}_{t})^{2}\big]=
\mathbb{E}\big[[ M^{\prime}-\mathbf{M}]_{t}
\big]$. By the fact that $\mathbf{A}_{t}$
is the increasing part of the  Doob-Meyer decomposition for
$E_{t}^{\frac{1}{2}}-E_{0}^{\frac{1}{2}}$ and Jensen's inequality we have
$$\mathbb{E}[\mathbf{A}_{t}]=\mathbb{E}[E_{t}^{\frac{1}{2}}-E_{0}^{\frac{1}{2}}]\leq \,\big(\mathbb{E}[E_{t}]\big)^{\frac{1}{2}}\leq \,\big(\mathbb{E}[E_{0}]+2^{-1}r_{2}t\big)^{\frac{1}{2}}\sim \,2^{-\frac{1}{2}} r_{2}^{\frac{1}{2}}\,t^{\frac{1}{2}}.
$$
Also we have
$$\mathbb{E}\big[Q_{t}    \big]=\mathbb{E}\big[ \sum_{n=1}^{\mathcal{N}_{t}}w_{n}^{2}\,\chi(|w_{n}|\geq t^{\frac{1}{4}})\big]\leq \mathbb{E}[\mathcal{N}_{t}]\,\rho \,t^{-\frac{1}{2}}=\mathcal{R}\,\rho t^{\frac{1}{2}}, $$
where we have used that the fourth moments of the jumps are
bounded as $\mathbb{E}\big[w_{n}^{4}]\leq \rho$,
 since
$$  \mathbb{E}\big[w_{n}^{2}\,\chi(|w_{n}|\geq t^{\frac{1}{4}})]\leq t^{-\frac{1}{2}}\,\mathbb{E}\big[w_{n}^{4}\,\chi(|w_{n}|\geq t^{\frac{1}{4}})]\leq t^{-\frac{1}{2}}\rho. $$
Thus, using the above and  Doob's maximal inequality
\begin{multline}\label{MGUB}
  \mathbb{E}\big[  \sup_{0\leq s\leq 1} t^{-1}|M_{st}^{\prime}-\mathbf{M}_{st}|^{2}    \big]\leq 4t^{-1}\mathbb{E}\big[(M_{t}^{\prime}-\mathbf{M}_{t})^{2}     \big]=4\mathbb{E}\big[[ M^{\prime}-\mathbf{M}]_{t} \big]\\ \leq 4\,c\,t^{-\frac{3}{4}}\mathbb{E}\big[\mathbf{A}_{t}\big]+8t^{-1}\mathbb{E}\big[Q_{t} \big]\leq  4\,c\,2^{-\frac{1}{2}} r_{2}^{\frac{1}{2}}\,t^{-\frac{1}{4}}+8\mathcal{R}\rho\, t^{-\frac{1}{2}}=\mathit{O}(t^{-\frac{1}{4}}).
\end{multline}

Next we show that $t^{-\frac{1}{2}}\mathbf{A}_{st}$ is typically bounded from below by
$t^{-\frac{1}{2}}A_{st}^{\prime}$ in the sense
$$\mathbb{E}\big[\sup_{0\leq s\leq 1} \big|(t^{-\frac{1}{2}}A_{st}^{\prime}-t^{-\frac{1}{2}}\mathbf{A}_{st})1_{A_{st}^{\prime}>\mathbf{A}_{st}}\big|^{2}\big]=\mathit{O}(t^{-\frac{1}{4}}).$$  We can write $\sqrt{2}t^{-\frac{1}{2}}E_{st}^{\frac{1}{2}}$ as
\begin{align}\label{Vamp}
\sqrt{2}t^{-\frac{1}{2}}E_{st}^{\frac{1}{2}}=t^{-\frac{1}{2}}M_{st}^{\prime}+\sqrt{2}t^{-\frac{1}{2}}E_{0}^{\frac{1}{2}}+t^{-\frac{1}{2}}\mathbf{M}_{st}-t^{-\frac{1}{2}}M_{st}^{\prime}+t^{-\frac{1}{2}}\mathbf{A}_{st}
 \end{align}
where since the left side is positive for all $a$, we must have
$$-t^{-\frac{1}{2}}M_{at}^{\prime}\,\chi(t^{-\frac{1}{2}}M_{at}^{\prime}\leq 0)\leq \big( \sqrt{2}t^{-\frac{1}{2} }E_{0}^{\frac{1}{2}}+t^{-\frac{1}{2}}\mathbf{M}_{at}-t^{-\frac{1}{2}}M_{at}^{\prime}+t^{-\frac{1}{2}}\mathbf{A}_{at}\big)\,\chi(t^{-\frac{1}{2}}M_{at}^{\prime}\leq 0).  $$
Taking the supremum in $a$ up to $s\leq 1$ of both sides,
$$t^{-\frac{1}{2}}A_{st}^{\prime}=\sup_{0\leq a\leq s}-t^{-\frac{1}{2}}M_{at}^{\prime}\leq \sup_{0\leq a\leq s}\big| \sqrt{2}t^{-\frac{1}{2} }E_{0}^{\frac{1}{2}}+t^{-\frac{1}{2}}\mathbf{M}_{at}-t^{-\frac{1}{2}}M_{at}^{\prime}\big|+t^{-\frac{1}{2}}\mathbf{A}_{st},$$
where we have used that $M_{0}^{\prime}=0$ for the left side and that $\mathbf{A}_{r}$ is increasing for the right side.  Subtracting $t^{-\frac{1}{2}}\mathbf{A}_{st}$ and taking $\mathbb{E}\big[\sup_{0\leq s\leq 1}|\cdot|^{2}\big]^{\frac{1}{2}}$ of both sides
\begin{multline}\label{IPLB}
\mathbb{E}\big[\sup_{0\leq s\leq 1} \big|(t^{-\frac{1}{2}}A_{st}^{\prime}-t^{-\frac{1}{2}}\mathbf{A}_{st})1_{A_{st}^{\prime}>\mathbf{A}_{st}}\big|^{2}\big]^{\frac{1}{2}}\\ \leq \sqrt{2} t^{-\frac{1}{2} }E_{0}^{\frac{1}{2}}+\mathbb{E}\big[ \sup_{0\leq s\leq 1} t^{-1}\big|\mathbf{M}_{st}-M_{st}^{\prime}\big|^{2}\big]^{\frac{1}{2}}=\mathit{O}(t^{-\frac{1}{8}}).
\end{multline}

Observe that $G_{s}^{(t)}< 0$ implies
\begin{align}\label{AbNorm}
\mathbf{A}_{st}-A_{st}^{\prime}\leq \sqrt{2}\,E_{0}^{\frac{1}{2}}-0\wedge (\mathbf{M}_{st}-M_{st}^{\prime}).
\end{align}
Finally, by the triangle inequality,~(\ref{MGUB}),~(\ref{IPLB}),
and~(\ref{AbNorm})
\begin{multline}\label{LowerBnd}
\mathbb{E}\big[ \sup_{0\leq s\leq 1}|G_{s}^{(t)}|^{2}1_{ G_{s}^{(t)}<0 } \big]^{\frac{1}{2}} \\ \leq t^{-\frac{1}{2}}\sqrt{2}\,E_{0}^{\frac{1}{2}}+ \mathbb{E}\big[ \sup_{0\leq s\leq 1}t^{-1}|\mathbf{A}_{st}-A_{st}^{\prime}  |^{2}1_{ G_{s}^{(t)}<0 } \big]^{\frac{1}{2}}+ \mathbb{E}\big[ \sup_{0\leq s\leq 1}t^{-1}|\mathbf{M}_{st}-M_{st}^{\prime}  |^{2} \big]^{\frac{1}{2}}  \\  \leq t^{-\frac{1}{2}}\sqrt{2}\,E_{0}^{\frac{1}{2}}+\mathbb{E}\big[\sup_{0\leq s\leq 1} t^{-\frac{1}{2}}\big|\mathbf{A}_{st}-A_{st}^{\prime}\big|^{2}1_{A_{st}^{\prime}>\mathbf{A}_{st}}\big]^{\frac{1}{2}}+2\mathbb{E}\big[  \sup_{0\leq s\leq 1} t^{-1}|\mathbf{M}_{st}-M_{st}^{\prime}|^{2}    \big]^{\frac{1}{2}}.
\end{multline}
The right-side is $\mathit{O}(t^{-\frac{1}{8}})$, and via Chebyshev's inequality
$\Pr\big[\sup_{0\leq s\leq 1} |G_{s}^{(t)}|>(\sqrt{2}-1)\epsilon \big]=\mathit{O}(t^{-\frac{1}{4}}) $, which is the bound that was claimed.

\end{proof}

In the following lemma, we set $L_{st}= M_{st}^{\prime}+A_{st}^{\prime}$ where $M_{st}^{\prime}$ and $A_{st}^{\prime}$ are defined as in the proof of Lemma~\ref{EIL}.  The Doob-Meyer decompositions in the lemma are not unique, since the increasing parts are not constrained to be predictable.

\begin{lemma}\label{DOOOOB}
Consider the submartingales $E_{t}$ and $L_{t}^{2}$.
\begin{enumerate}
\item  $E_{t}$ admits a Doob-Meyer decomposition as a sum of martingale and increasing parts
$$
\mathcal{M}_{t}=\sum_{n=1}^{\mathcal{N}_{t}}w_{n}K_{t_{n}^{-}}\quad \text{ and }\quad \mathcal{A}_{t}=E_{0}+\frac{1}{2}\sum_{n=1}^{\mathcal{N}_{t}}w_{n}^{2}.
$$
\item  $L_{t}^{2}$ admits a Doob-Meyer decomposition as a sum of martingale and increasing parts
$\bar{M}_{t}=\sum_{n=1}^{\mathcal{N}_{t}}\alpha_{n}$ and $\bar{A}_{t}=\sum_{n=1}^{\mathcal{N}_{t}}\beta_{n}$ respectively, where
\begin{eqnarray*}
\alpha_{n}&=& 2w_{n}L_{t_{n-1}} \hspace{5.8cm} \text{ for }\quad |L_{t_{n-1}}|\geq |w_{n}|,  \\ \alpha_{n}&=& \frac{1}{2}S(K_{t_{n}^{-} })S(w_{n})(L_{t_{n-1}}+|w_{n}|)^{2} \hspace{2cm} \text{ for }\quad |L_{t_{n-1}}|\leq |w_{n}|, \\
\beta_{n}&=& w_{n}^{2} \hspace{6.8cm} \text{ for }\quad |L_{t_{n-1}}|\geq |w_{n}|,  \\
\beta_{n}&=& w_{n}^{2}+\frac{1}{2}(|w_{n}|-L_{t_{n-1}})^{2}\hspace{3.4cm}  \text{ for }\quad |L_{t_{n-1}}|\leq |w_{n}|.
\end{eqnarray*}

\item In the limit $t\rightarrow \infty$,
$$\mathbb{E}\big[\sup_{0\leq s\leq 1}t^{-1}\big|2\mathcal{A}_{st}-\bar{A}_{st}    \big|   \big]=\mathit{O}(t^{-\frac{1}{4}}).     $$

\end{enumerate}

\end{lemma}

\begin{proof} \text{ }\\
\noindent Part (1):   $E_{t}$ can be rewritten as
\begin{align}\label{Dizzle}
 E_{t}=E_{0}+\sum_{n=1}^{\mathcal{N}_{t}}w_{n}K_{t_{n}^{-}}+\frac{1}{2}w_{n}^{2},
\end{align}
which follows by an inductive expansion using
the conservation of energy between momentum jumps:
\begin{eqnarray*}
E_{t_{n}}=\frac{1}{2}(K_{t_{n}^{-}}+w_{n})^{2}+V(X_{t_{n}^{-}})&=& \frac{1}{2}K_{t_{n}^{-}}^{2}+V(X_{t_{n}^{-}})+w_{n}K_{t_{n}^{-}}+\frac{1}{2}w_{n}^{2}\\ &=& \frac{1}{2}K_{t_{n-1}}^{2}+V(X_{t_{n-1}})+w_{n}K_{t_{n}^{-}}+\frac{1}{2}w_{n}^{2}.
\end{eqnarray*}
The middle term on the right of~(\ref{Dizzle}) is a martingale by the symmetry of the jump rates (\textit{III}) of List~\ref{AssumpOne}:
$E\big[w_{n}\,\big|\,\mathcal{F}_{t_{n}^{-}}\big]=0$.  \vspace{.2cm}\\

\noindent Part (2):  To find an expression for $L_{t}$, we apply an inductive argument as with $E_{t} $ except that the analysis breaks into two cases.  Expanding $L_{t_{n}}^{2}$ when  $|L_{t_{n-1}}|> |w_{n}|$ is the easy case since
$L_{t_{n}}^{2}=L_{t_{n-1}}^{2}+2v_{n}S(K_{t_{n}^{-}})+w_{n}^{2}$. Again $w_{n}S(K_{t_{n}^{-}}) $ is the martingale contribution, since $E\big[w_{n}^{2}\,\big|\,\mathcal{F}_{t_{n}^{-}}\big]=0$.

Expanding $L_{t_{n}}^{2}$ in the case that $|L_{t_{n-1}}|\leq|w_{n}|$ , then
$L_{t_{n}}=\frac{1}{2} \big(1+S(K_{t_{n}^{-} })S(w_{n})\big)(L_{t_{n-1}}+|w_{n}|)$, and
\begin{align*}
L_{t_{n}}^{2}-L_{t_{n-1}}^{2}= \frac{1}{2}S(K_{t_{n}^{-} })S(w_{n})(L_{t_{n-1}}+|w_{n}|)^{2}+w_{n}^{2}+\frac{1}{2}(|w_{n}|-L_{t_{n-1}})^{2}.
\end{align*}
The first term on the right has mean zero since $\mathbb{E}\big[S(w_{n})\,\big|\,\mathcal{F}_{t_{n}^{-}},\,|w_{n}|\big]=0$.\vspace{.25cm}\\

\noindent Part (3):  By Part (1) and Part (2),
$$t^{-1}(\bar{A}_{st}-2\mathcal{A}_{st})=-2t^{-1}E_{0}+t^{-1}\sum_{n=1}^{\mathcal{N}_{st}}\frac{1}{2}\big(|w_{n}|-L_{t_{n-1}}\big)^{2}\chi( |w_{n}|\geq L_{t_{n-1}} )     $$
Since $|w_{n}|\leq J$, the sum above is bounded by
\begin{multline}\label{Lampeke}
\sum_{n=1}^{\mathcal{N}_{st}}(|w_{n}|-L_{t_{n-1}})^{2} \chi(|L_{t_{n-1}}|\leq |w_{n}|)\\ \leq \,t^{\frac{1}{4}} \sum_{n=1}^{\mathcal{N}_{st}}(|w_{n}|-L_{t_{n-1}}) \chi(|L_{t_{n-1}}|\leq |w_{n}|)+\sum_{n=1}^{\mathcal{N}_{st}}|w_{n}|^{2}\chi(|w_{n}|\geq t^{\frac{1}{4}}).
\end{multline}
By the estimates in Lemma~\ref{EIL}  $$\mathbb{E}\Big[\sup_{0\leq s\leq 1}t^{-1}\sum_{n=1}^{\mathcal{N}_{st}}|w_{n}|^{2}\chi(|w_{n}|\geq t^{\frac{1}{4}}) \Big]=\mathbb{E}\Big[t^{-1}\sum_{n=1}^{\mathcal{N}_{t}}|w_{n}|^{2}\chi(|w_{n}|\geq t^{\frac{1}{4}}) \Big]   =\mathit{O}(t^{-\frac{1}{2}}).$$

The first term on the right-side is closely related to $A_{r}^{\prime}$, since $A_{r}^{\prime}$ can be written
$$\sup_{0\leq s\leq r}-\sum_{n=1}^{\mathcal{N}_{r} }w_{n}S(K_{t_{n}^{-}})   = A_{r}^{\prime}=\sum_{n=1}^{\mathcal{N}_{r}}(|w_{n}|-L_{t_{n-1}} ) \chi\big(|L_{t_{n-1}}|\leq |w_{n}|,\, S(K_{t_{n}^{-}})S(w_{n})=-1\big),    $$
since increases in $A_{t}^{\prime}$ occur $|L_{t_{n-1}}|\leq
|w_{n}|$ and the jump $w_{n}$ has sign such that
$S(K_{t_{n}^{-}})S(w_{n})=-1$. In fact, the conditional
expectation of a single term from the sum with respect to the
information up to time $t_{n}^{-}$ and the size $|w_{n}|$ of the
$n$th jump is
\begin{multline*}
\mathbb{E}\big[ (|w_{n}|-L_{t_{n-1}} ) \chi\big(|L_{t_{n-1}}|\leq |w_{n}|,\, S(K_{t_{n}^{-}})S(w_{n}) =-1\big)\,\big|\,\mathcal{F}_{t_{n}^{-}},|w_{n}| \big]\\ =\frac{1}{2}  (|w_{n}|-L_{t_{n-1}}) \chi(|L_{t_{n-1}}|\leq |w_{n}|),
\end{multline*}
since $w_{n}=-|w_{n}|$ and $w_{n}=|w_{n}|$ have equal probability.
Thus,
\begin{multline}
\mathbb{E}\big[t^{-\frac{3}{4}}\sum_{n=1}^{\mathcal{N}_{t}}(|w_{n}|-L_{t_{n-1}}) \chi(|L_{t_{n-1}}|\leq |w_{n}|)     \big]=t^{-\frac{3}{4}}\mathbb{E}\big[ A^{\prime}_{t} \big] = t^{-\frac{3}{4}}\mathbb{E}\big[ M^{\prime}_{t}+ A^{\prime}_{t} \big]\\ =
\mathbb{E}\big[2^{\frac{1}{2}}t^{-\frac{3}{4}}E^{\frac{1}{2}}_{t}-t^{-\frac{1}{4}}G_{1}^{(t)}  \big]
 \leq 2^{\frac{1}{2}}t^{-\frac{3}{4} }\mathbb{E}\big[E^{\frac{1}{2}}_{t}\big]+\mathit{O}(t^{-\frac{3}{8}}) \\ \leq 2^{\frac{1}{2}}t^{-\frac{3}{4} }\mathbb{E}\big[E_{t}\big]^{\frac{1}{2}}+\mathit{O}(t^{-\frac{3}{8}})\leq 2^{\frac{1}{2}}t^{-\frac{3}{4}}(\mathbb{E}[E_{0}]+ \frac{1}{2} r_{2}t)^{\frac{1}{2}} +\mathit{O}(t^{-\frac{3}{8}})=\mathit{O}(t^{-\frac{1}{4}}),
\end{multline}
where the second equality is because $M_{r}^{\prime}$ is a mean zero martingale, the third equality is from the definition of $G_{r}^{(t)}$, and the first inequality uses the result $\mathbb{E}\big[ \sup_{0\leq s\leq 1}|G_{s}^{(t)}|^{2}1_{ G_{s}^{(t)}<0 } \big]^{\frac{1}{2}}=\mathit{O}(t^{-\frac{1}{8}})$ from the proof of Lemma~\ref{EIL}.

\end{proof}

\section{The absolute value of the momentum}\label{SecReflection}

The following theorem takes us as far possible towards obtaining a
central limit theorem for the momentum $t^{-\frac{1}{2}}K_{st}$
without making an assumption about a reflection symmetry in the
periodic potential and in the jump probabilities.

\begin{theorem}\label{AbsoluteValue}
Assume List~\ref{AssumpOne} and (\textit{i})-(\textit{ii}) of List~\ref{AssumpTwo}. In the limit $t\rightarrow \infty$,  $t^{-\frac{1}{2}}|K_{st}|$  converges in distribution to the absolute value of a Brownian motion with diffusion constant $\sigma=\int_{\R}dv\,\tilde{P}(v)\,v^{2}$.
\end{theorem}

The above theorem is not enough to guarantee that $t^{-\frac{1}{2}}K_{st}$ converges to a Brownian motion.  To see this, consider a  random walk $X_{n}$ on $\Z$ which jumps to the right and the left with equal probability at every lattice site except at the origin where it jumps to $1$ with probability $\frac{2}{3}$ and to $-1$ otherwise. In this case, $|X_{n}|$ has the same distribution as a simple random walk and thus $n^{-\frac{1}{2}}|X_{\lfloor s\,n \rfloor}|$ converges in distribution to the absolute value of a Brownian motion $|\mathbf{B}_{s}|$.  However, removing the absolute values, then  $n^{-\frac{1}{2}}X_{\lfloor s\,n \rfloor}$ will have a drift determined by the process $\frac{1}{3}n^{-\frac{1}{2}}\sum_{n=1}^{\lfloor s\,n \rfloor}\chi(X_{n}=0)$.  Since a simple random walk $X_{n}$ spends on the order of $n^{\frac{1}{2}}$ steps at the origin, the drift will be non-vanishing.

The proof of Theorem \ref{AbsoluteValue} follows in Section \ref{SecMain} and is based on the Lemma~\ref{StochEqn} below.  The latter extends the analysis in the proof of Lemma~\ref{EIL} to show that $t^{-\frac{1}{2}}K_{s\,t} $ is close to being the solution of a stochastic equation, reminiscent of the fact for a Brownian motion $\mathbf{B}_{s}$ that $|\mathbf{B}_{s}|=\mathbf{B}_{s}^{\prime}+\sup_{0\leq r\leq s} -\mathbf{B}_{s}^{\prime}  $ for another Brownian motion $\mathbf{B}_{s}^{\prime}=\int_{0}^{s}S(\mathbf{B}_{r})d\mathbf{B}_{r}$.  This is related to Levy's theory of Brown local time~\cite{Karat}.

 For some bounded right-continuous adapted process $F_{t}$, we denote
\begin{align}\label{DefStochInt}
\int_{0}^{s}F_{r\,t}\,dM^{(t)}_{r}=t^{-\frac{1}{2}}\sum_{n=1}^{\mathcal{N}_{st}}w_{n}\,F_{t_{n}^{-}},\end{align}
where $v_{n}$ are the momentum jumps.  Thus the martingale
$M^{\prime}_{t}$ used in the proof of Lemma~\ref{EIL} can be
written
$t^{-\frac{1}{2}}M^{\prime}_{st}=\int_{0}^{s}S(K_{rt})\,dM^{(t)}_{r}
$.

\begin{lemma} \label{StochEqn}
Assume List~\ref{AssumpOne}. Then the momentum process $K_{st}$,
$s\in[0,1]$, satisfies the stochastic equation
\begin{align}\label{ReVamp}
t^{-\frac{1}{2}}|K_{st}| =\int_{0}^{s}S(K_{r\,t})\,dM^{(t)}_{r}+\sup_{0\leq a\leq s}-\int_{0}^{a}S(K_{rt})\,dM^{(t)}_{r}+\mathcal{E}_{s}^{(t)} ,
\end{align}
where the error $\mathcal{E}_{s}^{(t)}$ has $\mathbb{E}\big[\sup_{0\leq s\leq 1}\big|\mathcal{E}_{s}^{(t)}\big|^{2}\big]\rightarrow 0$ for large $t$.

\end{lemma}

\begin{proof}
The error $\mathcal{E}_{s}^{(t)}$ is close to error  $G_{s}^{(t)}=\sqrt{2}\,t^{-\frac{1}{2}}E_{st}^{\frac{1}{2}}-t^{-\frac{1}{2}}M_{st}^{\prime}-t^{-\frac{1}{2}}A_{st}^{\prime}$ which arose in the proof of Lemma~\ref{EIL}, since they differ only by the replacement of the momentum $t^{-\frac{1}{2}}|K_{st}|$ with $\sqrt{2}t^{\frac{1}{2}}E_{st}^{\frac{1}{2}}$ and $ \big|\sqrt{2} t^{-\frac{1}{2}}E_{s\,t}^{\frac{1}{2}}-t^{-\frac{1}{2}}|K_{st}|\big|\leq t^{-\frac{1}{2}}\bar{V}$.  In the proof of Lemma~\ref{EIL}, it was shown that
\begin{align}\label{BigGInq}
t^{-1}\mathbb{E}\big[ \sup_{0\leq s \leq 1}\big|G_{s}^{(t)}\big|^{2} 1_{G_{s}^{(t)}<0} \big]=\mathit{O}(t^{-\frac{1}{4}}).
\end{align}
We now work on showing that  $t^{-1}\mathbb{E}\big[ \sup_{0\leq s \leq 1}\big|G_{s}^{(t)}\big|^{2} 1_{G_{s}^{(t)}\geq 0} \big]$ tends to zero also.

By Part (3) of Lemma~\ref{DOOOOB}, the normalized difference in the increasing parts $\bar{A}_{st}$ and $\mathcal{A}_{st}$ for the Doob-Meyer decompositions of $L_{st}^{2}$ and $ E_{st}$ respectively, tends to zero as
\begin{align}\label{Observe}
 \mathbb{E}\big[\sup_{0\leq s\leq 1} t^{-1}\big| 2\mathcal{A}_{st}-\bar{A}_{st}\big|\big]=\mathit{O}(t^{-\frac{1}{4}}).
\end{align}
 We will make use of the error $G_{s}^{(t)}$ being the difference between two positive submartingales having a vanishing difference between the increasing parts of their Doob-Meyer decompositions:
\begin{multline*}
H_{s}^{(t)}= 2t^{-1}E_{st}-t^{-1}|M_{st}^{\prime}+A_{st}^{\prime}|^{2}       = 2t^{-1}E_{st}-(\sqrt{2}\,t^{-\frac{1}{2}}E_{st}^{\frac{1}{2}}-G_{s}^{(t)})^{2}  \\=G_{s}^{(t)}(2^{\frac{3}{2}}\,t^{-\frac{1}{2}}E_{st}^{\frac{1}{2}}-G_{s}^{(t)} )=\Upsilon_{s}^{(t)}+\Gamma_{s}^{(t)},
\end{multline*}
where $\Upsilon_{s}^{(t)}$ is a martingale and
$\Gamma_{s}^{(t)}=t^{-1}(2\mathcal{A}_{st}-\bar{A}_{st})  $ is
$\mathit{O}(t^{-\frac{1}{4}})$  by~(\ref{Observe}).

 Note that $2^{\frac{3}{2}}\,t^{-\frac{1}{2}}E_{st}^{\frac{1}{2}}-G_{s}^{(t)}$ is positive, so $G_{s}^{(t)}$ determines the sign of $H_{s}^{(t)}$.  By the Cauchy-Schwarz inequality
\begin{align}\label{HNeg}
 t^{-1}\mathbb{E}\big[ \sup_{0\leq s \leq 1}|H_{s}^{(t)}|\,1_{H_{s}^{(t)} <0} \big]\leq \big( \mathbb{E}\big[ \sup_{0\leq s \leq 1}\big|G_{s}^{(t)}\big|^{2}1_{G_{s}^{(t)}<0} \big]\big)^{\frac{1}{2}}\big(\mathbb{E}\big[ \sup_{0\leq s \leq 1}\big|2^{\frac{3}{2}}t^{-\frac{1}{2}}E_{st}^{\frac{1}{2}}-G_{s}^{(t)}\big|^{2}  \big]\big)^{\frac{1}{2}}.
\end{align}
The first factor on the right is $\mathit{O}(t^{-\frac{1}{8}})$
by~(\ref{BigGInq}). Bounding the second factor on the right comes
through the triangle inequality and the use of Doob's maximal
inequality for the two positive submartingales $M_{st}+A_{st}$ and
$E_{st}^{\frac{1}{2}}$:
\begin{align}\label{BigHBnd}
\big(\mathbb{E}\big[ \sup_{0\leq s \leq 1}\big|2^{\frac{3}{2}}t^{-\frac{1}{2}}E_{st}^{\frac{1}{2}}-G_{s}^{(t)}\big|^{2}  \big]\big)^{\frac{1}{2}} \leq  4\sqrt{2}t^{-\frac{1}{2}}\mathbb{E}\big[  E_{st} \big]^{\frac{1}{2}}+4t^{-\frac{1}{2}}\mathbb{E}\big[ | M_{st}^{\prime}+A_{st}^{\prime}|^{2} \big]^{\frac{1}{2}}\leq 12\,r_{2}^{\frac{1}{2}},
\end{align}
where
$2^{\frac{3}{2}}t^{-\frac{1}{2}}E_{st}^{\frac{1}{2}}-G_{s}^{(t)}=2^{\frac{1}{2}}t^{-\frac{1}{2}}E_{st}^{\frac{1}{2}}+t^{-\frac{1}{2}}M_{st}^{\prime}+t^{-\frac{1}{2}}A_{st}^{\prime}$.
The last inequality follows by the forms for the increasing parts
of the Doob-Meyer decompositions for for
$t^{-1}|M_{st}^{\prime}+A_{st}^{\prime}|^{2}$ and
$t^{-1}E_{st}$ from Lemma~\ref{DOOOOB} which are
both bounded by $t^{-1}[ M]_{st}$  .  The
right side of~(\ref{HNeg}) vanishes as
$\mathit{O}(t^{-\frac{1}{8}})$, and  thus the values of
$t^{-1}H_{s}^{(t)}$ do not typically go far in the negative
direction.  Since $H_{s}^{(t)}$ is a martingale with mean
$2t^{-1}E_{0}\rightarrow 0$, it would be expected that
$t^{-1}H_{s}^{(t)}$ can also not go far in the positive direction.

We will argue below that $\lim_{t\rightarrow \infty}\mathbb{E}\big[ \sup_{0\leq s \leq 1}|H_{s}^{(t)}\,|\,1_{H_{s}^{(t)}>0} \big]\rightarrow 0$.   This would complete the proof since
$$t^{-1}\mathbb{E}\big[ \sup_{0\leq s \leq 1}\big|G_{s}^{(t)}\big|^{2}1_{G_{st}>0} \big]\leq t^{-1}\mathbb{E}\big[ \sup_{0\leq s \leq 1}H_{s}^{(t)}1_{H_{s}^{(t)}>0} \big], $$
which clearly follows since $2^{\frac{3}{2}}t^{-\frac{1}{2}}E_{st}^{\frac{1}{2}}-G_{s}^{(t)}$ is larger in absolute value than $G_{s}^{(t)}$ when $G_{s}^{(t)}$ is positive.

 By the optional sampling theorem,
\begin{align}\label{OST}
\mathbb{E}[\Upsilon_{\tau}^{(t)}]=\mathbb{E}[ H_{\tau}^{(t)}-\Gamma_{\tau}^{(t)}]=0
\end{align}
for any adapted stopping time $\tau$.  Given some $a>0$, let $\tau\in[0,\,1]$ be the first time that $H_{s}^{(t)}$ reaches above $a$ and put $\tau=1$ if that event does not occur.  By~(\ref{OST}),
$$\mathbb{E}\big[\Gamma_{\tau}^{(t)}\big]= a\Pr\big[ \sup_{0\leq s\leq 1}H_{s}^{(t)}\geq a\big]+\mathbb{E}\big[ H_{1}^{(t)}\,\chi\big(\sup_{0\leq s\leq 1}H_{s}^{(t)}< a \big)\big],$$
which implies that
\begin{align}\label{Tibet}
\sup_{a\in \R^{+}} a\Pr\big[ \sup_{0\leq s\leq 1}H_{s}^{(t)}\geq a\big]\leq \big|\mathbb{E}\big[\Gamma_{\tau}^{(t)}\big]\big|+ \mathbb{E}\big[ \sup_{0\leq s \leq 1}|H_{s}^{(t)}|1_{H_{s}^{(t)}<0} \big]=\mathit{O}(t^{-\frac{1}{4}}).
\end{align}
Set $Y_{t}= \sup_{0\leq s \leq 1}H_{s}^{(t)}1_{H_{s}^{(t)}>0}$ and $\alpha_{t}= \big(\sup_{a\in \R^{+}} a \Pr[  Y_{t}\geq a]\big)^{\frac{1}{2}}   $.   If $\mathbb{E}\big[ Y_{t}^{2} \big]\leq \mathbf{y}^{2}$ for all $t$, then for any $m>\alpha_{t}>0$
\begin{multline*}
 \mathbb{E}\big[Y_{t}]=\int_{0}^{\infty}da\,\Pr\big[Y_{t}\geq a\big]= \big(\int_{0}^{\alpha_{t} }+\int_{\alpha_{t}}^{m}+\int_{m}^{\infty}\big)da \,\Pr\big[Y_{t}\geq a\big] \\ \leq \alpha_{t}+\frac{1}{\alpha_{t}}\int_{\alpha_{t}}^{m}da \,a\,\Pr\big[Y_{t}\geq a\big]+\frac{1}{m}\int_{m}^{\infty}da \,a\,\Pr\big[Y_{t}\geq a\big]  \leq \alpha_{t}+m \alpha_{t}+\frac{\mathbf{y}^{2}}{2m}.
 \end{multline*}
The second inequality above uses that $\alpha_{t}^{2}$ is the supremum of $a  \Pr[  Y_{t}\geq a]$ for the first integral and the relation $\frac{1}{2}\mathbb{E}\big[ Y_{t}^{2} \big]=\int_{0}^{\infty}da\,a\,\Pr\big[Y_{t}\geq a\big]$ for the second.  We can pick $m$ large to make the last term on the right-side small, and then pick $t$ large enough so that with~(\ref{Tibet}) $\alpha_{t}(1+m)$ is small.

To finish the argument we just need to show that $\mathbb{E}\big[ Y_{t}^{2} \big]$ can be uniformly bounded in $t$.  Since $G_{\tau}^{(t)}$ is the difference of $2^{\frac{1}{2}}t^{-\frac{1}{2}}E_{st}^{\frac{1}{2}}$ and $t^{-\frac{1}{2}}L_{st}$, and $2^{\frac{3}{2}}t^{-\frac{1}{2}}E_{s}^{\frac{1}{2}}-G_{s}^{(t)}$ is their sum,
$$\mathbb{E}\big[Y_{t}^{2}\big]^{\frac{1}{2}}<4\,\mathbb{E}\big[\sup_{0\leq s\leq 1}  t^{-2} E_{st}^{2}     \big]^{\frac{1}{2}}+2\,\mathbb{E}\big[\sup_{0\leq s\leq 1} t^{-2}|L_{st}|^{4}\big]^{\frac{1}{2}}.
     $$
To obtain a value $\mathbf{y}$, the terms on the right can be bounded by standard calculations using the Doob-Meyer decompositions for $E_{r}$ and $|L_{r}|^{2}$ from Lemma~\ref{DOOOOB} and using Doob's maximal inequality.

\end{proof}

\section{Estimates at high energy}\label{SecEstimates}

In this section, we provide estimates that are useful for understanding the dynamics when the particle has high energy, which by Lemma~\ref{EIL}, is the majority of time.  The estimates are based on the idea that the particle will feel a spatially averaged noise and that the momentum is too high to be shifted through the action of the force generated by the potential.

The following elementary bound is used many times in this  section
and later sections.  It follows from the conservation of energy
and the quadratic formula.  It basically says that if the initial
momentum $k_{0}$ has $|k_{0}|\gg
\bar{V}^{\frac{1}{2}}=\sup_{a\in[0,1]}V(x)^{\frac{1}{2}}$, then
the future momenta $k_{s}$, as determined by the Hamiltonian
evolution, will stay close to $k_{0}$.

\begin{lemma} \label{RealBasics}
Let $(x_{t},\,k_{t})$ evolve according to the Hamiltonian $H(x,\,k)=\frac{1}{2}k^{2}+V(x)$, for positive potential bounded by $\bar{V}$.  If the initial momentum has $|k_{0}|^{2}>4\bar{V}$, then the displacements in momentum $k_{t}-k_{0}$ and $k_{t}-k_{s}$ have bounds
$$|k_{t}-k_{0}|< \frac{2\bar{V}}{|k_{0}|} \text{ and } |k_{t}-k_{s}|< \frac{4\bar{V}}{|k_{0}|}\quad \text{ for all $t,\,s\in \R$.} $$

\end{lemma}

\begin{proof}
Since $|k_{0}|^{2}>4\bar{V}$, the momentum $k_{t}$ will not change signs at any time. By the conservation of energy
$$\frac{1}{2}\big|k_{0}+(k_{t}-k_{0})\big|^{2}-\frac{1}{2}k_{0}^{2}=-V(x_{t})+V(x_{0}).$$
Using the quadratic formula and that $k_{t},\,k_{0}$ have the same
sign,
$$|k_{t}-k_{0}|=\big||k_{0}|-\big(k_{0}^{2}+2V(x_{0})-2V(x_{t})\big)^{\frac{1}{2}}\big|\leq \Big| \frac{1}{2}\int_{0}^{2V(x_{0})-2V(x_{t})}da\,\big(k_{0}^{2}+a\big)^{-\frac{1}{2}}\Big|<\frac{2\bar{V}}{|k_{0}|}   $$
Since $\big(\frac{1}{2}k^{2}_{0}+a\big)^{-\frac{1}{2}}\leq \sqrt{2}|k_{0}|^{-1}<2|k_{0}|^{-1}$   for $a\leq \frac{1}{2}\,k_{0}^{2}$.  By the triangle inequality, we can bound the difference $|k_{t}-k_{s}|$.

\end{proof}

As before $\tilde{P}(v)=\int_{0}^{1}da\,\frac{\kappa(a)}{\bar{\kappa}}\,\mathcal{P}_{a}(v)$ where $\bar{\kappa}= \int_{0}^{1}da\,\kappa(a)$. By our assumptions $\inf_{0\leq a\leq 1}\kappa(a)=\nu>0$.

\begin{lemma}\label{SingleJump}
Assume  List~\ref{AssumpOne} and (\textit{i})-(\textit{ii}) of
List~\ref{AssumpTwo}. Starting from the point $(x,\,k)$ with
$|k|^{2}\gg \bar{V}$, let $r_{(x,\,k)}\in L^{1}([0,\,1])$ and
$\tilde{r}_{(x,\,k)}\in L^{1}([0,\,1])$ be the probability density for the
torus position of the particle at the first Poisson time and at the time of the first momentum jump
respectively.   Let $P_{(x,\,k)}\in L^{1}(\R)$ be the density
for the first momentum jump. Let $T_{(x,\,k),t}\in L^{1}([0,\,1])$ be the density of time that a deterministic trajectory starting from $(x,\,k)$ and evolving according to the Hamiltonian $H(x,k)=k^{2}+V(x)$ spends at a torus point over a time interval $[0,t]$.   We have the following bounds
\begin{enumerate}
\item $\sup_{a\in [0,\,1]}\big|r_{(x,\,k)}(a)-1\big| \leq  \frac{2\mathcal{R}}{|k|}+\mathit{O}(\frac{1}{|k|^{2}}),
$
\item
$\sup_{a\in [0,\,1]}\big|\tilde{r}_{(x,k)}(a)\frac{\bar{\kappa}}{\kappa(a)}-1\big|\leq \frac{2\mathcal{R}\nu^{-2} }{|k|}+\mathit{O}(\frac{1}{|k|^{2}})$,
\item
$\sup_{v\in \R}\Big| \frac{P_{(x,\,k)}(v)   }{\tilde{P}(v)  }-1\Big| \leq  \frac{2\mathcal{R}\nu^{-2}}{|k|}+\mathit{O}(\frac{1}{|k|^{2}}). $

\item
$\sup_{a\in [0,\,1]}\big|T_{(x,\,k),t}(a)-t \big|\leq \frac{2t}{|k|} $.

\end{enumerate}

\end{lemma}
\begin{proof}\text{ }\\
Part (1):\vspace{.1cm}

Let $(x_{s},\, k_{s})\in \R^{2}$ be the position and momentum for a particle beginning at $(x,\,k)$ and evolving over a time period $s$ for Hamiltonian $H(x,k)=\frac{1}{2}k^{2}+V(x)$.  Notice that $r_{(x,\,k)}(a)$ can be written as
\begin{align*}
r_{(x,\,k)}(a)=\sum_{n=1}^{\infty}|\mathbf{k}_{(x,k)}(a)|^{-1}\mathcal{R}\,e^{-\mathcal{R}\,r_{n}(a)},
\end{align*}
where $s=r_{1}(a),\, r_{2}(a),\cdots$ are the periodic sequence of times for which $x_{s}\,\textup{mod}(1)=a$, and $\mathbf{k}_{(x,k)}(a)=s(k)\big(H(x,k)-V(a)\big)^{\frac{1}{2}}   $ is their momentum at the point.  These times will exist for every $a\in[0,\,1]$ as long as $H(x,\,k)>\bar{V}$.

 If $4\bar{V}\leq k^{2}$, then  $|k_{s}-k|\leq 2\bar{V}|k|^{-1}$ by Lemma~\ref{RealBasics}. Thus for large momentum $|k|\gg (\bar{V})^{\frac{1}{2}}$, $k_{s}$ is nearly constant and the hit times $r_{n}(a)$ will be close to the sequence of times $s=s_{n}(a)$ at which $ x+sk\,\textup{mod}(1)=a$ for a time period at least on the order of $t^{\frac{1}{2}}$.    The period $ \tau $ such that $r_{n}(a)-r_{n-1}(a)=\tau$ should thus be close to $\frac{1}{|k|}$.  When $|k|$ is large enough so that $|k_{s}-k|\leq 2\bar{V}|k|^{-1}<\frac{1}{2}|k|$, then clearly $\tau\leq \frac{2}{|k|}$, and
\begin{align}\label{TimePeriod}
|\tau-\frac{1}{|k|}|\leq \frac{1}{|k|}\big|\int_{0}^{\tau}ds\, k-\int_{0}^{\frac{1}{|k|}}ds\, k\big|\leq \frac{1}{|k|}\big(\int_{0}^{\tau}ds\,|k_{s}- k|+\int_{0}^{\frac{1}{|k|}}ds\,| k_{s}-k|\big)< \frac{6\bar{V}}{|k|^{3}}.
\end{align}
The difference between the first crossing-times  $|r_{1}(a)-s_{1}(a)|$ of the point $a$ can be similarly bounded.

Using the triangle inequality
\begin{multline}\label{Aqui}
\big|r_{(x,k)}(a)-1  \big|\leq \big|r_{(x,k)}(a)-\frac{1}{|k|}\sum_{n=1}^{\infty}\mathcal{R}\,e^{-\mathcal{R}\,r_{n}(a) }  \big|\\+\big|\frac{1}{|k|}\sum_{n=1}^{\infty}\mathcal{R}\,e^{-\mathcal{R}\,r_{n}(a) }-\frac{1}{|k|}\sum_{n=1}^{\infty}\mathcal{R}\,e^{-\mathcal{R}\,s_{n}(a) }\big|+  \big|\frac{1}{|k|}\sum_{n=1}^{\infty}\mathcal{R}\,e^{-\mathcal{R}\,s_{n}(a) }-1\big|\\ \leq \frac{2\mathcal{R}}{|k|}+\mathit{O}(\frac{1}{|k|^{2}})   ,
\end{multline}
where the last inequality follows by further computation using the inequalities above.  For instance, we can bound the first term on the right as
$$\big|r_{(x,k)}-\frac{1}{|k|}\sum_{n=1}^{\infty}\mathcal{R}\,e^{-\mathcal{R}\,r_{n}(a) }  \big|\leq \big|1-\frac{k}{\mathbf{k}_{(x,k)}(a)} \big|\, e^{-\mathcal{R}\,r_{1}(a)}\frac{1}{|k|\tau}  \frac{\mathcal{R}\,\tau}{1-e^{-\mathcal{R}\,\tau}}\leq  \frac{2\bar{V}}{k^{2}}.   $$

\noindent Part (2):\vspace{.1cm}

We will bound $\sup_{a\in
[0,\,1]}\big|\tilde{r}_{(x,k)}(a)\frac{\bar{\kappa} }{\kappa(a) }-
1\big|$ by invoking Part (1).  This will involve breaking down an
expression for $\tilde{r}_{(x,k)}(a)$. Since there are a random
number of Poisson times before the time of the first momentum
jump, the expression will have a series of integrals whose $n$th
term corresponds to a momentum jump occurs at the $(n+1)$th
Poisson time.
\begin{align}\label{StrangeSum}
\tilde{r}_{(x,\,k)}(a)=\kappa(a)\, \sum_{n=0}^{\infty} \int_{(\R^{+})^{n}}ds_{1}\dots ds_{n}\, \mathcal{R}^{n}\,e^{-\mathcal{R}S(n)   }\,\,\Pi_{m=1}^{n} \big|1-\kappa(x_{S(m)})\big| r_{(x_{S(n)},\,k_{S(n)})}(a),
\end{align}
and $S(m)=s_{1}+\cdots +s_{m}$.

  Let
$|k|>4\bar{V}^{\frac{1}{2}}$ so that with two applications of Lemma~\ref{RealBasics}, $\sup_{s,\,t\geq 0}|k_{t}-k_{s}|\leq |k|^{-1}2\bar{V}$.  In particular, for any time $S(n)$, we can apply Part (1) to the difference $|r_{(x_{S(n)},\,k_{S(n)})}(a)-1|$ to get
\begin{multline}\label{NiHa}
\Big|\tilde{r}_{(x,\,k)}(a)\frac{ \bar{\kappa}  }{ \kappa(a)}-\bar{\kappa} \sum_{n=0}^{\infty} \int_{(\R^{+})^{n}}ds_{1}\dots ds_{n}\, \mathcal{R}^{n}\,e^{-\mathcal{R}S(n)   }\,\,\Pi_{m=1}^{n} \big|1-\kappa(x_{S(m)})\big| \Big| \\ \leq \big(\frac{2\mathcal{R}}{|k|}+\mathit{O}(\frac{1}{|k|^{2}})  \big)\sum_{n=1}^{\infty} \int_{(\R^{+})^{n}}ds_{1}\dots ds_{n}\, \mathcal{R}^{n}\,e^{-\mathcal{R}S(n)   }\,\,\Pi_{m=1}^{n} \big|1-\kappa(x_{S(m)})\big| \\
\leq \Big(\frac{2\mathcal{R}}{|k|}+\mathit{O}(\frac{1}{|k|^{2}})\Big)\sum_{n=1}^{\infty} (1-\nu)^{n}= \nu^{-1}(1-\nu) \Big(\frac{2\mathcal{R}}{|k|}+\mathit{O}(\frac{1}{|k|^{2}})\Big),
\end{multline}
where the second inequality follows since $\big|1-\kappa(x_{S(m)})\big|\leq 1-\nu$ for all $m$ and $\mathcal{R}^{n}\,e^{-\mathcal{R}S(n)} $ defines a probability measure on $(\R^{+})^{n}$.

 If $\int_{(\R^{+})^{n} }ds_{1} \cdots ds_{n} \mathcal{R}^{n}\,e^{-\mathcal{R}S(n)   }\,\Pi_{m=1}^{n} \big| 1-\kappa(x_{S(m)})\big|$ were replaced by $(1-\bar{\kappa })^{n}$ in the left-side of~(\ref{NiHa}), then identity $\sum_{n=0}^{\infty}(1-\bar{\kappa})^{n}=\kappa^{-1}$ would make the difference zero.

  Using a telescoping sum and the definition of $r_{(x,k)}(a)$,
\begin{multline*}
\Big| \int_{(\R^{+})^{n} }ds_{1} \cdots ds_{n} \mathcal{R}^{n}\,e^{-\mathcal{R}S(n)   }\,\Pi_{m=1}^{n} \big| 1-\kappa(x_{S(m)})\big|-(1-\bar{\kappa })^{n}\Big|\\ \leq  \sum_{m=0}^{n} (1-\kappa)^{n-m} \int_{(\R)^{m-1}} ds_{1} \dots  ds_{m-1}  \mathcal{R}^{m-1}\,e^{-\mathcal{R}S(m-1)}\Pi_{r=1}^{m-1} \big|1-\kappa(x_{S(r)})\big|\,  \\ \times \int_{0}^{1}da\,\big|1-\kappa(a)\big|\,\big| r_{(x_{S(m-1)},\,k_{S(m-1)})}(a)-1 \big|
\end{multline*}
Again by Part (1), since $\mathcal{R}^{m-1}e^{-\mathcal{R}S(m-1)}$
is a probability measure on $(\R^{+})^{m-1}$, and by the bounds
$1-\kappa(a),\,1-\bar{\kappa} \leq\, 1-\nu$, we can estimate the
right-side above by
$$n (1-\nu)^{n}\Big(\frac{2\mathcal{R}}{|k|}+\mathit{O}(\frac{1}{|k|^{2}})\Big).      $$

Putting everything together
$$\sup_{a\in [0,\,1]}\big|\tilde{r}_{(x,k)}(a)\frac{\bar{\kappa}}{\kappa(a)}-1 \big|\leq  \Big(\frac{2\mathcal{R}}{|k|}+\mathit{O}(\frac{1}{|k|^{2}})\Big)\,\sum_{n=0}^{\infty}(n+1)(1-\nu)^{n},  $$
and the sum of the series is $\nu^{-2}$.\vspace{.25cm}

\noindent Part (3):\vspace{.1cm}

Now we study the probability density $P_{(x,\,k)}(v)$ for the next momentum jump.   We can write the density for the next momentum jump as
\begin{align}\label{NotAvg}
P_{(x,\,k)}(v)=\int_{0}^{1}da\,\tilde{r}_{(x,\,k)}(a)\,\mathcal{P}_{a}(v),
\end{align}

We then have a bound using
\begin{multline}\label{IG}
\sup_{v\in \R}\big| \frac{P_{(x,\,k)}(v)   }{\tilde{P}(v)  }-1\big|\leq \sup_{v\in \R} \frac{1}{\tilde{P}(v)}\int_{0}^{1}da\,\big|\tilde{r}_{(x,\,k)}(a)\frac{\bar{\kappa}}{ \kappa(a)}-1\big|\frac{\kappa(a)}{\bar{\kappa}} \mathcal{P}_{a}(v) \\  \leq \sup_{a\in [0,\,1]}\big|\tilde{r}_{(x,k)}(a)\frac{\bar{\kappa}}{   \kappa(a)}-1\big|\sup_{|v|\leq J} \frac{1}{\tilde{P}(v)} \int_{0}^{1}da\,\mathcal{P}_{a}(v)\frac{\bar{\kappa}}{   \kappa(a)}= \sup_{a\in [0,\,1]}\big|\tilde{r}_{(x,k)}(a)\frac{\bar{\kappa}}{   \kappa(a)}-1\big|,
\end{multline}
where in the last equality we have used the definition of $\tilde{P}(v)$.  Applying Part (2) then we get the bound.\vspace{.25cm}

\noindent Part (4): \vspace{.1cm}
The density $T_{(x,k),t}$ can be written
$$T_{(x,k),t}(a)= n_{(x,k)}(a,t)k_{(x,k)}(a)   $$
where $\mathbf{k}_{(x,k)}(a)$ is defined as in Part (1), and $n_{(x,k)}(a,t)$ is the number of times the particle passes over the torus point $a$ over a time period $t$ when starting from $(x,k)$.  The result follows by bounding the errors for $\mathbf{k}_{(x,k)}(a)\sim k$ and $|k|^{-1}n_{(x,k)}(a,t)\sim t$ which is similar to Part (1).

\end{proof}

The following lemma bounds the contribution of the  cumulative
drift over periods of high-energy. Define $\tau(t)$ to be the time
of the next to last momentum jump and put it equal to zero if two
jumps have not occurred.

\begin{lemma}\label{HEDrift}
Assume List~\ref{AssumpOne}.   In the limit $t\rightarrow \infty$,
for $0< \beta <\frac{1}{2}$,
$$\mathbb{E}\Big[\sup_{0\leq r\leq t}\Big|t^{-1+2\beta} \int_{0}^{r}ds\,\frac{dV}{dx}(X_{s})\,\chi \big(|K_{\tau(s)}|>t^{\beta} \big)  \Big|^{2} \Big]^{\frac{1}{2}}\longrightarrow 0.    $$

\end{lemma}

\begin{proof}

It is convenient to split the total integral into a sum of integrals over the periods between Poisson times (at which there may be a momentum jump breaking the conservation of energy) which include only the even and odd terms respectively.  Let $t_{1},\dots,t_{\mathcal{N}_{r}}$ be the Poisson times up to a time $r$.
\begin{multline}
\int_{0}^{r}ds\,\frac{dV}{dx}(X_{s})\,\chi\big(|K_{\tau(s)}|>t^{\beta }\big) \approx \sum_{n=1}^{\lfloor \frac{\mathcal{N}_{r} }{2}\rfloor } \chi\big(|K_{t_{2n-1}}|\geq t^{\beta}  \big)\, \int_{t_{2n}}^{t_{2n+1}}ds\,\frac{dV}{dx}(X_{s})\\+\sum_{n=0}^{\lfloor \frac{\mathcal{N}_{r}}{2}\rfloor } \chi\big(|K_{t_{2n}}|\geq t^{\beta}   \big)\, \int_{t_{2n+1}}^{t_{2n+2}}ds\,\frac{dV}{dx}(X_{s}),
\end{multline}
where on the right side, we have neglected the integral from  $t_{2\lfloor \frac{\mathcal{N}_{r}}{2} \rfloor}$ to $r$, which will be small, and we define $t_{-1}=0$.  We will focus on the sum with interval starting at even numbered times $[t_{2n},t_{2n+1}]$.  First, we will show that the terms in the sum can be replaced by the same expressions conditioned on the information up to time $t_{n-1}$.  These conditional expectations can be uniformly bounded by an argument following at the end of the proof.  The strategy is to invoke Lemma~\ref{EIL} to guarantee that most of the terms in the sum have $|K_{t_{2n-1}}|$ on the order of $t^{\frac{1}{2}}$ rather than just $|K_{t_{2n}}|\geq t^{\beta}$.  This lowers the bounds available for the conditional expectations of those terms.

The following is a martingale
\begin{align}\label{MartiGraw}
Y_{r}=\sum_{n=0}^{\lfloor \frac{\mathcal{N}_{r}}{2}\rfloor } \chi\big(|K_{t_{2n-1}}|\geq \frac{1}{2}t^{\beta}\big)\,\Big(  \int_{t_{2n}}^{t_{2n+1}}ds\,\frac{dV}{dx}(X_{s})-\mathbb{E}\Big[\int_{t_{2n}}^{t_{2n+1}}ds\,\frac{dV}{dx}(X_{s})\, \Big|\, \mathcal{F}_{t_{2n-1}} \Big]\Big).
\end{align}
However, the second moment for a single term from the sum is bounded by
\begin{multline}\label{NotDrop}
\mathbb{E}\Big[ \chi\big(|K_{t_{2n-1}}|\geq \frac{1}{2}t^{\beta}\big) \,\Big|  \int_{t_{2n}}^{t_{2n+1}}ds\,\frac{dV}{dx}(X_{s})-\mathbb{E}\Big[\int_{t_{2n}}^{t_{2n+1}}ds\,\frac{dV}{dx}(X_{s})\, \Big|\, \mathcal{F}_{t_{2n-1}} \Big]\Big|^{2}\Big] \\ \leq
\mathbb{E}\Big[\Big|  \int_{t_{2n}}^{t_{2n+1}}ds\,\frac{dV}{dx}(X_{s})  \Big|^{2} \Big|\,
|K_{t_{2n-1}}|\geq t^{\beta}\Big]\leq 16\bar{V}^{2}t^{-2\beta} \\+ 2\mathcal{R}^{-2}\sup_{0\leq s\leq 1}\big|\frac{dV}{dx}(a) \big|^{2}\, \Pr\Big[ |K_{t_{2n}}|\leq \frac{1}{2}t^{2\beta}\Big| |K_{t_{2n-1}}|\geq t^{\beta}\Big],
\end{multline}
where we have considered separate bounds for the event that $|K_{t_{2n}}|\geq \frac{1}{2}t^{\beta}$ or $|K_{t_{2n}}|< \frac{1}{2}t^{\beta}$. When $|K_{t_{2n}}|\geq \frac{1}{2}t^{\beta}$, then we can apply Lemma~\ref{RealBasics} to bound the drift by $4\bar{V}t^{-\beta}$, and when $|K_{t_{2n}}|< \frac{1}{2}t^{\beta}$ then we use that the forces $\big|\frac{dV}{dx}(a) \big|$ are bounded and that the difference between two Poisson times has an exponential distribution with second moment $2\mathcal{R}^{-2}$.  Finally, since the force can only change the momentum by at most $\bar{V}^{\frac{1}{2}}\ll \frac{1}{4}t^{\beta}$ over any time interval, only a large momentum jump can send $|K_{t_{2n}}|$ below $\frac{1}{2}t^{\beta}$.  However, since the fourth moments of $\mathcal{P}_{a}(w)$ are less than $\rho$,
$$\Pr\Big[ |K_{t_{2n}}|\leq \frac{1}{2}t^{\beta}\Big| |K_{t_{2n-1}}|\geq t^{\beta}\Big]\leq \Pr\Big[|w_{n}|\geq \frac{1}{4}t^{\beta} \Big|  |K_{t_{2n-1}}|\geq t^{\beta}\Big]\leq  \sup_{a\in [0,1]} \int_{\frac{1}{4}t^{\beta} }^{\infty}dw\, \mathcal{P}_{a}(w)\leq \frac{4^{4}\rho}{t^{4\beta}},   $$
where the last inequality is Chebyshev's and thus the right side is $\mathit{O}(t^{-4\beta})$ which make the right term on the left side of~(\ref{NotDrop}) negligible compared to the left term.

Consider again the variance of a single term in~(\ref{MartiGraw}).
By Doob's maximal inequality
$$\mathbb{E}\big[\sup_{0\leq r\leq t} \big| Y_{r}\big|^{2}\big]\leq 4\,\mathbb{E}\big[\big|Y_{t}\big|^{2}\big]\leq \big(16\bar{V}^{2}t^{-2\beta}+\mathit{O}(t^{-4\beta}) \big) \mathbb{E}[\mathcal{N}_{t}]= 16\bar{V}^{2}\mathcal{R} t^{1-2\beta}+\mathit{O}(t^{1-4\beta}). $$
Thus we can focus on bounding the expressions
$\big|\mathbb{E}\big[\int_{t_{2n}}^{t_{2n+1}}ds\,\frac{dV}{dx}(X_{s})\,
\big|\, \mathcal{F}_{t_{2n-1}} \big]\big|$ when
$|K_{t_{2n-1}}|\geq t^{\beta}$.  The end result of the
analysis below will be to show that there is a constant $c$ such
that for all sufficiently large $t^{\beta}\gg 1$
\begin{align}\label{TypeA}
\Big|\mathbb{E}\big[\chi\big(|K_{t_{2n-1}}|\geq t^{\beta} \big)\int_{t_{2n}}^{t_{2n+1}}ds\,\frac{dV}{dx}(X_{s})\, \big|\, \mathcal{F}_{t_{2n-1}} \big]\Big|  \leq c\,t^{-2\beta}.
\end{align}
Applying the above inequality (also with $t^{\beta}$ replaced by $\epsilon t^{\frac{1}{2}}$),
$$\sum_{n=1}^{\lfloor \frac{\mathcal{N}_{t}}{2}\rfloor   } \Big|\mathbb{E}\big[\chi\big(|K_{t_{2n-1}}|\geq t^{\beta}\big)\int_{t_{2n}}^{t_{2n+1}}ds\,\frac{dV}{dx}(X_{s})\, \big|\, \mathcal{F}_{t_{2n-1}} \big]\Big|\leq c\, \mathcal{N}_{t,\epsilon}\epsilon^{-2}t^{-1}+c\,\big(\mathcal{N}_{t}- \mathcal{N}_{t,\epsilon}\big)\,t^{-2\beta},$$
where $ \mathcal{N}_{r,\epsilon}$ is the number of terms with  $|K_{t_{2n-1}}|  \geq \epsilon\,t^{\frac{1}{2}}$ up to time $r$.  By the triangle inequality
\begin{multline}\label{Height}
\mathbb{E}\Big[\Big| t^{-1+2\beta}  \sum_{n=1}^{\lfloor \frac{\mathcal{N}_{t} }{2}\rfloor   } \mathbb{E}\big[\chi\big(|K_{t_{2n-1}}|\geq t^{\beta}\big)\int_{t_{2n}}^{t_{2n+1}}ds\,\frac{dV}{dx}(X_{s})\, \big|\, \mathcal{F}_{t_{2n-1}} \big]\,\Big|^{2}\Big]^{\frac{1}{2}} \\ \leq 2^{\frac{1}{2}}c\,\mathcal{R}\,\epsilon^{-2}t^{-1+2\beta}+c\,t^{-1} \mathbb{E}\big[\big| \mathcal{N}_{t}-\mathcal{N}_{t,\epsilon} -\gamma_{\epsilon}(t)   \big|^{2}\big]^{\frac{1}{2}}+ c\,t^{-1} \mathbb{E}\big[\big| \gamma_{\epsilon}(t)   \big|^{2}\big]^{\frac{1}{2}}  ,
\end{multline}
where it was used that $\mathcal{N}_{t,\epsilon}\leq \mathcal{N}_{t}$, $\mathbb{E}\big[\big|\mathcal{N}_{t}\big|^{2}\big]^{\frac{1}{2}}\leq 2^{\frac{1}{2}}t\mathcal{R}$, and $ \gamma_{\epsilon}(t)$ is defined as
$$\gamma_{\epsilon}(t)=\mathcal{R}\sum_{n=1}^{\mathcal{N}_{r}}(t_{n}-t_{n-1})\chi(|K_{t_{n-1}}|\leq \epsilon t^{\frac{1}{2}}).$$
 The difference $\mathcal{N}_{r}- \mathcal{N}_{t,\epsilon}-\gamma_{\epsilon}(r) $ is a
  martingale since $t_{n}-t_{n-1}$ are exponentially distributed with mean $\mathcal{R}^{-1}$.
   The variance of the martingale satisfies
$$\mathbb{E}\big[\big|\mathcal{N}_{t}- \mathcal{N}_{t,\epsilon}-\gamma_{\epsilon}(t)\big|^{2}\big]= \mathbb{E}\big[\sum_{n=1}^{\mathcal{N}_{t} }|\mathcal{R}(t_{n}-t_{n-1})-1|^{2}\,\chi(|K_{t_{n-1}}|\leq \epsilon t^{\frac{1}{2}})     \big]  \leq \mathbb{E}[\mathcal{N}_{t}].$$
Thus the middle term on the right-side of~(\ref{Height}) is
$\mathit{O}(t^{\frac{1}{2}-2\beta})$.   $\gamma_{\epsilon}(t)$ is
less than the amount of time $r\in[0,\,t]$ the particle spends
with $t^{-\frac{1}{2}}E_{r}^{\frac{1}{2}}\leq \epsilon $.  In
other terms, $t^{-1}\gamma_{\epsilon}(t)\leq 1-
T_{\epsilon,\,V}^{(t)}$ where $T_{\epsilon,\,V}^{(t)}$ is defined
as in Lemma~\ref{EIL}.  Thus by Lemma~\ref{EIL}, $ \Pr\big[
t^{-1}\gamma_{\epsilon}(t) \geq \delta \big]\leq
C\frac{r_{2}^{\frac{1}{2}} }{r_{1}}\frac{\epsilon}{\delta}$.
Since $t^{-1} \gamma_{\epsilon}(t)$ is bounded by $1$,
$$ \mathbb{E}\big[\big| t^{-1} \gamma_{\epsilon}(t)   \big|^{2}\big]^{\frac{1}{2}}\leq  \delta \Pr\big[ t^{-1}\gamma_{\epsilon}(t) < \delta \big]+  \Pr\big[ t^{-1}\gamma_{\epsilon}(t) \geq \delta \big] \leq \delta+C\frac{r_{2}^{\frac{1}{2}} }{r_{1}}\frac{\epsilon}{\delta}\big.   $$
Thus we can pick $\delta$ to make the first term  small and then
pick $\epsilon$ to make the second term small. We now turn to
showing~(\ref{TypeA}).

 By the Markov property
\begin{multline}
\big|\mathbb{E}\big[\chi\big(|K_{t_{2n-1}}|\geq t^{\beta} \big)\int_{t_{2n}}^{t_{2n+1}}ds\,\frac{dV}{dx}(X_{s})\, \big|\, \mathcal{F}_{t_{2n-1}} \big]\big| \\ \leq \int dx\,dk\,P_{\omega}(x,\,k)\,\chi\big(|k|\geq \,t^{\beta}\big)\big|\mathbb{E}\big[\int_{t_{2n}}^{t_{2n+1}}ds\,\frac{dV}{dx}(X_{s})\, \big|\, (X_{t_{2n-1}},\,K_{t_{2n-1}})=(x,\,k) \big]\big|.
\end{multline}
where $P_{\omega}(x,\,k)$ is the distribution $(X_{t_{2n-1}},\,K_{t_{2n-1}})$ conditioned on $\omega\in \mathcal{F}_{t_{2n-1}}$.

Assuming that $|K_{t_{2n-1}}|\geq t^{\beta}$ and $K_{t_{2n}}\geq \frac{1}{2}t^{\beta}$,
\begin{multline}
\Big|\int_{t_{2n}}^{t_{2n+1}}ds\,\frac{dV}{dx}(X_{s})-\frac{V(X_{t_{2n+1}})-V(X_{t_{2n}})}{K_{t_{2n-1}} }\Big| \\ \leq \Big|\int_{t_{2n}}^{t_{2n+1}}ds\,\frac{dV}{dx}(X_{s})-\frac{1}{K_{t_{2n}}}\int_{t_{2n}}^{t_{2n+1}}ds \frac{dV}{dx}(X_{s})\,K_{s}\Big|+\big|\frac{1}{K_{t_{2n}}}-\frac{1}{K_{t_{2n-1}}}\big|\,\big|V(X_{t_{2n+1}})-V(X_{t_{2n}})\big|\\ < 2t^{-2\beta}\bar{V}\sup_{0\leq x\leq 1}\big|\frac{dV}{dx}(x)\big|\big(t_{2n+1}-t_{2n}\big) +2t^{-2\beta}\bar{V}|K_{t_{2n}}-K_{t_{2n-1}}|,
\end{multline}
where we have used the identity $\int_{r}^{t}ds \frac{dV}{dx}(X_{s})\,K_{s}=V(X_{t})-V(X_{r})$. The only thing random in the final bound is the difference $t_{n+1}-t_{n}$, which is an exponential random variable with mean $\mathcal{R}^{-1}$ and the difference $|K_{t_{2n}}-K_{t_{2n-1}}|$ which has variance less than $\frac{r_{2}}{\mathcal{R}}$.  Thus
\begin{multline}
\Big|\mathbb{E}\Big[\chi\big(|K_{t_{2n-1}}|\geq t^{\beta}\big)\int_{t_{2n}}^{t_{2n+1}}ds\,\frac{dV}{dx}(X_{s})\, \Big|\, \mathcal{F}_{t_{2n-1}} \Big]\Big| \\ < \mathit{O}(t^{-4\beta})+  2\,t^{-2\beta}\,\mathcal{R}\bar{V}   \sup_{0\leq x\leq 1}\big|\frac{dV}{dx}(x)\big|+ 2\,t^{-2\beta} \bar{V}\frac{r_{2}^{\frac{1}{2}} }{\mathcal{R}^{\frac{1}{2}}} \\+\int dx\,dk\,P_{\omega}(x,\,k)\,\chi\big(|k|\geq \,t^{\beta}\big)\Big|\mathbb{E}\Big[\frac{V(X_{t_{2n+1}})-V(X_{t_{2n}})}{K_{t_{2n-1}} } \, \Big|\, (X_{t_{2n-1}},\,K_{t_{2n-1}})=(x,\,k) \Big]\Big|,
\end{multline}
where $\omega \in \mathcal{F}_{2n-1}$ and
$\mathit{O}(t^{-4\beta})$ corresponds the unlikely event that
$\Pr\big[K_{t_{2n}}< \frac{1}{2}t^{\beta}\big]$ which we have
treated above following~(\ref{NotDrop}).

 Adding and subtracting the spatial average of the potential,
 $\int_{0}^{1}da\,V(a)=\mathcal{V}$ in the expectation above,
\begin{multline}\label{Lantern}
\Big|\mathbb{E}\Big[\frac{V(X_{t_{2n+1}})-V(X_{t_{2n}})}{K_{t_{2n-1}} }\, \Big|\, (X_{t_{2n-1}},\,K_{t_{2n-1}}) \Big]\Big| \\
 \leq \frac{1}{|K_{t_{2n-1}}|}\big( \int_{\R^{2}} dx\,dk\,P_{(X_{t_{2n-1}},\,K_{t_{2n-1}})}(x,\,k)\,\big|\mathbb{E}\big[V(X_{t_{2n+1}})-\mathcal{V}\, \big|\, (X_{t_{2n}},\,K_{t_{2n}})=(x,\,k) \big]\big| \\ +\big|\mathbb{E}\big[V(X_{t_{2n}})-\mathcal{V} \, \big|\, (X_{t_{2n-1}},\,K_{t_{2n-1}}) \big]\big|   \big),
\end{multline}
where $P_{(X_{t_{2n-1}},\,K_{t_{2n-1}})}$ is the probability density for $(X_{t_{2n}},\,K_{t_{2n}})$ given $(X_{t_{2n-1}},\,K_{t_{2n-1}})$.

Finally, we can work with quantities that allow more explicit
expressions
$$ \mathbb{E}\big[V(X_{t_{2n}})  \, \big|\, (X_{t_{n-1}},\,K_{t_{n-1}}) \big]=\int_{0}^{\infty}dt\,\mathcal{R}\,e^{-\mathcal{R}t}V(x_{t})=\int_{0}^{1}da\,r_{(x_{0},\,k_{0})}(a)\,V(a), $$
where $x_{t}$ is the position at time $t$ for the particle
evolving according to the dynamics from the initial point
$(x_{0},\,k_{0})=(X_{t_{n-1}},\,K_{t_{n-1}})$, and $r_{(x,\,k)}\in L^{1}([0,1])$ is defined as in Lemma~\ref{SingleJump}.
$$\big|\mathbb{E}\big[V(X_{t_{2n}})  \, \big|\, (X_{t_{n-1}},\,K_{t_{n-1}}) \big]-\mathcal{V}\big|\leq \bar{V}\,\int_{0}^{1}da\,\big|r_{(x_{0},\,k_{0})}(a)-1\big|.   $$
 By Part (1) of Lemma~\ref{SingleJump}, when $|K_{t_{n-1}}|\geq t^{\beta}$, then $|r_{(x_{0},\,k_{0})}(a)-1|\leq 4\mathcal{R}\,t^{-\beta}+\mathit{O}(t^{-2\beta})$.  A similar analysis bounds the term $\big|\mathbb{E}\big[V(X_{t_{2n+1}})-\mathcal{V}\, \big|\, (X_{t_{2n}},\,K_{t_{2n}})=(x,\,k) \big]\big|$.  Thus with the factor of $|K_{t_{n-1}}|^{-1}$ on the right side of~(\ref{Lantern}),  then~(\ref{Lantern}) is $\mathit{O}(t^{-2\beta})$, which completes the proof.

\end{proof}

\section{Bounding the momentum drift   }\label{SecDrift}

In general, we have that $K_{t}=K_{0}+M_{t}+\int_{0}^{t}dr
\,\frac{dV}{dx}(X_{r})$.  In this section, we develop tools for
controlling the cumulative drift
$\int_{0}^{st}dr\,\frac{dV}{dx}(X_{r})$.  The end result, under
the assumption of the symmetry (\textit{iii}) of~\ref{AssumpTwo},
is that
\begin{align}\label{Horse}
\mathbb{E}\Big[\Big|\sup_{0\leq s\leq 1}
t^{-\frac{1}{2}}\int_{0}^{st}ds\,\frac{dV}{dx}(X_{s})
\Big| \Big]\longrightarrow 0.
\end{align}
Thus on the scale $t^{\frac{1}{2}}$ of  a central limit theorem
for $K_{t}$, the drift term vanishes.

  By the Lemma~\ref{EIL}, the particle  spends most of the time at ``high energy''
   (in Lemma~\ref{EIL} this meant $\propto t$),
    where the contribution to the total drift over any given finite
    time interval is small.  However, the particle is also making
     occasional shorter incursions to ``low energy''
      where the contribution may be larger over a finite interval.
      In this section, ``low energy'' roughly means below $t^{\frac{1}{4}}$.
In order to bound~(\ref{Horse}), the  analysis is split into parts
treating the drift at high and low energies respectively.

From this section onwards, we contract the position degree of
freedom to a single periodic cell $x\in[0,1]$.  This clearly does
not affect the statistics for the drift process~(\ref{Horse}).
Thus the dynamics satisfies the same linear Boltzmann
equation~(\ref{Boltzmann}) as before but with periodic boundary
conditions; the derivatives in position at the boundaries of the
interval are symmetric.

We now define what we mean by low energy incursions.  They are
limited by starting and ending times. Define the hitting time $\theta_{0}= \textup{min}\{s\in [0,t]\,
\big| \, |K_{s}|\geq \,t^{\frac{1}{4}} \} $.  For $j\geq 1$ define
the sequences of hitting times $\sigma_{j}, \theta_{j}$:
\begin{eqnarray}\label{StopTime}
\sigma_{j}&=& \textup{min}\{s\in [0,\infty), M_{s}-M_{s^{-}}\neq 0 \big| \, s>\theta_{j-1},\,  |K_{s}|< t^{\frac{1}{4}} \},\\
\theta_{j}&=& \textup{min}\{s\in [0,\infty) \,\big| \, s>\sigma_{j},\,  |K_{s}|>2\, t^{\frac{1}{4}} \},
\end{eqnarray}
Notice that $\theta_{0}$ is defined  differently than $\theta_{j}$
for $j\geq 1$.  We refer to $[\sigma_{j},\theta_{j}]$ as the time
period of the $j$th incursion.

In the lemma below, we give a bound on the expected number of
incursions $N_{Y}(\varrho)$ over a time interval $[0,\varrho t]$,
and show that the  time periods of incursions
$\theta_{j}-\sigma_{j}$ have finite first moments.  The time
periods between incursions $\sigma_{j+1}-\theta_{j}$ can be shown
to be almost surely finite.   This follows from an argument using
L\'evy's zero-one law and Theorem~\ref{MainThm}, but showing
$\sigma_{j+1}-\theta_{j}$ to be finite is not required to prove
Theorem~\ref{MainThm}.  In any case, bounds on
$\sigma_{j+1}-\theta_{j}$ are intrinsically less important to us,
since the challenge is to get estimates for the low-energy part of
the walk.

\begin{lemma}\label{FiniteTimes}
Assume~\ref{AssumpOne} and (\textit{i}) of List~\ref{AssumpTwo}.

\begin{enumerate}

\item
  Given $\sigma_{j}<\infty$,
   the difference $\theta_{j}-\sigma_{j}$ has expectation $\mathit{O}(t^{\frac{1}{2}})$.

\item  Let $\varrho>0$.  For large enough $t$, the expectation for
the number of incursions in the interval $[0,\varrho t]$ is
bounded as
$$ \mathbb{E}[N_{Y}(\varrho t)]\leq \varrho^{\frac{1}{2}}  r_{2}^{\frac{1}{2}} t^{\frac{1}{4} }$$

\end{enumerate}

\end{lemma}
\begin{proof} \text{ }\\
\noindent Part (1): \vspace{.1cm}

Let us set construct the stopping time $\theta_{\mathcal{T}}=(\theta_{j}-\sigma_{j}) \wedge \mathcal{T}$ for some bound $\mathcal{T}> 1$ and set $\sigma_{j}=0$ .
$$\mathbb{E}\big[\theta_{\mathcal{T}} \big]\leq  \frac{r_{2}}{r_{1}}\mathbb{E}\Big[\langle M\rangle_{\theta_{T}}\Big] = \frac{r_{2}}{r_{1}}\mathbb{E}\big[ [ M]_{\theta_{T}} \big], $$
where the equality follows since $[M]_{t}-\langle M\rangle_{t}$ is a martingale, and the inequality is a consequence of $r_{1}\leq \frac{d}{dt}\langle M\rangle_{t}\leq r_{2}$.

Lemma~\ref{DOOOOB} states that
$2^{-1}\sum_{n=1}^{\mathcal{N}_{r}}w_{n}^{2}=2^{-1}[
M]_{r} $ differs from $E_{r}-E_{0}$ by a martingale, and by the Optional Sampling
Theorem
 the expectation of the martingale part is zero at time $\theta_{\mathcal{T}}$ so
$$\mathbb{E}\big[[ M]_{\theta_{\mathcal{T} }}  \big]=2\mathbb{E}\big[ E_{\theta_{\mathcal{T}}}-E_{0}\big]. $$

Let $D$ be the size of the over-jump of the boundary
$-2t^{\frac{1}{4}}$ or $2t^{\frac{1}{4}}$ when $\theta_{\mathcal{T}}<\mathcal{T}$, and $\bar{V}$
be the max of the potential.
$$\mathbb{E}\big[E_{\theta_{T}}-E_{0}\big]\leq  \mathbb{E}\big[\big|\frac{1}{2}(2t^{\frac{1}{4}}+D)^{2}+\bar{V}\big|  \theta_{\mathcal{T}} < \mathcal{T}\big]+(2t^{\frac{1}{2} }+\bar{V})
\Pr\big[\theta_{\mathcal{T}} =\mathcal{T}\big]\leq 2t^{\frac{1}{2}}+\mathit{O}(1)<3^{\frac{1}{2}}        $$
 By Lemma~\ref{BoundryVariance},
 there are universal bounds determined by $\mathcal{C}$ and $\eta$ on
  all the moments of $D$,   $\mathbb{E}\big[ D^{2} \big]<\rho_{2}(\mathcal{C},\,\eta)$.
 Using that $|x+ y|^{2}\leq 2\,x^{2}+2\,y^{2}$,
 $\Pr\big[\theta>\mathcal{T}\big]\leq 1$, and that $\bar{V}\ll t^{\frac{1}{4}}$
$$ \mathbb{E}[\theta_{\mathcal{T}}]<  \frac{r_{2}}{r_{1}} \big( 6\,t^{\frac{1}{2}}+\mathit{O}(1)    \big)=\mathit{O}(t^{\frac{1}{2}} ). $$
Finally, by taking the limit $\mathcal{T}\rightarrow \infty$, we get a bound for the second moment of $\theta_{j}-\sigma_{j}$:
$$\mathbb{E}\big[\theta_{j}-\sigma_{j}\big]=\limsup_{ \mathcal{T}\rightarrow \infty} \mathbb{E}[\theta_{\mathcal{T}}]\leq  \frac{r_{2}}{r_{1}} \big( 6\,t^{\frac{1}{2}}+\mathit{O}(1)   \big)=\mathit{O}(t^{\frac{1}{2}} ).  $$

\noindent Part (2):\vspace{.1cm}

By Part (1) each incursion ends.  Thus for each
count of $N_{Y}(\varrho t)$ there is a distinct up-crossing in
which $|K_{s}|$ begins below $t^{\frac{1}{4}}$ and ends up above
$2 t^{\frac{1}{4}}$.  However, for large values of momentum
$\frac{1}{\sqrt{2}}|K_{s}|\approx E_{s}^{\frac{1}{2}}$, and we
bound $N_{Y}$ by the number of up-crossings $U_{\varrho
t}\big(2^{-1}
t^{\frac{1}{4}},\,2\,t^{\frac{1}{4}};E_{s}(\omega)\big)$ that
$E_{s}^{\frac{1}{2}}$ makes between $2^{-1}\,t^{\frac{1}{4}}$ and
$2 t^{\frac{1}{4}}$.  Since $E_{s}^{\frac{1}{2}}$ is a
submartingale, we can apply the submartingale up-crossing
inequality~\cite{Chung} to obtain
$$
\mathbb{E}[N_{Y}(\varrho t)]\leq \mathbb{E}\big[U_{\varrho t}\big(2^{-1} t^{\frac{1}{4}},\,2\,t^{\frac{1}{4}};E_{s}(\omega)\big)\big]\leq  \frac{\mathbb{E}[E_{\varrho t}^{\frac{1}{2}}] }{2\,t^{\frac{1}{4}}-2^{-1}t^{\frac{1}{4} }} \approx  \frac{\sqrt{2}}{3}\,r_{2}^{\frac{1}{2}}\varrho^{\frac{1}{2}} t^{\frac{1}{4}}<r_{2}^{\frac{1}{2}}\varrho^{\frac{1}{2}} t^{\frac{1}{4}}. \vspace{.5cm}
$$

\end{proof}

The basic idea for our analysis is the following:
\begin{itemize}
\item For large $t$, there is an  asymptotic independence between
the events during a single incursion and all events up to the end
of a previous incursion;

\item  Events during the  incursion, which occur far enough after
the starting time of the incursion are independent of the initial
state of the incursion.
\end{itemize}

It is convenient in many places to have an effective bound on the
size of the momentum jumps that are likely to occur in the
interval $[0,t]$.  We thus consider the statistics for our model
conditioned on the event
$$\{ |v_{n}|\leq t^{\frac{1}{40}} \text{ for all $n$ such that } t_{n}\leq t \},$$
 where a jump greater than $t^{\frac{1}{40}}$ is considered to be large. The lemma below
shows that the probability that there is a jump above
$t^{\frac{1}{40} }$ over the interval $[0,t]$ decays super-polynomially,
and that for dealing with the drift over up to time $t$, we can
neglect the possibility of large jumps.  The choice of $\frac{1}{40}$ involves constraints from Proposition~\ref{BC2NHalf}.

\begin{lemma}\label{CappedJumps}
Assume (\it{i}) of List~\ref{AssumpTwo}, then the probability of a
momentum jump $v_{n}$ with $|v_{n}|\geq t^{\frac{1}{40} }$ over the
interval $[0,t]$ is $\mathit{O}(t\,e^{-\frac{\eta}{2}t^{\frac{1}{40} } })
$, for $\eta$ as in (\textit{i}) of List~\ref{AssumpTwo}.
Moreover, the difference in the quantity
$$\mathbb{E}\Big[\Big|\sup_{0\leq s\leq 1}
t^{-\frac{1}{2}}\int_{0}^{st}ds\,\frac{dV}{dx}(X_{s}) \Big|^{2}
\Big]^{\frac{1}{2}}, $$ for the dynamics conditioned not to make
jumps greater than $t^{\frac{1}{40} }$ and the unconditioned dynamics is
$\mathit{O}(t^{2}\,e^{-\frac{\eta}{2}t^{\frac{1}{40} } }) $

\end{lemma}

\begin{proof}
Assumption (\textit{i}) of List~\ref{AssumpTwo} implies that the Laplace transforms for single momentum jumps are uniformly bounded by $\mathcal{C}(1-e^{-\eta-q})^{-1}$.  Using Chebyshev's inequality we can bound the probability that any jump is greater than $t^{\frac{1}{40}}$.  On the other hand, the integral of the drift can be at most $t\, \sup_{a\in [0,1]}\big|\frac{dV}{dx}(a)\big|$.

\end{proof}

 We refer to Appendix~\ref{AppendixBC} for a discussion of boundary crossing distributions and the
definition of the boundary crossing density
$\phi_{\infty}:[0,1]\times \R^{+}\rightarrow \R^{+}$.
If $H$ is a random variable, then $\Pr[H=y]$, for dummy variable $y\in \R$, refers to the distributional measure of $H$ or its probability density if it exists.  For a signed measure $\mu$ on $\R$, then $\|\mu\|_{1}=|\mu|(\R) $, where $|\mu|$ is the absolute value of the measure.  Most of the random variables in this article (e.g. energy and momentum jumps) have well-defined densities.

The following lemma states that incursions beginning at points
$(X_{\sigma_{j}},K_{\sigma_{j}})$ for $s_{1}K_{\sigma_{j}}\in
[t^{\frac{1}{4}}-t^{\frac{1}{40} }, t^{\frac{1}{4}}]$ for fixed $s_{1}=\pm$ all
have  approximately the same probabilities for ending the
incursion in the positive or negative direction.

\begin{lemma}\label{ExitProb}
Assume List~\ref{AssumpOne} and (\it{i})-(\it{ii}) of
List~\ref{AssumpTwo}.  Consider the dynamics conditioned not to
have jumps greater than $t^{\frac{1}{40}}$.   Let
$\mathbf{s}_{1},\mathbf{s}_{2}\in \{+,-\} $, and
$(x,\mathbf{s}_{1}k)\in [0,1]\times
[t^{\frac{1}{4}}-t^{\frac{1}{40}}, t^{\frac{1}{4}}]$
  There are constants $\rho_{\mathbf{s}_{1},\,\mathbf{s}_{2}}(t)$ such that as $t\rightarrow \infty$
$$\sup_{(x,k)}\Big|  \frac{ \Pr\big[\mathbf{s}_{2}K_{\theta_{j}}>0\,\big|\, (X_{\sigma_{j}},\,\mathbf{s}_{1} K_{\sigma_{j}})=(x,k)\big]}{\rho_{\mathbf{s}_{1},\,\mathbf{s}_{2}}(t)}-1 \Big|\longrightarrow 0$$

\end{lemma}

\begin{proof}
Fix $\mathbf{s}_{1}=\mathbf{s}_{2}=+$.  Define
$\phi_{\uparrow,\,t}^{(a,\,v)}$,
$\phi_{\downarrow,\,t}^{(a,\,v)}$ to be the boundary
crossing distributions above and below the set
$S=[t^{\frac{1}{4}}-t^{\frac{1}{40} },\,t^{\frac{1}{4}}+t^{\frac{1}{40} }]$
starting from $(a,\,t^{\frac{1}{4}}-v)$.

Using the Markov property
and that $H$ is a function of the process after the time $\tau$,
\begin{multline}
 \textup{Pr}\big[ K_{\theta_{j}}>0 \,\big| \,(X_{\sigma_{j}},K_{\sigma_{j}})= (a,\,t^{\frac{1}{4}}-\,v) \big]\\ =\int_{[0,1]\times \R^{+}} dq\,dp\, \Big( \phi_{\uparrow,\,t}^{(a,\,v)}(q,\,p)
 \Pr\big[K_{\theta_{j}}>0\,\big|\, (X_{\tau},K_{\tau}) = (q,\,t^{\frac{1}{4}}+t^{\frac{1}{20} }+p) \big]\\+\phi_{\downarrow,\,t}^{(a,\,v)}(q,\,p)\Pr\big[K_{\theta_{j}}>0\,\big|\, (X_{\tau},\,K_{\tau})= (q,\,t^{\frac{1}{4}}-\,t^{\frac{1}{20} }-p) \big]     \|_{1}.
\end{multline}
 Thus for
$(a,\,v),\,(a^{\prime},\,v^{\prime})\in [0,1]\times
[0,t^{\frac{1}{40} }]$,
\begin{multline}
\sup_{a,\,a^{\prime},\,v,\,v^{\prime} }  \|\textup{Pr}\big[H=y\,\big|
\,(X_{\sigma_{j}},K_{\sigma_{j}})= (a,\,t^{\frac{1}{4}}-\,v) \big]- \textup{Pr}\big[H=y\,\big| \,(X_{\sigma_{j}},K_{\sigma_{j}})= (a^{\prime},\,t^{\frac{1}{4}}-\,v^{\prime}) \big] \|_{1}\\ \leq \sup_{a,\,a^{\prime},\,v,\,v^{\prime} }\|\phi_{\uparrow,\,t}^{(a,\,v)}-  \phi_{\uparrow,\,t}^{(a^{\prime},v^{\prime})}\|_{1}+\sup_{a,\,a^{\prime},\,v,\,v^{\prime} }\|\phi_{\downarrow\,t}^{(a,\,v)}-  \phi_{\downarrow,\,t}^{(a^{\prime},v^{\prime})}\|_{1}
\end{multline}
By Proposition~\ref{BC3}, $\phi_{\uparrow,\,t}^{(a,\,v)}$\,and
$\phi_{\downarrow,\,t}^{(a,\,v)}$ converge uniformly to
 $\phi_{\infty}$ in $L^{1}$. Thus the diameter $D(t)$ of the set of possible values for $\Pr\big[K_{\theta_{j}}>0\,\big|\, (X_{\tau},K_{\tau}) = (a,\,t^{\frac{1}{4}}-v) \big]$  as a function of $(a,v)$ shrinks to zero as $t\rightarrow \infty$.
Let us define
$$\rho_{+,\,+}(t)= \Pr\big[K_{\theta_{j}}>0\,\big|\, (X_{\sigma_{j}},\, K_{\sigma_{j}})=(x,k)\big]+D_{t}^{\frac{1}{2}}+\frac{1}{t}, $$
for any choice of  $(x,k)$ where $t^{-1}$ is merely to ensure that
$\rho_{+,\,+}(t)$ is non-zero, and the square root is introduced
so that $D_{t}D_{t}^{-\frac{1}{2}}=D_{t}^{\frac{1}{2}}\rightarrow
0$.  Then we will have the conclusion of the lemma.

\end{proof}

\subsection{Bounding the drift over an incursion          }

In this section, we define incursions to have end times
$\varsigma_{j}$ which are different but related to the end times
to $\theta_{j}$.  Define the sequence of times $\varsigma_{j}$:
\begin{align}\label{RevStopTime}
\varsigma_{j}= \textup{min}\{s\in [0,t],\, M_{s}-M_{s^{-}}\neq 0
\, \big| \, s>\sigma_{j},\,  \inf_{s< r\leq \theta_{j} } |K_{r}|
> t^{\frac{1}{4}} \}.
\end{align}
The $\varsigma_{j}$'s are not hitting times  since information up
to time $\theta_{j}$ is required to determine them.  However, for
the time-reversed dynamics the $\varsigma_{j}$ are well-defined
hitting times and are defined to be symmetric with the
$\sigma_{j}$'s.  This will be important in the next section. All of the results in this section apply when $\varsigma_{j}$ is replaced by $\theta_{j}$.

 Define the random variables
$$ Y_{j}=t^{-\frac{1}{4}}\int_{\sigma_{j}}^{\varsigma_{j}}dr\,\frac{dV}{dx}(X_{r}). $$
We consider the $Y_{r}$'s to contain the low  energy contribution
to the cumulative drift.  Yet, during some portion of the interval
$[\sigma_{j},\varsigma_{j}]$, the particle is likely to spend some
time with momentum above $t^{\frac{1}{4}}$, and thus have overlap
with the quantity in Lemma~\ref{HEDrift} with $\beta=\frac{1}{4}$.  Let $N_{Y}(t)$ be the
number of $Y_{j}$ terms up to time $t$. It is convenient to split
the $Y_{j}$'s into four classes.  For $m\geq 1$,
$\mathbf{s}_{1},\mathbf{s}_{2}\in \{+,-\} $,
\begin{align*}
Y_{\mathbf{s}_{1},\, \mathbf{s}_{2}}(m)=t^{-\frac{1}{4}}\chi(\mathbf{s}_{1}K_{\sigma_{j}}>0,\,  \mathbf{s}_{2}K_{\varsigma_{j}}>0  )\, \int_{\sigma_{j}}^{\varsigma_{j} }dr\,\frac{dV}{dx}(X_{r}) ,
\end{align*}
where $j$ and $m$ are related through
\begin{align}\label{jANDm}
 j=\min \{n\geq 0 \,\big| \,  m=\sum_{i=1}^{n}\chi( \mathbf{s}_{1}K_{\sigma_{i}}>0)\, \}.
\end{align}
 $Y_{\mathbf{s}_{1},\, \mathbf{s}_{2}}(m)$ is
 equal to the drift for the $m$th incursion that begins
 with momentum having sign $s_{1}$ provided that the incursion ends
 with sign $s_{2}$.  Naturally, if there is no $m$th
 incursion with sign $s_{1}$, then we
  set $Y_{\mathbf{s}_{1},\, \mathbf{s}_{2}}(m)=0$. For $\mathbf{s}\in \{\pm \} $, we also define $N_{\mathbf{s}}(r)$ to  be the number $N_{\mathbf{s}}(r)= \sum_{i=1}^{N_{Y}(r)}\chi(\mathbf{s}K_{\sigma_{i}}>0 )$.

 Define the constant
 \begin{align}\label{FirstStop}
c_{\mathbf{s}_{1},\,\mathbf{s}_{2}}(t)=t^{-\frac{1}{4}}\int_{[0,1]\times (0,\, t^{\frac{1}{40} }) } da\,dv\,\phi(a,v)\,\mathbb{E}_{(a,\,\mathbf{s}_{1}t^{\frac{1}{4}}-\mathbf{s}_{1}v )}\big[ \int_{0}^{\varsigma}ds\,\frac{dV}{dx}(X_{s}) \chi(\mathbf{s}_{2}K_{\varsigma}>0)   \big],
\end{align}
where $\varsigma$ is defined analogously to the $\varsigma_{j}$'s as the last time that there is a momentum jump inside $(-t^{\frac{1}{4}}, t^{\frac{1}{4}})$ before exiting the larger interval $(-2\,t^{\frac{1}{4}},\,2\, t^{\frac{1}{4}})$:
\begin{eqnarray*}
\varsigma &= &\textup{min}\{s\in [0,\theta), \, M_{s}-M_{s^{-}}\neq 0  \, \big|\,  \inf_{s< r\leq \theta } |K_{r}|   > t^{\frac{1}{4}} \},\\
\theta &=& \textup{min}\{s\in [0,\infty) \,\big| \,  |K_{s}|>2\, t^{\frac{1}{4}} \}.
\end{eqnarray*}
In fact, $c_{\mathbf{s}_{1},\,\mathbf{s}_{2}}(t)$ will be asymptotically close to  $t^{-\frac{1}{4}}\mathbb{E}_{(a,\,\mathbf{s}_{1}t^{\frac{1}{4}}-\mathbf{s}_{1}v )}\big[ \int_{0}^{\varsigma}ds\,\frac{dV}{dx}(X_{s})\, \chi( \mathbf{s}_{2}K_{\theta}>0   )\big]$ for any fixed $(a,v)\in [0,1]\times \R^{+}$ as $t\rightarrow \infty$  so the density $\phi(a,v)$ appearing in the definition of $c_{\mathbf{s}_{1},\mathbf{s}_{2}}$ is not important (except as a matter of convenience).

The main purpose of
 the following proposition is to establish Part (3) which says that the sum of the $Y_{\mathbf{s}_{1},\,\mathbf{s}_{2}}(m)$'s can be replaced by the number $N_{\mathbf{s}_{1}}(st)$ multiplied by the constant $ c_{\mathbf{s}_{1},\mathbf{s}_{2}}(t)$.

\begin{proposition}\label{OddsEnds}
Assume List~\ref{AssumpOne} and (\textit{i})-(\textit{ii}) of~\ref{AssumpTwo}.
\begin{enumerate}
\item
For large enough $t$, then for all $j$ and $\omega \in \mathcal{F}_{\sigma_{j}}$
$$
\mathbb{E}[Y_{j}^{2}|\mathcal{F}_{\sigma_{j}} ]^{\frac{1}{2}}< 5\,.
$$

\item Let $c_{\mathbf{s}_{1},\mathbf{s}_{2}}(t)$ be defined as in~(\ref{FirstStop}) and $j$ and $m$ be related by~(\ref{jANDm}).
As $t\rightarrow \infty$, we have the $L^{2}(\Omega)$ convergence
$$\mathbb{E}\big[ \big|\,\mathbb{E}[Y_{\mathbf{s}_{1},\,\mathbf{s}_{2}}(m) | \mathcal{F}_{\theta_{j-1}} ]- c_{\mathbf{s}_{1},\mathbf{s}_{2}}(t) \big|^{2}\big]^{\frac{1}{2}}\longrightarrow 0.$$

\item As $t\rightarrow \infty$, there is $L^{2}(\Omega)$ convergence
$$  t^{-\frac{1}{4}} \mathbb{E}\Big[ \sup_{0\leq s\leq 1}\Big|\sum_{m=1}^{N_{\mathbf{s}_{1}}(st) }Y_{\mathbf{s}_{1},\,\mathbf{s}_{2}}(m)- \,N_{\mathbf{s}_{1}}(st)\, c_{\mathbf{s}_{1},\mathbf{s}_{2}}(t)     \Big|\Big]\longrightarrow 0.   $$

\end{enumerate}

\end{proposition}

\begin{proof}\text{ }\noindent \\
Part ($1$):\vspace{.1cm}

By Lemma~\ref{CappedJumps}, we can take the jumps to be bounded by
$t^{\frac{1}{40} }$, although we will not employ this till the end
of the proof.  Set $\sigma=\sigma_{j}$, $\varsigma=\varsigma_{j}$,
and $\theta=\theta_{j}$.   By the triangle inequality,
$$
\mathbb{E}[Y_{j}^{2}|\mathcal{F}_{\sigma} ]^{\frac{1}{2}}\leq t^{-\frac{1}{4}}\mathbb{E}[(M_{\varsigma^{-}}-M_{\sigma})^{2}|\mathcal{F}_{\sigma}]^{\frac{1}{2}}+t^{-\frac{1}{4}}\mathbb{E}\big[(K_{\varsigma^{-}}-K_{\sigma})^{2}|\mathcal{F}_{\sigma }]^{\frac{1}{2}}.
$$
The times $\sigma$ and $\varsigma$ are defined such that  $|K_{\sigma}|,\, |K_{\varsigma^{-}}|\leq t^{\frac{1}{4}}$.  Thus $\mathbb{E}\big[(K_{\varsigma^{-}}-K_{\sigma})^{2}|\mathcal{F}_{\sigma}]^{\frac{1}{2}}\leq 2\,t^{\frac{1}{4}}$.

As in the proof of Part (1) of Lemma~\ref{FiniteTimes} define the stopping time $\theta_{\mathcal{T}}=\theta \wedge \mathcal{T}$ and the  capped time $\varsigma_{\mathcal{T}}=\varsigma \wedge \mathcal{T}$.

By Doob's maximal inequality and $\varsigma_{\mathcal{T} }\leq \theta_{\mathcal{T} }$,
\begin{multline}\label{AnotherDoob}
\mathbb{E}\big[\big|M_{\varsigma_{\mathcal{T}}^{-}}-M_{\sigma}\big|^{2}|\mathcal{F}_{\sigma}\big]^{\frac{1}{2}}\leq \mathbb{E}\big[ \sup_{ \sigma \leq r\leq \theta_{\mathcal{T}}  } \big| M_{r}-M_{\sigma}\big|^{2}|\mathcal{F}_{\sigma}\big]^{\frac{1}{2}}\\ \leq 2\mathbb{E}\big[ \big| M_{\theta_{\mathcal{T}}}-M_{\sigma }\big|^{2}|\mathcal{F}_{\sigma}\big]^{\frac{1}{2}}=2\mathbb{E}\big[ [ M]_{\theta_{\mathcal{T}} }-[ M]_{\sigma} |\mathcal{F}_{\sigma}\big]^{\frac{1}{2}},
\end{multline}
where $[ M]_{t} $ is the quadratic variation and of the martingale $M_{t}$ up to time $t$, and the last equality follows from the optional sampling theorem.
By Lemma~\ref{DOOOOB}, $\mathcal{A}_{t}$ in the increasing part in a Doob-Meyer decomposition for $E_{t}$,
$$2\mathcal{A}_{t}-2\mathcal{A}_{s}  = [ M]_{t} -[ M]_{s}    . $$
Applying the above for the time interval from $\sigma$ up to stopping time $\theta_{T}$,
the right-hand side of~(\ref{AnotherDoob}), is bounded by
$$
2 \mathbb{E}[\mathcal{A}_{\theta_{T}}-\mathcal{A}_{\sigma} |\mathcal{F}_{\sigma}]=2 \mathbb{E}[E_{\theta_{T}}-E_{\sigma} |\mathcal{F}_{\sigma}]\leq 4t^{\frac{1}{2}}+\mathit{O}(1) $$
where the equality is another use of the optional sampling theorem.  The inequality comes from the proof of Part (1) of Lemma~\ref{FiniteTimes}.\vspace{.2cm}

\noindent Part (2):\vspace{.1cm}

Let us take $\mathbf{s}_{1},\,\mathbf{s}_{2}=+$.   By
Lemma~\ref{CappedJumps}, we can take jumps to be bounded by
$t^{\frac{1}{40} }$.  By the Markov property and by the definition of
$Y_{+,\,+}(m)$,
\begin{multline}\label{Piglet}
\mathbb{E}\big[Y_{+,\,+}(m)|\mathcal{F}_{\theta_{j-1}}
\big]
=\mathbb{E}_{(X_{\theta_{j-1}},K_{\theta_{j-1}})}\big[t^{-\frac{1}{4}}\int_{0}^{\varsigma}ds\,\frac{dV}{dx}(X_{s})\,\chi(K_{\varsigma}>0
) \big]\\ = \int_{[0,1]\times [0,t^{\frac{1}{40}}]
}da^{\prime}\,dv^{\prime}\,  \phi_{t} (a^{\prime},\,v^{\prime})\,
\mathbb{E}_{ (a^{\prime},\,t^{\frac{1}{4}}-v^{\prime}
)}\big[t^{-\frac{1}{4}}\int_{0}^{\varsigma}ds\,\frac{dV}{dx}(X_{s})\,\chi(K_{\varsigma}>0)
\big].
\end{multline}
where $\phi_{t}(a,\,v)$ is the joint distribution of
$(X_{\sigma_{j}},\,-K_{\sigma_{j}}+t^{\frac{1}{4}} )$ given
$(X_{\theta_{j-1}},\,K_{\theta_{j-1}})$. We attach a subscript
$(a,\,v)$ to the symbol $\phi_{t}$ to indicate the ending point
$(X_{\theta_{j-1}},\,K_{\theta_{j-1}})=(a,t^{\frac{1}{4}}-\,v)$ of
the last previous incursion.

We have that
\begin{multline}
\mathbb{E}\Big[ \Big|c_{+,\,+}(m) - \mathbb{E}\big[Y_{+,+}(m)\,\big|\,\mathcal{F}_{\theta_{j-1}} \big] \Big|^{2} \Big]
\\
 \leq  \sup_{(a^{\prime},\,v^{\prime})} \mathbb{E}_{ (a^{\prime},\,t^{\frac{1}{4}}-v^{\prime} )}\Big[\Big| t^{-\frac{1}{4}}\int_{0}^{\varsigma}ds\,\frac{dV}{dx}(X_{s})\Big|^{2} \,\chi(\,K_{\sigma}>0)    \Big]\, \mathbb{E}\Big[   \| \phi_{t,\,(a,\,v)} -\phi_{\infty}\|_{1}\big|_{(X_{\theta_{j-1}},\,t^{\frac{1}{4}}-K_{\theta_{j-1}})}^{(a,v)} \Big]  \\ \leq  25\,r_{2}  \sup_{(a,\,v)\in [0,1]\times \R^{+}}  \| \phi_{t,\,(a,\,v)} -\phi_{\infty}\|_{1},
\end{multline}
where the first inequality follows from the definition of
$c_{+,+}$ and~(\ref{Piglet}) to which we apply Jensen's inequality
over the measure determined by
$\big|\phi_{t,\,(a,\,v)}(a^{\prime},v^{\prime})
-\phi_{\infty}(a^{\prime},v^{\prime})|\,da^{\prime}\,dv^{\prime} $
and finally H\"older's inequality to pull the supremum outside the
integral. The second inequality follows since $\mathbb{E}_{(x,k
)}\big[\big|
t^{-\frac{1}{4}}\int_{0}^{\varsigma}ds\,\frac{dV}{dx}(X_{s})\big|^2\big]$
is smaller than $25 r_{2}$  for $(x,k)\in [0,1]\times
[0,t^{\frac{1}{20} }]$ by the same argument as in Part (1).

Finally by Proposition~\ref{BC4}, $\phi_{t,\,(a,\,v)}$ converges to $\phi_{\infty}$ in $L^{1}$ uniformly for $(a,\,v)\in [0,1]\times [0,2t^{\frac{1}{4}}] $ (which includes $[0,1]\times [0,t^{\frac{1}{40}}]$) as $t\rightarrow \infty$.

\vspace{.25cm}

\noindent Part (3):\vspace{.1cm}

Again we invoke Lemma~\ref{CappedJumps}, to work with the process
conditioned to have jumps bounded by $t^{\frac{1}{40} }$.  Let
$\mathcal{F}_{\theta_{j-1}}$ be the $\sigma$-algebra of all
information known up to the end of the last incursion
$\theta_{j-1}$, and $j$ and $m$ are related by~(\ref{jANDm}).  For
a random process $X_{s}$, $0\leq s \leq 1$, define $\|X_{s}\|_{p,
\infty}$ for $p\geq 1$ as $\mathbb{E}\big[\sup_{0\leq s\leq
1}|X_{s}|^{p}\big]^{\frac{1}{p}}$.     By the triangle inequality
and  by Jensen's inequality for the first term on the right
\begin{multline}\label{Kremblin}
\left \| \sum_{m=1}^{N_{\mathbf{s}_{1}}(st) }
\Big(Y_{\mathbf{s}_{1},\,\mathbf{s}_{2}}(m)-  c_{\mathbf{s}_{1},\mathbf{s}_{2}}(t)\Big) \right \|_{1,\infty}   \leq \left \| \sum_{j=1}^{N_{\mathbf{s}_{1}}(st) }\Big(Y_{\mathbf{s}_{1},\,\mathbf{s}_{2}}(m)- \mathbb{E}\big[Y_{\mathbf{s}_{1},\,\mathbf{s}_{2}}(m)|\mathcal{F}_{\theta_{j-1}} ]\Big) \right \|_{2,\infty}\\ +  \left \|  \sum_{j=1}^{N_{\mathbf{s}_{1}}(st) } \mathbb{E}\big[Y_{\mathbf{s}_{1},\,\mathbf{s}_{2}}(m)|\mathcal{F}_{\theta_{j-1}} ] - N_{\mathbf{s}_{1}}(st)\, c_{\mathbf{s}_{1},\mathbf{s}_{2}}(t)     \right \|_{1,\infty}
\end{multline}

  Since the information of
  previous incursions is contained in $\mathcal{F}_{\theta_{j-1}}$, the sum of the differences $Y_{\mathbf{s}_{1},\,\mathbf{s}_{2}}(m)-\mathbb{E}\big[Y_{\mathbf{s}_{1},\,\mathbf{s}_{2}}(m)|\mathcal{F}_{\theta_{j-1}} ]$ up to $m=N_{st}$ is a martingale.
By Doob's inequality and Lemma~\ref{HoldersInequality}
\begin{multline}
\left \| \sum_{m=1}^{N_{\mathbf{s}_{1}}(st) }\Big(Y_{\mathbf{s}_{1},\,\mathbf{s}_{2}}(m)- \mathbb{E}\big[Y_{\mathbf{s}_{1},\,\mathbf{s}_{2}}(m)|\mathcal{F}_{\theta_{j-1}} ]\Big) \right \|_{2,\infty}  \leq  \mathbb{E}\Big[\Big|\sum_{m=1}^{N_{\mathbf{s}_{1}}(st)}Y_{\mathbf{s}_{1},\,\mathbf{s}_{2}}(m)-\mathbb{E}\big[Y_{\mathbf{s}_{1},\,\mathbf{s}_{2}}(m)|\mathcal{F}_{\theta_{j-1}} ]     \Big|^{2}\Big]^{\frac{1}{2}}\\
\leq
\mathbb{E}\big[N_{\mathbf{s}_{1}}(st)\big]^{\frac{1}{2}}\,\sup_{m } \mathbb{E}\Big[\big|Y_{\mathbf{s}_{1},\,\mathbf{s}_{2}}(m)-\mathbb{E}[Y_{\mathbf{s}_{1},\,\mathbf{s}_{2}}(m)|\mathcal{F}_{\theta_{j-1}}] \big|^{2}\Big| m\leq N_{\mathbf{s}_{1}}(st)  \Big]^{\frac{1}{2}} \\ \leq  r_{2}^{\frac{1}{2}} t^{\frac{1}{8} }\sup_{m,\omega \in \mathcal{F}_{\sigma_{j}} }\mathbb{E}\Big[\big|Y_{\mathbf{s}_{1},\,\mathbf{s}_{2}}(m) \big|^{2}\Big|\mathcal{F}_{\sigma_{j}}  \Big]^{\frac{1}{2}}
\leq 5 r_{2}^{\frac{1}{2}}\,t^{\frac{1}{8} } \,
\end{multline}
where the last inequality uses Part (1), and the third inequality
uses Part (2) of Lemma~\ref{FiniteTimes}, and the following
\begin{multline}
 \mathbb{E}\Big[\big|Y_{\mathbf{s}_{1},\,\mathbf{s}_{2}}(m)-\mathbb{E}[Y_{\mathbf{s}_{1},\,\mathbf{s}_{2}}(m)|\mathcal{F}_{\theta_{j-1}}] \big|^{2}\Big| m\leq N_{\mathbf{s}_{1}}(st)  \Big]^{\frac{1}{2}}      \\ \leq \sup_{m,\omega \in \mathcal{F}_{\sigma_{j}} } \mathbb{E}\Big[\big|Y_{\mathbf{s}_{1},\,\mathbf{s}_{2} }(m)-\mathbb{E}[Y_{\mathbf{s}_{1},\,\mathbf{s}_{2}}(m)|\mathcal{F}_{\theta_{j-1}}]\big|^{2}\Big|\mathcal{F}_{\sigma_{j}}  \Big]^{\frac{1}{2}}\leq \sup_{m,\omega \in \mathcal{F}_{\sigma_{j}} }\mathbb{E}\Big[\big|Y_{\mathbf{s}_{1},\,\mathbf{s}_{2}}(m) \big|^{2}\Big|\mathcal{F}_{\sigma_{j}}  \Big]^{\frac{1}{2}}.
\end{multline}

Now we can work on the second term on the right
of~(\ref{Kremblin}).  By the triangle inequality and conditioning
 that $m\leq N_{\mathbf{s}_{1}}(st)$ for the terms in the sum as above,
\begin{multline*}
 \mathbb{E}\Big[\,\Big|\sum_{m=1}^{N_{\mathbf{s}_{1}}(st)}\mathbb{E}\big[Y_{\mathbf{s}_{1},\,\mathbf{s}_{2}}(m)|\mathcal{F}_{\theta_{j-1}} ]-N_{\mathbf{s}_{1} }(st)c_{\mathbf{s}_{1},\,\mathbf{s}_{2}}(t)     \Big|\Big] \\
 \leq \mathbb{E}\Big[\sum_{m=1}^{N_{\mathbf{s}_{1}}(st)}  \mathbb{E}\Big[\big| \mathbb{E}[Y_{\mathbf{s}_{1},\,\mathbf{s}_{2}}(m)|\mathcal{F}_{\theta_{j-1}}]-c_{\mathbf{s}_{1},\,\mathbf{s}_{2}}(t) \big|  \Big| m\leq N_{\mathbf{s}_{1}}(st)    \Big]  \Big]
\end{multline*}

However, we will split the terms in the sum on the right-side into the two groups $m\in  [N_{\mathbf{s}_{1}}(st-\varrho t)+1, N_{\mathbf{s}_{1}}(st)]$ and $m\in  [0,N_{\mathbf{s}_{1}}(st-\varrho t)]$ for some $0<\varrho \ll 1 $ and in particular $\varrho <s$.  For $m\in  [N_{\mathbf{s}_{1}}(st-\varrho t)+1, N_{\mathbf{s}_{1}}(st)]$
\begin{multline*}
 \mathbb{E}\Big[\sum_{m=N_{\mathbf{s}_{1}}(st-\varrho t)+1}^{N_{\mathbf{s}_{1}}(st)}  \mathbb{E}\Big[\big| \mathbb{E}[Y_{\mathbf{s}_{1},\,\mathbf{s}_{2}}(m)|\mathcal{F}_{\theta_{j-1}}]-c_{\mathbf{s}_{1},\,\mathbf{s}_{2}}(t) \big|  \Big| m\leq N_{\mathbf{s}_{1}}(st)    \Big]  \Big] \\ \leq
2\mathbb{E}\big[ N_{\mathbf{s}_{1}}(st)-N_{\mathbf{s}_{1}}(st-\varrho t ) \big]\,\sup_{\omega\in \mathcal{F}_{\sigma_{j}} } \mathbb{E}\big[|Y_{\mathbf{s}_{1},\,\mathbf{s}_{2}}(m)| \big|\mathcal{F}_{\sigma_{j}}\big]\leq 5\, r_{2}^{\frac{1}{2}}\varrho^{\frac{1}{2}}t^{\frac{1}{4}},
\end{multline*}
where the first inequality follows since $\mathbb{E}[Y_{\mathbf{s}_{1},\,\mathbf{s}_{2}}(m)|\mathcal{F}_{\theta_{j-1}}]$ and $c_{\mathbf{s}_{1},\,\mathbf{s}_{2}}(t)$ are convex combinations of values $\mathbb{E}\big[|Y_{\mathbf{s}_{1},\,\mathbf{s}_{2}}(m)| \big|\mathcal{F}_{\sigma_{j}}\big]$ for different $\omega \in \mathcal{F}_{\sigma_{j}}$.
The second inequality employs Part (1) and then Part (2) of Lemma~\ref{FiniteTimes} for an interval of length $\varrho t$.

For the sum of the terms with $m\in  [0,N_{\mathbf{s}_{1}}(st-\varrho t)]$, we need to better understand the expressions
\begin{align}\label{Fun}
\mathbb{E}\Big[\big|\mathbb{E}[Y_{\mathbf{s}_{1},\,\mathbf{s}_{2}}(m)|\mathcal{F}_{\theta_{j-1}}]-c_{\mathbf{s}_{1},\,\mathbf{s}_{2}}(t) \big|  \Big| m\leq N_{\mathbf{s}_{1}}(st)    \Big]  ,
\end{align}
and, in particular, how the information $m\leq N_{\mathbf{s}_{1}}(st)$ will change the expectation.  If $m\leq N_{\mathbf{s}_{1}}(st)$, then it will already be known at time $\theta_{j-1}\leq (s-\varrho)t$, that $m-1 \leq N_{\mathbf{s}_{1}}(st)$.
However, it was shown in the beginning of the proof of Corollary~\ref{BC4} that the probability of a jump into the region $[-t^{\frac{1}{4}},t^{\frac{1}{4}}]$ (which is the beginning of an incursion) after starting in any point $(x,k)$ with $|k|\in [2t^{\frac{1}{4}},4t^{\frac{1}{4}}]$ occurs with probability approaching one for $t\rightarrow \infty$.  Since, we have assumed jumps bounded by $t^{\frac{1}{40} }$, the point $(X_{\theta_{j-1}},K_{\theta_{j-1}})$ will have $|K_{\theta_{j-1}}|\leq [2t^{\frac{1}{4}},2t^{\frac{1}{4}}+t^{\frac{1}{40} }]$.  Thus knowing $ m\leq N_{\mathbf{s}_{1}}(st)$ will add little is known at time $\theta_{j}$. With this consideration, we can give a crude upper bound for the expression~(\ref{Fun}) by doubling the unconditioned value of the expectation of $|Y_{\mathbf{s}_{1},\,\mathbf{s}_{2}}(m)|\mathcal{F}_{\theta_{j-1}}]-c_{\mathbf{s}_{1},\,\mathbf{s}_{2}}(t)|$:
\begin{multline}
\mathbb{E}\Big[\big|\mathbb{E}[Y_{\mathbf{s}_{1},\,\mathbf{s}_{2}}(m)|\mathcal{F}_{\theta_{j-1}}]-c_{\mathbf{s}_{1},\,\mathbf{s}_{2}}(t) \big| \Big| N_{\mathbf{s}_{1}}(st)-N_{\mathbf{s}_{1}}(\theta_{j-1})>0  \Big] \\ < 2 \mathbb{E}\Big[\big|\mathbb{E}[Y_{\mathbf{s}_{1},\,\mathbf{s}_{2}}(m)|\mathcal{F}_{\theta_{j-1}}]-c_{\mathbf{s}_{1},\,\mathbf{s}_{2}}(t) \big|   \Big]
\end{multline}
where the inequality is due to the event $  N_{\mathbf{s}_{1}}(st)-N_{\mathbf{s}_{1}}(\theta_{j-1})>0$ having probability close to one by Corollary~\ref{BC4}.  Finally,
\begin{multline*}
\mathbb{E}\Big[ t^{-\frac{1}{4}}\sum_{m=1}^{N_{\mathbf{s}_{1}}(st-\varrho t)}  \mathbb{E}\Big[\big| \mathbb{E}[Y_{\mathbf{s}_{1},\,\mathbf{s}_{2}}(m)|\mathcal{F}_{\theta_{j-1}}]-c_{\mathbf{s}_{1},\,\mathbf{s}_{2}}(t)\big| \Big]\Big] \\ < 2t^{-\frac{1}{4}}\mathbb{E}\big[N_{\mathbf{s}_{1}}(st-\varrho t)   \big] \sup_{m} \mathbb{E}\Big[\big| \mathbb{E}[Y_{\mathbf{s}_{1},\,\mathbf{s}_{2}}(m)|\mathcal{F}_{\theta_{j-1}}]-c_{\mathbf{s}_{1},\,\mathbf{s}_{2}}(t)\big| \Big],
\end{multline*}
which goes to zero by Part (2) of this proposition and by Part (2) of Lemma~\ref{FiniteTimes}.

\end{proof}

\subsection{The torus reflection symmetry}\label{Torus}

Now we move on to results which are specific to having a reflection symmetry on the torus such that the potential $V(x)$ and jump densities $j_{x}(v)$ satisfy
$$V(x)=V(R(x))\quad \text{and} \quad j_{x}(v)=j_{R(x)}(v).$$
A consequence of this symmetry along with the symmetric jump rates
$j_{a}(v)= j_{a}(-v)$ between positive and negative momenta
is that a specific phase space
trajectory $(x_{s},\,k_{s})$, $s\in[0,t]$ from $(x_{0},\,k_{0})$
to $(x_{t},\,k_{t})$ will occur with the same density as a
trajectory $(R(x_{s}),\,-k_{s})$ from  $(R(x_{0}),\,-k_{0})$ to
$(R(x_{t}),\,-k_{t})$.  This combines with the time reversal
symmetry of the model to yield a ``time reversal in momentum''.

\begin{center}
\begin{tabular}{|p{6cm}|p{6cm}|}
\hline
\hspace{1.3cm}Original trajectory   & \hspace{1.3cm}Torus reflection  \\
\hspace{2cm} $(x_{s},\,k_{s})$ \hspace{1cm} & \hspace{1,3cm} $(R(x_{s}),\,-k_{s})$  \vspace{.1cm}\\
\hline
\hspace{.4cm} Time reversal symmetry & \hspace{.1cm} Time rev. with torus reflection \\
\hspace{1.7cm} $(x_{t-s},\,-k_{t-s})$ \hspace{1cm}         &  \hspace{1.3cm} $(R(x_{t-s}),\,k_{t-s})$      \vspace{.1cm}   \\
\hline
\end{tabular}

\end{center}

The significance of the symmetries for the study of the drift is that $\frac{dV}{dx}(x)=-\frac{dV}{dx}(R(x_{s}))$. Thus for the trajectories $(x_{s},\,k_{s})$ and $ (R(x_{t-s}),\,k_{t-s})$ over the interval $[0,t]$,
$$\int_{0}^{t}ds\,\frac{dV}{dx}(x_{s})= -\int_{0}^{t}ds\frac{dV}{dx}(R(x_{t-s})). $$
This will provide a basis for showing that there is no systematic bias for the contributions of the incursions.

\begin{lemma}\label{ExitProbTwo}
Assume List~\ref{AssumpOne} and List~\ref{AssumpTwo}. For the momentum-capped dynamics
$\frac{\rho_{+,-}(t)}{\rho_{-,+}(t)}-1$ tends to zero.

\end{lemma}

\begin{proof}
  By the torus reflection symmetry of the dynamics, the probability density of going from $(a, t^{\frac{1}{4}}-v)$ to $(a^{\prime},-2t^{\frac{1}{4}}-v^{\prime} )$ from the initial time to time $\theta$ is the same as the probability density of going from $(R(a), -t^{\frac{1}{4}}+v)$ to $(R(a^{\prime}),\,2t^{\frac{1}{4}}+v^{\prime} )$.  Thus
$$\textup{Pr}_{(a,\,t^{\frac{1}{4}}-v)}\big[ K_{\theta}< 0   \big]= \textup{Pr}_{(R(a),\,-t^{\frac{1}{4}}+v)}\big[ K_{\theta}> 0   \big].  $$
By the triangle inequality,
\begin{multline}\label{DayTrip}
\big|\frac{\rho_{+,-}(t)}{\rho_{-,+}(t)}-1 \big| \leq  \Big|\frac{\rho_{+,-}(t)}{\rho_{-,+}(t)}-\frac{\rho_{+,-}(t)}{\textup{Pr}_{(R(a),\,-t^{\frac{1}{4}}+v)}\big[ K_{\theta}> 0   \big]} \Big|+\Big|\frac{ \rho_{+,-}(t)}{ \textup{Pr}_{(a,\,t^{\frac{1}{4}}-v)}\big[ K_{\theta}< 0   \big]}-1 \Big|
\end{multline}
By Proposition~\ref{ExitProb},
$$\Big|\frac{\textup{Pr}_{(a,\,t^{\frac{1}{4}}-v)}\big[ K_{\theta}< 0   \big]}{\rho_{+,-} }-1\Big|,\quad \Big|\frac{\textup{Pr}_{(a,\,-t^{\frac{1}{4}}+v)}\big[ K_{\theta}>0\big]}{\rho_{-,+} }-1 \Big| \longrightarrow 0,     $$
for all $(a,v)\in [0,1]\times [0,t^{\frac{1}{40} }]$.  Thus~(\ref{DayTrip}) goes to zero.

\end{proof}

The following lemma constructs a specific joint distribution
$\phi_{t}^{*}(a,\,v\,;\,a^{\prime},\,v^{\prime}) $ for the first
entrance coordinates $(a,\,v)$ and last exit coordinates
$(a^{\prime},\,v^{\prime})$ for the set $S=\{(x,\,k)\in
[0,\,1]\times \R\,|\, |k|\leq t^{\frac{1}{4}}\} $ for trajectories
conditioned to begin with $s_{1}K_{0}> t^{\frac{1}{4}}$ and to end
with $s_{2}K_{\theta}> t^{\frac{1}{4}}$.  The
symmetry~(\ref{FirstLast}) will play a key role in the proof of
Proposition~\ref{UtSym}.

\begin{lemma}[Equilibrium first-entrance/last-exit distribution] \label{Equilibrium}
Assume List~\ref{AssumpOne} and List~\ref{AssumpTwo}.  Consider the dynamics conditioned to have jumps capped by $t^{\frac{1}{40}}$.  For large enough $t$, there exists a unique joint density $\phi_{t}^{*}(a,\,v;\,a^{\prime},\,v^{\prime}) $ with support in $([0,1]\times [0,\,t^{\frac{1}{40}}])^{2}$ such that the marginals
$$\phi_{\mathcal{F},t}^{*}(a,\,v)= \int_{[0,1]\times \R^{+}}da^{\prime}\,dv^{\prime}\,\phi_{t}^{*}(a,\,v;\,a^{\prime},\,v^{\prime})\, \text{ and }\, \phi_{\mathcal{L},t}^{*}(a^{\prime},\,v^{\prime})= \int_{[0,1]\times \R^{+}}da\,dv\,\phi_{t}^{*}(a,\,v;\,a^{\prime},\,v^{\prime}),   $$
satisfy the relations
\begin{eqnarray*}
\phi_{\mathcal{L},t}^{*}(a^{\prime},\,v^{\prime}) &=&   \int_{[0,1]\times \R^{+}}da\,dv\, \phi_{\mathcal{F},t}^{*}(a,\,v)\,\textup{Pr}_{(a,\,s_{1}\,t^{\frac{1}{4}}-s_{1}v)}\big[(X_{\varsigma},\,s_{2}\,K_{\varsigma})=(a^{\prime},\,t^{\frac{1}{4}}-v^{\prime} )\,\big|\,s_{2}\,K_{\varsigma}>0  \big] \\
\phi_{\mathcal{F},t}^{*}(a^{\prime},\,v^{\prime}) &=&   \int_{[0,1]\times \R^{+}}da\,dv\, \phi_{\mathcal{L},t}^{*}(a,\,\,v)\,\textup{Pr}_{(a,\,s_{1}t^{\frac{1}{4}}-s_{1}v)}^{(R)}\big[(X_{\varsigma}, s_{2}\,K_{\varsigma})=(a^{\prime},\,t^{\frac{1}{4}}-v^{\prime} )\,\big|\,s_{2}\,K_{\varsigma}>0,    \big]
\end{eqnarray*}
where $\textup{Pr}_{(x,\,k)}^{(R)}$ refers to the law of the time-reversed dynamics starting from the point $(x,\,k)\in [0,\,1]\times \R$.

Moreover, $\phi_{t}^{*}$ has the symmetry
\begin{align}\label{FirstLast}
\phi_{t}^{*}(a,\,v\,;\,a^{\prime},\,v^{\prime})=\phi_{t}^{*}(R(a^{\prime}),\,v^{\prime};\,R(a),\,v).
\end{align}

\end{lemma}

\begin{proof}
Set $s_{1}=s_{2}=+$.  We pick a number in $(1,2)$, say, $\frac{3}{2}$. In the event that $K_{\sigma}>0$, define  $\varsigma^{\prime}$ to be the last time the particle has a jump with momentum $<\frac{3}{2}\,t^{\frac{1}{4}}$ before it continues on to reach $2\,t^{\frac{1}{2}}$ (at time $\theta$),
\begin{align}
\varsigma^{\prime}= \textup{inf}\{s\in [0,\,\theta],\, M_{s}-M_{s^{-}}\neq 0  \, \big|\,  \inf_{s< r\leq \theta } |K_{r}| > \frac{3}{2}\,t^{\frac{1}{4}} \}.
\end{align}
Clearly $\varsigma^{\prime}>\varsigma$.

Consider the two maps $\Psi^{F},\Psi^{R}:L^{1}([0,1]\times \R^{+})$
\begin{eqnarray*}
\Psi^{F}(\varphi)(a,v)&=&\int da^{\prime}\,dv^{\prime}\, \varphi(a,v)\,\textup{Pr}_{(a,\,\frac{3}{2}t^{\frac{1}{4}}-v)}\big[(X_{\varsigma^{\prime}},\,K_{\varsigma^{\prime}})=(a^{\prime},\,\frac{3}{2}t^{\frac{1}{4}}-v^{\prime} )\,\big|\,K_{\varsigma}>0,\, \inf_{0\leq r\leq \theta }K_{r}<t^{\frac{1}{4}}  \big] \\
\Psi^{R}(\varphi)(a,\,v)&=&\int da^{\prime}\,dv^{\prime}  \, \varphi(a,v)\,\textup{Pr}_{(a,\,\frac{3}{2}t^{\frac{1}{4}}-v)}^{(R)}\big[(X_{\varsigma},\,K_{\varsigma})=(a^{\prime},\,\frac{3}{2}t^{\frac{1}{4}}-v^{\prime} )\,\big|\,K_{\varsigma}>0,\, \inf_{0\leq r\leq \theta }K_{r}<t^{\frac{1}{4}}    \big],
\end{eqnarray*}
where $\textup{Pr}_{(x,\,k)}^{(R)}$ refers to the statistics for the time-reversed Markov dynamics starting from the point $(x,\,k)$.

$\Psi^{F}$ and $\Psi^{R}$ send probability densities to
probability densities, and for large enough $t$, we claim that
$\Psi^{F},\,\Psi^{R}$ are contractive on differences of densities.
Consider the hitting time $\tau= \inf\{s\in [0,\theta],\,
M_{s}-M_{s^{-}}\neq 0 \,\big| \, K_{s}\notin S \} $ for the set
$S=\{(x,\,k)\in [0,1]\times \R
\,|\,|k-\frac{3}{2}t^{\frac{1}{4}}|< t^{\frac{1}{20} } \} $.

Since $\tau<\varsigma^{\prime}$, by the Markov property,
\begin{multline}\label{Contract}
\Psi^{F}(\varphi)(a,\,v)=\int_{[0,1]\times \R^{+}} da^{\prime}\,dv^{\prime}\varphi(a^{\prime},\,v^{\prime})\\ \times\int_{[0,1]\times \R^{+}}dx\,dk\, \big(\psi_{\uparrow,\,t}^{(a^{\prime},\,v^{\prime})}(x,\,k-\frac{3}{2}t^{\frac{1}{4}}-t^{\frac{1}{20}})+\psi_{\downarrow,\,t}^{(a^{\prime},\,v^{\prime})}( x,\,-k+\frac{3}{2}t^{\frac{1}{4}}-t^{\frac{1}{20}})\big)\\
\times \textup{Pr}_{(x,\,k)}\big[(X_{\varsigma^{\prime}},\,K_{\varsigma^{\prime}})=(a,\, \frac{3}{2}t^{\frac{1}{4}}-v )\,\big|\,K_{\varsigma}>0,\,\inf_{0\leq r\leq \theta }K_{r}<t^{\frac{1}{4}}   \big]
\end{multline}
$\psi_{\uparrow,\,t}^{(a^{\prime},\,v^{\prime})}$ and
$\psi_{\downarrow,\,t}^{(a^{\prime},\,v^{\prime})}$ are the
boundary crossing densities for the set $S$ for paths starting
from $(a^{\prime},\,\frac{3}{2}t^{\frac{1}{4}}-v^{\prime})$ and
conditioned to reach below $t^{\frac{1}{4}}$ before going above
$t^{\frac{1}{4}}$.
  Using~(\ref{Contract}), for two probability densities $\varphi_{1},\,\varphi_{2}$,
\begin{multline}
\|\Psi^{F}(\varphi_{1}-\varphi_{2})\|_{1}\leq  \|\varphi_{1}-\varphi_{2}\|_{1}\\ \times \sup_{a,\,a^{\prime}\in [0,1],\, v\,v^{\prime}\in [0,\,t^{\frac{1}{40} }] }\big( \|\psi_{
\uparrow,\, t}^{(a,\,v)}-\psi_{\uparrow,\,t}^{(a^{\prime},\,v^{\prime})}\|_{1}+   \|\phi_{\downarrow,\, t}^{(a,\,v)}-\phi_{\downarrow,\,t}^{(a^{\prime},\,v^{\prime})}\|_{1}\big).
\end{multline}
 However, by Proposition~\ref{BC3NHalf} we have that
$$\sup_{a\in [0,1],\, v\,v^{\prime}\in [0,\,t^{\frac{1}{40} }]  } \|\psi_{\uparrow,\, t}^{(a,\,v)}-\frac{1}{2}\phi_{\infty}\|_{1}\longrightarrow 0\quad\text{ and }\quad \sup_{a\in [0,1],\, v\in [0,\,t^{\frac{1}{40} }]} \|\phi_{\downarrow,\,t}^{(a,\,v)}-\frac{1}{2}\phi_{\infty}\|_{1}\longrightarrow 0.$$

Thus for large enough $t$ there exists a constant $0<\lambda<1$ such that
$$\|\Psi^{F}(\varphi_{1}-\varphi_{2})\|_{1}\leq \lambda\|\varphi_{1}-\varphi_{2}\|_{1},$$
for any two probability densities $\varphi_{1},\,\varphi_{2}\in L^{1}([0,1]\times \R^{+})$.  For large enough $t$, the constant $\lambda>0$ can be made arbitrarily small. By symmetry, the same proof holds for $\Psi^{R}$.

For $t$ large enough for the strict contractive property above, we
can now construct special densities by defining the limits
$$ \pi^{F}_{t}=\lim_{n\rightarrow \infty}(\Psi^{R}\Psi^{F})^{n}(P)\quad \text{and} \quad \pi^{R}_{t}= \lim_{n\rightarrow \infty}(\Psi^{F}\Psi^{R})^{n}(P),  $$
where the limit is independent of the probability density $P$.  We have constructed equilibrium states in which we can do computations using the Markov property for both the original dynamics and the time-reversed dynamics.

We thus define the first entrance distribution
\begin{multline}\label{ForEqu}
\phi^{*}_{\mathcal{F},t}(a,\,v)= \int_{[0,1]\times \R^{+}} da^{\prime}\,dv^{\prime}\,\pi^{F}_{t}(a^{\prime},\,v^{\prime})\\ \textup{Pr}_{(a,\,t^{\frac{1}{4}}-v)}\big[(X_{\tau},\,K_{\tau})=(a^{\prime},\,t^{\frac{1}{4}}-v^{\prime} )\,\big|\,K_{\theta}>0\,\inf_{0\leq r\leq \theta }K_{r}<t^{\frac{1}{4}}   \big],
\end{multline}
and finally $\phi_{t}^{*}(a,\,v\,;\,a^{\prime},\,v^{\prime})$ as the product
$$\phi_{t}^{*}(a,\,v\,;\,a^{\prime},\,v^{\prime})=\phi^{*}_{\mathcal{F},t}(a,\,v)\textup{Pr}_{(a,\,t^{\frac{1}{4}}-v)}\big[(X_{\varsigma},\,K_{\varsigma})=(a^{\prime},\,t^{\frac{1}{4}}-v^{\prime} )\,\big|\,K_{\varsigma}>0  \big].  $$

The first relation in the statement of the lemma, which determines
$\phi^{*}_{\mathcal{L},t}(a,\,v)$ with
$\phi^{*}_{\mathcal{F},t}(a,\,v)$ using the forward dynamics,
follows immediately from the definition of $\phi_{t}^{*}$ .  In
the above formula,
$\phi^{*}_{\mathcal{L},t}(a^{\prime},\,v^{\prime})$ is determined
as the last exit time of $S=\{(x,\,k)\in [0,1]\times \R\,|\,
|k|\leq t^{\frac{1}{4}}\} $.  However, starting from the points
$(q,\,\frac{3}{2}t^{\frac{1}{4}}-p)$ with distribution
$\pi^{R}_{t}(q,\,p)$ in the time-reversed dynamics, then
$\phi^{*}_{\mathcal{L},t}(a,\,v)$ is the first entrance
distribution for set $S$. The second relation then follows from
the Markov property for the time-reversed dynamics
\begin{align}
\phi^{*}_{\mathcal{F},t}(a,\,v)=\int da^{\prime},\,dv^{\prime}\phi^{*}_{\mathcal{L},t}(a^{\prime},\,v^{\prime})\,  \textup{Pr}_{(a,\,t^{\frac{1}{4}}-v)}^{(R)}\big[(X_{\varsigma}, \,K_{\varsigma})=(a^{\prime},\,t^{\frac{1}{4}}-v^{\prime} )\,\big|\,K_{\varsigma}>0,    \big].
\end{align}

\end{proof}

\begin{proposition}[Antisymmetry of constants]\label{UtSym}
Assume List~\ref{AssumpOne} and List~\ref{AssumpTwo}. The constants
$c_{+,\,+}(t)$, $c_{-,\,-}(t)$, $c_{+,\,-}(t)+c_{-,\,+}(t)$ tend to zero for large times.

\end{proposition}

\begin{proof}
Let us fix $\mathbf{s}_{1}=\mathbf{s}_{2}=+$.  By
Lemma~\ref{CappedJumps}, we can take the dynamics conditioned to
make jumps capped by $t^{\frac{1}{40}}$.

  Define the functional $\Lambda
:L^{1}\big([0,1]\times \R^{+}\big)\rightarrow \R$,
\begin{align}\label{PreBigZero}
\Lambda(\varphi)=\int_{[0,1]\times [0,t^{\frac{1}{40} }]}  da\,dv\,\varphi(a,v)\,\mathbb{E}_{(a,\,t^{\frac{1}{4}}-v )}\big[ t^{-\frac{1}{4}}\int_{0}^{\varsigma}ds\,\frac{dV}{dx}(X_{s}) \,  \chi( K_{\varsigma}>0)   \big].
\end{align}
By our comments above $c_{+,\,+}\approx \Lambda(\phi_{\infty})$, and when $\varphi=\phi_{\mathcal{F},t}^{*}$ then
\begin{multline}\label{BigZero}
\Lambda( \phi_{\mathcal{F},t}^{*}) = \Pr[ K_{\theta}>0  ] \,  \int da\,dv\,da^{\prime}\,dv^{\prime}\,\phi_{t}^{*}(a,\,v\,;\,a^{\prime},\,v^{\prime})\,\mathbb{E}_{(a,\,t^{\frac{1}{4}}-v )}^{(a^{\prime},\,t^{\frac{1}{4}}-v^{\prime} )}\big[ t^{-\frac{1}{4}}\int_{0}^{\varsigma}ds\,\frac{dV}{dx}(X_{s})\big],
\end{multline}
where $\mathbb{E}_{(x,\,k )}^{(x^{\prime},\,k^{\prime} )}$ is the expectation conditioned on the trajectories that begin at $(x,\,k )$ and have a last exit $(x^{\prime},\,k^{\prime} )$ at the time $\varsigma$.

Note the anti-symmetry
\begin{align}\label{Antisymmetry}
\mathbb{E}_{(a,\,t^{\frac{1}{4}}-v )}^{(a^{\prime},\,t^{\frac{1}{4}}-v^{\prime} )}\big[ t^{-\frac{1}{4} }\int_{0}^{\varsigma}ds\,\frac{dV}{dx}(X_{s})\big]= -\mathbb{E}_{(R(a^{\prime}),\,t^{\frac{1}{4}}-v^{\prime} )}^{(R(a),\,t^{\frac{1}{4}}-v )}\big[ t^{-\frac{1}{4}}\int_{0}^{\varsigma}ds\,\frac{dV}{dx}(X_{s})\big],
\end{align}
which is due to the ``time-reversal in momentum'' mentioned at the
beginning of the section in which for every trajectory
$(x_{t},\,k_{t})$ on the interval $[0,t]$, there is a backwards
trajectory with a torus-reflected position $(R(x_{t-s}
),\,k_{t-s})$ which occurs in the forward dynamics with the same
``probability'' as a fraction of the trajectories that begin at
$(x_{0},\,k_{0})$ and $(x_{t},\,k_{t})$ respectively.  Even though
the final time $\varsigma$ is not deterministic, the two
trajectories are still weighted equally in the
expectations~(\ref{Antisymmetry}).  This follows since $\varsigma$
is a hitting time for the time-reversed Markov process (for when
the momentum first jumps below $t^{\frac{1}{4}})$.   Due
to~(\ref{Antisymmetry}) and the symmetry~(\ref{FirstLast}) of
$\phi_{t}^{*}(a,\,v\,;\,a^{\prime},\,v^{\prime})$, it follows
that~(\ref{BigZero}) is zero.

We now focus on showing that due to a dynamical loss of memory of
the initial conditions over the time interval $[0,\sigma]$ , the
values of $\Lambda(\varphi)$ for any $\varphi$ with support in
$[0,1]\times [0,\,t^{\frac{1}{40} } ]$ are close.  In
particular, $\Lambda(\varphi)$ for
$\varphi=\bar{\phi}_{\infty}=\phi_{\infty}\chi(|v|\leq
t^{\frac{1}{40}}) $ or $\varphi=\phi_{\mathcal{F},t}^{*}$ are close,
which would prove the result.

 Define $\tau$ to be the exit time for the set $S=\{(x,\,k)\in [0,1]\times \R\,\big| \,|k-t^{\frac{1}{4}}|< t^{\frac{1}{20}}\} $.    Thus for a probability density $\varphi$,
 \begin{multline}\label{HereNoMore}
\Lambda(\varphi)= \int_{[0,1]\times \R^{+}}da \,dv\, \varphi(a,\,v)\,\mathbb{E}_{(a,\,t^{\frac{1}{4}}-v)}  \big[t^{-\frac{1}{4}}\int_{0}^{\tau \wedge \sigma}ds\,\frac{dV}{dx}(X_{s})\, \chi\big( K_{\varsigma}>0 \big)  \big]\\ +   \int dx\,dk\, \big(\phi_{\uparrow ,t}^{\varphi}( x,\,k-t^{\frac{1}{4}}-t^{\frac{1}{20}})+\phi_{\downarrow,\,t}^{\varphi}( x,\,-k+t^{\frac{1}{4}}-t^{\frac{1}{20}})\big)\\
\times \mathbb{E}_{(x,k)}\big[t^{-\frac{1}{4}}\int_{0}^{\sigma}ds\,\frac{dV}{dx}(X_{s})\,\chi\big( K_{\varsigma}>0 \big)  \big].
\end{multline}
where $\phi_{\uparrow ,t}^{\varphi}=\int dq\, dp\,\varphi(q,\,p)\,\phi_{\uparrow ,t}^{(q,p)}$ and an analogous definition for $\phi_{\downarrow ,t}^{\varphi}$.

We argue that the first term on the right-side
of~(\ref{HereNoMore}) tends to zero for large $t$.  Since the
momentum is greater than $\frac{1}{2}t^{\frac{1}{4}}$ up to time
$\tau$, then by Lemma~\ref{RealBasics}
$$\big| \mathbb{E}_{(a,\,t^{\frac{1}{4}}-v)}  \big[t^{-\frac{1}{4}}\int_{0}^{\tau \wedge \sigma}ds\,\frac{dV}{dx}(X_{s})   \big]\big| \leq \frac{4\bar{V}}{t^{\frac{1}{4}}}\, \mathbb{E}_{(a,\,t^{\frac{1}{4}}-v)}\big[N_{\tau}\big],$$
where $N_{\tau}$ is the number of momentum jumps up to time
$\tau$.  In the proof of~Proposition~\ref{BC2NHalf}, it was shown
that $\mathbb{E}_{(x,\,k)}\big[N\big]= \mathit{O}( t^{\frac{1}{10}
})$.  Thus the drift up to time $\tau\wedge \sigma$ vanishes for
large $t$ vanishes as $\mathit{O}(t^{-\frac{3}{20}})$.

Thus
\begin{multline*}
\big|\Lambda\big(\phi \big)-\Lambda\big(\phi_{\mathcal{F},t }^{*}\big)\big|\leq  \mathit{O}(t^{-\frac{3}{20}})+  \big(\|\phi_{\uparrow,t}^{\varphi}|_{\varphi=\bar{\phi}_{\infty}}-\phi_{\uparrow,t}^{\varphi}|_{\varphi= \phi_{\mathcal{F},t}^{*}}\|_{1}  + \|\phi_{\downarrow,\,t}^{\varphi}|_{\varphi=\bar{\phi}_{\infty}}-\phi_{\downarrow,\,t}^{\varphi}|_{\varphi= \phi_{\mathcal{F},t}^{*}}\|_{1}\big) \\
\times \sup_{(x,\,k)}  \mathbb{E}_{(x,k)}\big[ \big| t^{-\frac{1}{4}}\int_{0}^{\sigma}ds\,\frac{dV}{dx}(X_{s})\,  \chi \big( K_{\varsigma}>0\big) \big| \big] .
\end{multline*}
However, by Proposition~\ref{BC3}, $\phi_{\uparrow,\,t}^{\varphi},\,\phi_{\downarrow,\,t}^{\varphi}\rightarrow \varphi_{\infty}$, since by definition $\phi_{\uparrow,\,t}^{\varphi}$ is a convex combination of $\phi_{\uparrow,\,t}^{(q,p)}$ for $(a,p)$ with $ p\in [0,\,t^{\frac{1}{40} }]$.

$\phi_{\uparrow,t}^{\varphi}|_{\varphi=\phi_{\infty}}$, $\phi_{\uparrow,t}^{\varphi}|_{\varphi= \phi_{\mathcal{F},t}^{*}}$, $\phi_{\downarrow,\,t}^{\varphi}|_{\varphi=\phi_{\infty}}$, and $\phi_{\downarrow,\,t}^{\varphi}|_{\varphi= \phi_{\mathcal{F},t}^{*}}$ tend to $\phi_{\infty}$ in $L^{1}$.  Moreover, by the same argument as for Part (1) of Proposition~\ref{OddsEnds},
$$\sup_{(x,k)}\Big|\mathbb{E}_{(x,k)}\big[\big|t^{-\frac{1}{4}}\int_{0}^{\sigma}ds\,\frac{dV}{dx}(X_{s}) \big|^{2} \chi(K_{\theta}>0 ) \big]\Big|^{\frac{1}{2}}\leq 5.$$
We then have that $\big|\Lambda\big(\phi \big)-\Lambda\big(\phi_{\mathcal{F},t}^{*}(a,\,v)\big)\big|$
converges to zero for $t\rightarrow \infty$ which proves the result.  That shows the $c_{+,+}$ case.  The other cases for $c_{-,-}$, and $c_{+,-}+c_{-,+}$ are similar.

\end{proof}

\begin{theorem}\label{BigDriftLem}
Assume List~\ref{AssumpOne} and List~\ref{AssumpTwo}.  In the limit $t\rightarrow \infty$,
\begin{align}
\mathbb{E}\Big[  \sup_{0\leq s\leq 1} \Big| t^{-\frac{1}{2}}\int_{0}^{st}dr\,\frac{dV}{dx}(X_{r})\Big|\Big] \longrightarrow 0.
\end{align}

\end{theorem}

\begin{proof}
Our basic idea is to break  the integral
$t^{-\frac{1}{2}}\int_{0}^{st}dr\,\frac{dV}{dx}(X_{r})$ into parts
where corresponding to where $|K_{r}|$ is high and low energy respectively.
The low energy parts are controlled by our study of the random
variables
$Y_{n}=\int_{\sigma_{n}}^{\varsigma_{n}}dr\,\frac{dV}{dx}(X_{r})$
and the high energy parts will be controlled by
Lemma~\ref{HEDrift}.
\begin{multline}\label{PreBreakDown}
 t^{-\frac{1}{2}}\int_{0}^{st}dr\,\frac{dV}{dx}(X_{r})= t^{-\frac{1}{2}}\int_{0}^{\theta_{0}}dr\,\frac{dV}{dx}(X_{r})+ t^{-\frac{1}{2}}\int_{0}^{st}dr\,\frac{dV}{dx}(X_{r})\chi(|K_{\tau(r)}|>t^{\frac{1}{4}} )\\ + t^{-\frac{1}{4}}\sum_{n=1}^{N_{Y}(st)}Y_{n}-\int_{\sigma_{n}}^{\varsigma_{n}}dr\,\frac{dV}{dx}(X_{r})\chi(|K_{\tau(r)}|>t^{\frac{1}{4}} ) \\
 -\chi\big(\exists(j):  st \in [\sigma_{j},\varsigma_{j}] \big)\, \int^{\theta_{N_{Y}(st)}}_{st}dr\,\frac{dV}{dx}(X_{r})
\end{multline}
where $\tau(r)$ the next to last jump time before time $r$ as in Lemma~\ref{HEDrift}, and the last term corresponds removing an overlap due to a last incomplete incursion which begins before

For the first term on the right side of~(\ref{PreBreakDown}),
an argument analogous to Part (1) of Proposition~\ref{OddsEnds} gives the bound,
\begin{align}\label{Stray}
 \mathbb{E}\big[ \big|\int_{0}^{\theta_{0}}dr\,\frac{dV}{dx}(X_{r})\big|^{2}    \big]^{\frac{1}{2}}\leq 5\,t^{\frac{1}{4}},
\end{align}
and thus that term is negligible.  In this case, the end time $\gamma$ is a hitting time, which makes the argument easier.  The last term has the same bound.  By Lemma~\ref{HEDrift},
$$\mathbb{E}\big[\sup_{0\leq s\leq 1}\big|t^{-\frac{1}{2}}\int_{0}^{st}dr\,\frac{dV}{dx}(X_{r})\chi(|K_{\tau(r)}|>t^{\frac{1}{4}} )\big|^{2}\big]^{\frac{1}{2}}$$
 converges to zero.  The same argument as in the proof of Lemma~\ref{HEDrift} shows that
$$\mathbb{E}\Big[\sup_{0\leq s\leq 1}\Big|t^{-\frac{1}{2}}\sum_{n=1}^{N_{Y}(st)}\int_{\sigma_{n}}^{\varsigma_{n}}dr\,\frac{dV}{dx}(X_{r})\chi\big(|K_{\tau(r)}|>t^{\frac{1}{4}} \big)\Big|^{2}  \Big]^{\frac{1}{2}}\longrightarrow 0 ,  $$
since it includes even less terms.

We are left with the sequence
$t^{-\frac{1}{4}}\sum_{n=1}^{N_{Y}(st)-1}Y_{n}$,
which we can write
as
$$t^{-\frac{1}{4}}\sum_{m=1}^{N_{+}(st)}Y_{+,+}(m)+Y_{+,-}(m)+ t^{-\frac{1}{4}}\sum_{m=1}^{N_{-}(st)}Y_{-,-}(m)+Y_{-,+}(m).$$
 By Part~(4) of Proposition~\ref{OddsEnds} these sums
  can be approximated by $t^{-\frac{1}{4}}N_{\mathbf{s}_{1}}(st)\,c_{\mathbf{s}_{1},\,\mathbf{s}_{2}}(t)$:
$$\mathbb{E}\Big[ \sup_{0\leq s\leq 1}\Big| t^{-\frac{1}{4}}\sum_{j=1}^{N_{\mathbf{s}_{1}}(st) }Y_{\mathbf{s}_{1},\,\mathbf{s}_{2}}(m)- t^{-\frac{1}{4}}\,N_{\mathbf{s}_{1}}(st)\, c_{\mathbf{s}_{1},\mathbf{s}_{2}}(t)     \Big|\Big]\longrightarrow 0.   $$
For the sequences with $\mathbf{s}_{1} =\mathbf{s}_{2}=\mathbf{s}$
$$\mathbb{E}\big[\sup_{0\leq s\leq 1}\big| t^{-\frac{1}{4}}N_{\mathbf{s},\,\mathbf{s}}(st)\,c_{\mathbf{s},\,\mathbf{s}_{1}}(t)\big|\big]=\mathbb{E}\big[\big| t^{-\frac{1}{4}}N_{\mathbf{s},\,\mathbf{s}}(t)\,c_{\mathbf{s},\,\mathbf{s}}(t)\big| \big]\leq r_{2}^{\frac{1}{2}}\,|c_{\mathbf{s},\,\mathbf{s}}(t)|, $$
where the inequality uses Part (2) of Lemma~\ref{FiniteTimes}.  By Proposition~\ref{UtSym}, $c_{\mathbf{s},\,\mathbf{s}}(t)$ converges to zero.

The  cases of $(\mathbf{s}_{1},\mathbf{s}_{2})=(+,-)$  and $(\mathbf{s}_{1},\mathbf{s}_{2})=(-,+)$ must be treated together.  We will take a step backward and approximate $t^{-\frac{1}{4}}N_{\mathbf{s}_{1}}(st)\,c_{\mathbf{s}_{1},\,\mathbf{s}_{2}}(t)$ with $t^{-\frac{1}{4}}\big(\rho_{\mathbf{s}_{1},\mathbf{s}_{2}}(t)\big)^{-1} \sum_{n=1}^{N_{\mathbf{s}_{1}}(st) } c_{\mathbf{s}_{1},\,\mathbf{s}_{2}}(t)\,\chi(\mathbf{s}_{2}K_{\varsigma_{j}}>0)$.  By the triangle inequality
\begin{multline}\label{Cops}
\mathbb{E}\big[\sup_{0\leq s\leq 1}\big| t^{-\frac{1}{4}}\big(\rho_{\mathbf{s}_{1},\mathbf{s}_{2}}(t)\big)^{-1}\sum_{n=1}^{N_{\mathbf{s}_{1}}(st) }c_{\mathbf{s}_{1},\,\mathbf{s}_{2}}(t)\,\chi(\mathbf{s}_{2}K_{\varsigma_{j}}>0)-t^{-\frac{1}{4}}N_{\mathbf{s}_{1}}(st)\,c_{\mathbf{s}_{1},\,\mathbf{s}_{2}}(t)\big|\big] \\ \leq   t^{-\frac{1}{4}}\big(\rho_{\mathbf{s}_{1},\mathbf{s}_{2}}(t)\big)^{-1}\big|c_{\mathbf{s}_{1},\,\mathbf{s}_{2}}(t)\big|\,\mathbb{E}\big[\sup_{0\leq s\leq 1}\big| \sum_{n=1}^{N_{\mathbf{s}_{1}}(st) }\,\chi(\mathbf{s}_{2}K_{\varsigma_{j}}>0)-\Pr[\mathbf{s}_{2}K_{\sigma_{j}} >0 ] \big|\big]\\ +t^{-\frac{1}{4}}(\rho_{\mathbf{s}_{1},\mathbf{s}_{2}}(t)\big)^{-1}c_{\mathbf{s}_{1},\,\mathbf{s}_{2}}(t)\,\mathbb{E}\big[\sup_{0\leq s\leq 1}\big| \sum_{n=1}^{N_{\mathbf{s}_{1}}(st) }\,\Pr[\mathbf{s}_{2}K_{\varsigma_{j}}>0]-\rho_{s_{1},s_{2}}(t) \big|\big].
\end{multline}
The sum $\sum_{n=1}^{N_{\mathbf{s}_{1}}(st) }\,\chi(\mathbf{s}_{2}K_{\varsigma_{j}}>0)-\Pr[\mathbf{s}_{2}K_{\varsigma_{j}}>0 ]$ is a martingale, so by Doob's maximal inequality and by Lemma~\ref{HoldersInequality}
\begin{multline}
\mathbb{E}\Big[\sup_{0\leq s\leq 1}\Big| \sum_{n=1}^{N_{\mathbf{s}_{1}}(st) }\,\chi(\mathbf{s}_{2}K_{\varsigma_{j}}>0)-\Pr[\mathbf{s}_{2}K_{\varsigma_{j}}>0] \Big|^{2}\Big]^{\frac{1}{2}}   \leq 2 \mathbb{E}\Big[\Big| \sum_{n=1}^{N_{\mathbf{s}_{1}}(t) }\,\chi(\mathbf{s}_{2}K_{\varsigma_{j}}>0)-\Pr[\mathbf{s}_{2}K_{\varsigma_{j}}>0] \Big|^{2}\Big]^{\frac{1}{2}}
\\  \leq \sup_{j}\mathbb{E}\big[\big(\chi(\mathbf{s}_{2}K_{\varsigma_{j}}>0) -\Pr[\mathbf{s}_{2}K_{\varsigma_{j}}>0]\big)^{2}\big| j\leq N_{\mathbf{s}_{1}}(st) \big]^{\frac{1}{2}} \,  \mathbb{E}\big[ N_{\mathbf{s}_{1}}(t)\big]^{\frac{1}{2}}\\ \leq \sup_{j,s}\Big|\Pr[\mathbf{s}_{2}K_{\varsigma_{j}}>0|j\leq N_{\mathbf{s}_{1}}(st)]-   \Pr[\mathbf{s}_{2}K_{\varsigma_{j}}>0] \Big|  \, r_{2}t^{\frac{1}{8}},
\end{multline}
where the last inequality follows since $ N_{\mathbf{s}_{1}}(t)< N_{Y}(t)$ and from Part (2) of Proposition~\ref{FiniteTimes} and by an explicit calculation for the expectation of the indicator in the variance-type formula. Since the factor of $t^{-\frac{1}{4}}$ in~(\ref{Cops}) over-powers the factor $t^{\frac{1}{8}}$, the only worry is $\big(\rho_{\mathbf{s}_{1},\mathbf{s}_{2}}(t)\big)^{-1}$ going to zero.
 Since the event $j\leq N_{\mathbf{s}_{1}}(st)$ is adapted to the information known up to time $\sigma_{j}$,
\begin{multline}
\big(\rho_{\mathbf{s}_{1},\mathbf{s}_{2}}(t)\big)^{-1}\sup_{j}\Big|\Pr[\mathbf{s}_{2}K_{\varsigma_{j}}>0|j\leq N_{\mathbf{s}_{1}}(st)]-   \Pr[\mathbf{s}_{2}K_{\varsigma_{j}}>0] \Big| \\ \leq  \sup_{j,\,\omega\in \mathcal{F}_{\sigma_{j}}}\Big|\frac{ \Pr[\mathbf{s}_{2}K_{\varsigma_{j}}>0|\mathcal{F}_{\sigma_{j}}]}{\rho_{\mathbf{s}_{1},\mathbf{s}_{2}}(t)   }-   \frac{\Pr[\mathbf{s}_{2}K_{\varsigma_{j}}>0]}{\rho_{\mathbf{s}_{1},\mathbf{s}_{2}}(t) } \Big|.
\end{multline}
By adding and subtracting $1$ in the expression on the right-side, then by two applications of Lemma~\ref{ExitProb}, which is permitted by our assumption on the boundedness of the jumps, shows that the above goes to zero.  For the application of Lemma~\ref{ExitProb}, note that by the definitions $\theta_{j}$ and $\varsigma_{j}$ that  $\mathbf{s}_{2}K_{\varsigma_{j}}>0$ is equivalent to $\mathbf{s}_{2}K_{\theta_{j}}>0$.  Also note that
$\Pr[\mathbf{s}_{2}K_{\varsigma_{j}}>0]$ is a convex combination of the probabilities $\Pr[\mathbf{s}_{2}K_{\varsigma_{j}}>0|(X_{\sigma_{j}},K_{\sigma_{j}} )]  $.
Due to the decay of $t^{-\frac{1}{8}}$, we had only needed this term to be bounded, but we apply these principles again below.

For the second term in~(\ref{Cops}),
$$t^{-\frac{1}{4}} \mathbb{E}\Big[\sup_{0\leq s\leq 1}\Big| \sum_{n=1}^{N_{\mathbf{s}_{1}}(st) }\,\frac{\Pr\big[\mathbf{s}_{2}K_{\varsigma_{j}}>0]}{\rho_{\mathbf{s}_{1},\mathbf{s}_{2}}(t)}-1 \Big|\Big]\leq t^{-\frac{1}{4}}\sup_{j} \Big|\frac{ \Pr[\mathbf{s}_{2}K_{\varsigma_{j}}>0 ]}{\rho_{\mathbf{s}_{1},\mathbf{s}_{2}}(t)}-1 \Big|\,\mathbb{E}[N_{Y}].    $$
We apply Part (2) of Lemma~\ref{FiniteTimes}  to show that
$t^{-\frac{1}{4}}\, \mathbb{E}[N_{Y}]\leq r_{2}^{\frac{1}{2}}$ is
bounded and Lemma~\ref{ExitProb} to show that the ratio of
probabilities converges to one.

Now we just need to bound
\begin{multline}
t^{-\frac{1}{4}}\big(\rho_{+,-}(t)\big)^{-1} \sum_{n=1}^{N_{+}(st) } c_{+,\,-}(t)\,\chi(K_{\varsigma_{j}}<0 )+ t^{-\frac{1}{4}}\big(\rho_{-,+}(t)\big)^{-1} \sum_{n=1}^{N_{-}(st) } c_{-,\,+}(t)\,\chi(K_{\varsigma_{j}}>0).
\end{multline}
Using Lemma~\ref{ExitProbTwo} and the same techniques  above, we can replace
$\rho_{-,+}(t)$ by $\rho_{+,-}(t)$.  More critically, since the
number of up-crossings from below $-2\,t^{\frac{1}{4}}$ to above
$2\,t^{\frac{1}{4}}$ can  differ by at most one from the number of
down-crossings from  above $2\,t^{\frac{1}{4}}$ to below
$-2\,t^{\frac{1}{4}}$,
\begin{multline*}
\big(\rho_{+,-}(t)\big)^{-1}\mathbb{E}\Big[\sup_{0\leq s \leq 1}\Big|   t^{-\frac{1}{4}} \sum_{n=1}^{N_{+}(st) } c_{+,\,-}(t)\,\chi(K_{\varsigma_{j}}<0 )+ t^{-\frac{1}{4}} \sum_{n=1}^{N_{-}(st) } c_{-,\,+}(t)\,\chi(K_{\varsigma_{j}}>0)\Big|\Big] \\ < \sup_{j}\Big(\frac{\Pr\big[K_{\varsigma_{j}}<0\big]}{\rho_{+,-}(t) } +\frac{\Pr\big[K_{\varsigma_{j}}>0\big]}{\rho_{+,-}(t)}   \Big) \Big(       t^{-\frac{1}{4}}|c_{+,\,-}(t)+c_{-,\,+}(t)|\,   \mathbb{E}\big[N_{Y}(t)\big]  +  c_{+,\,-}(t)\wedge c_{-,\,+}(t)      \Big)
\end{multline*}
As above $\mathbb{E}\big[N_{Y}(t)\big]=\mathit{O}(t^{\frac{1}{4}})$, and by Lemmas~\ref{UtSym} and~\ref{ExitProb}
$$\sup_{j}\Big(\frac{\Pr\big[K_{\varsigma_{j}}>0\big]}{\rho_{+,-}(t) } +\frac{\Pr\big[K_{\varsigma_{j}}>0\big]}{\rho_{+,-}(t)}\Big)<4.  $$
Finally, $|c_{+,\,-}(t)+ c_{-,\,+}(t)|$  converges to zero by
Lemma~\ref{UtSym}, which finishes the proof.

\end{proof}

\section{Proof of main results}\label{SecMain}

\begin{proof}[Proof of Theorem \ref{MartBrown}]
Denote the initial distribution $P_{0}(x,k)$ as $\mu$.  We wish to apply the martingale central limit theorem.  The Lindberg condition follows easily, since jumps occur at Poisson rate $\mathcal{R}$ and have finite fourth moments by (\textup{II}) of List~\ref{AssumpOne}.  Weak convergence with respect to the uniform metric then follows by convergence in probability of the predictable quadratic variation $t^{-1}\langle M\rangle_{st}$ to $\sigma s$ for every $s$.  Without loosing generality, we take $s=1$ and show $\mathbb{E}_{\mu}\big[\big|t^{-1}\langle M\rangle_{t}-\sigma \big|\big]\rightarrow 0$.

Define $S_{\epsilon,\,\delta}^{(t)}$ to be the event that $\int_{0}^{1}ds\,\chi\big( t^{-\frac{1}{2}}|K_{st}|>\epsilon \big)> 1-\delta $.  The same statement as in Lemma~\ref{EIL} holds for $t^{-\frac{1}{2}}E_{st}$ replaced by $t^{-\frac{1}{2}}|K_{st}|$, since $E_{r}^{\frac{1}{2}}\approx 2^{-\frac{1}{2}}|K_{r}|$ when $|K_{r}|\gg 1$.   Thus for some $C$, $\Pr[S_{\epsilon,\delta}^{(t)}]\geq 1-C\frac{\sqrt{r_{2}} }{r_{1}}\frac{\epsilon}{\delta}$ for large enough $t$.  It will be convenient to observe that  $\mathcal{R}\int_{0}^{r}ds\,\chi\big( t^{-\frac{1}{2}}|K_{r}|\leq \epsilon \big)$ is close to the number $n_{r,\epsilon}$ of Poisson times $t_{m}$ such that $|K_{t_{m}^{-}}|\leq \epsilon t^{\frac{1}{2}}$ up to time $r\leq t$.  This is a law of large numbers following from the difference $n_{r,\epsilon}- \mathcal{R}\int_{0}^{r}ds\,\chi\big( t^{-\frac{1}{2}}|K_{s}|\leq\epsilon \big) $ being a martingale with
\begin{align}\label{Cosmo}
\mathbb{E}_{\mu}\big[\big(t^{-1}n_{t,\epsilon}- \mathcal{R}\int_{0}^{1}ds\,\chi\big( t^{-\frac{1}{2}}|K_{st}|\leq\epsilon \big)  \big)^{2}\big]^{\frac{1}{2}}\leq t^{-\frac{1}{2}}\mathcal{R}.
\end{align}

Using conditional expectations and the triangle inequality over a telescoping sum determined by the Poisson times $t_{m}$, $m=1,\cdots,\mathcal{N}_{t}$, we get the first inequality below.
\begin{multline}\label{Tiznit}
\mathbb{E}_{\mu}\Big[\Big|\frac{\langle M\rangle _{t}}{t}-\sigma \Big| \Big]
 \leq \mathbb{E}_{\mu}\Big[\frac{1}{t} \sum_{m=0}^{\mathcal{N}_{t}}
\mathbb{E}_{(X_{t_{m} },\,K_{t_{m} })}\big[\big| \langle M \rangle_{t_{m+1}}-\langle M\rangle_{t_{m}}-\sigma(t_{m+1}-t_{m})\big|  \big]   \Big]  \\ <
2^{\frac{1}{2}}r_{2}\big(1-\textup{Pr}\big[S_{\epsilon,\delta}^{(t)} \big]\big)^{\frac{1}{2}}  +  \mathbb{E}_{\mu}\Big[\chi( S_{\epsilon,\delta}^{(t)}) \frac{1}{t} \sum_{m=0}^{\mathcal{N}_{t}}
\mathbb{E}_{(X_{t_{m} },\,K_{t_{m} })}\big[ \big|\langle M \rangle_{t_{m+1}}-\langle M\rangle_{t_{m}}-\sigma(t_{m+1}-t_{m}) \big|  \big]   \Big].
\end{multline}
The second inequality follows by using the brute upperbound for the complement of the event $S_{\epsilon,\delta}$:
\begin{align}\label{Tagalog}
\frac{1}{t} \sum_{m=0}^{\mathcal{N}_{t}}
\mathbb{E}_{(X_{t_{m} },\,K_{t_{m} })}\big[ \big|\langle M \rangle_{t_{m+1}}-\langle M\rangle_{t_{m}}-\sigma(t_{m+1}-t_{m})\big|  \big]  < r_{2}\frac{\mathcal{N}_{t}}{t\mathcal{R}},
\end{align}
which holds since $\langle M \rangle_{r}$ always increases at rates ranging between $r_{1}$ and $r_{2}$ and thus
\begin{align}\label{PreJudah}
\mathbb{E}_{(X_{t_{m} },\,K_{t_{m} })}\big[ \big|\langle M \rangle_{t_{m+1}}-\langle M\rangle_{t_{m}}-\sigma(t_{m+1}-t_{m})\big|  \big] \leq |r_{2}-\sigma|\vee |r_{1}-\sigma|\, \mathcal{R}^{-1}<r_{2}\mathcal{R}^{-1}.
\end{align}
To bound $t^{-1}\mathcal{R}^{-1}r_{2}\mathbb{E}\big[\big(1-\chi(S_{\epsilon,\delta}^{(t)}) \big)  \mathcal{N}_{t} \big]$, we apply the Cauchy-Schwarz inequality to obtain $2^{\frac{1}{2}}t\mathcal{R}(1-\Pr[S_{\epsilon,\delta}^{(t)}])^{\frac{1}{2}} $.

In the event $S_{\epsilon,\delta}^{(t)}$, by the discussion in the second paragraph of this proof, we expect that a fraction of at least $(1-\delta)$ of the times $t_{n}$ have $|K_{t_{n}^{-}}|\geq \epsilon t^{\frac{1}{2}}$.  For large $t$, it will be easy to see that the terms corresponding to momentum jumps $K_{t_{n}}-K_{t_{n}^{-}}=v_{n}$ with $|v_{n}|\geq \frac{1}{2}\epsilon t^{\frac{1}{2}}$ will be negligible, and we can assume $|K_{t_{n}}|>\frac{1}{2}\epsilon t^{\frac{1}{2}}$.  Using that the differences $t_{n+1}-t_{n}$ are exponentially distributed with expectation $\mathcal{R}^{-1}$, then for $(x,k)=(X_{t_{n}},K_{t_{n}})$, $|k|\geq 2^{-1}\epsilon t^{\frac{1}{2}}$
\begin{multline}\label{Judah}
\mathbb{E}_{(x ,\,k )}\big[\big| \langle M \rangle_{t_{n+1}}- \langle M \rangle_{t_{n}}-\sigma (t_{n+1}-t_{n})\big| \big]  =\int_{0}^{\infty}d\tau \,\mathcal{R}e^{-\mathcal{R}\tau}\Big| \int_{0}^{1}da\,\big(T_{(x,k),\,\tau}(a)-\tau\big) \int_{\R}dv\,j_{a}(v)\,v^{2} \Big| \\ \leq r_{2}|k|^{-1}\int_{0}^{\infty}d\tau\mathcal{R}\tau e^{-\mathcal{R}\tau}\leq 2 r_{2}\mathcal{R}^{-1}\epsilon^{-1} t^{-\frac{1}{2}},
\end{multline}
where $r_{2}=\sup_{0\leq a\leq 1}\int_{\R}dv\,j_{a}(v)\,v^{2}$ and  $T_{(x,k),\,r}\in L^{1}([0,1])$ is as defined in Lemma~\ref{SingleJump} which has been applied to get the second inequality.

  Our recipe for bounding~(\ref{Tiznit}) is the following:
\begin{enumerate}
\item Pick $\delta $ so that $r_{2}\delta \ll 1$,
\item Pick $\epsilon$ so that $r_{2}(C \frac{\sqrt{r_{2}} }{r_{1}}\frac{\epsilon}{\delta})^{\frac{1}{2}}\ll 1$,
\item Pick $t$ so that $r_{2}\epsilon^{-1}t^{-\frac{1}{2}}\ll 1$.
\end{enumerate}
By the observations in the second paragraph of this proof,
$1-\textup{Pr}\big[S_{\epsilon,\delta}^{(t)} \big]\leq C
\frac{\sqrt{r_{2}} }{r_{1}}\frac{\epsilon}{\delta}$, and hence the first term on the right side of~(\ref{Tiznit}) is small.

Let  $n_{\epsilon}^{\prime}$ be the number of Poisson times with $|K_{t_{n}^{-}}|> 2^{-1}\epsilon t^{\frac{1}{2}} $ and $|K_{t_{n}}|\leq 2^{-1}\epsilon t^{\frac{1}{2}} $ up to time $t$, and $n_{\epsilon}^{\prime \prime}$ be the number of times with $|K_{t_{n}}|\geq 2^{-1}\epsilon t^{\frac{1}{2}}  $.

Applying the bounds (\ref{Judah}) and~(\ref{PreJudah}) depending on whether or not $|K_{t_{n}}|\geq 2^{-1}\epsilon t^{\frac{1}{2}}$, then we get the first inequality below:
\begin{multline}\label{Lobby}
\mathbb{E}_{\mu}\Big[\chi(S_{\epsilon,\delta}^{(t)})  \, \frac{1}{t}
 \sum_{m=0}^{\mathcal{N}_{t} } \mathbb{E}_{(X_{t_{m} },\,K_{t_{m} })}\big[\big| \langle M \rangle_{t_{m+1}}-\langle M\rangle_{t_{m}}-\sigma(t_{m+1}-t_{m})\big|  \big] \Big] \\
< \mathbb{E}_{\mu}\Big[ \chi( S_{\epsilon,\delta}^{(t)} ) \big( (n_{\epsilon, t}   +n_{\epsilon}^{\prime}) \, \mathcal{R}^{-1}r_{2}\,t^{-1} +n_{\epsilon}^{\prime \prime}\,
2r_{2}\mathcal{R}^{-1}\epsilon^{-1} t^{-\frac{3}{2} }\big)\Big] \leq r_{2}\delta+\mathit{O}(t^{-\frac{1}{2}})    \ll 1.
\end{multline}
Now we look at the expectation of the terms $n_{\epsilon,t}$, $n_{\epsilon}^{\prime}$, and  $n_{\epsilon}^{\prime \prime}$ to reach the second inequality. Since $n_{\epsilon}^{\prime \prime}$ is smaller than the total number of Poisson times $\mathcal{N}_{t}$, $\mathbb{E}_{\mu}[n_{\epsilon}^{\prime \prime}]\leq \mathcal{R}t$, which makes its contribution in~(\ref{Lobby}) $\mathit{O}(t^{-\frac{1}{2}})$.  The number $n_{\epsilon,t}^{\prime}$ is smaller than $\sum_{n=1}^{\mathcal{N}_{t}}\chi(|v_{n}|\geq 2^{-1}\epsilon t^{\frac{1}{2}})$, so for the distribution $\mathcal{P}_{a}(v)$ of a momentum jump conditioned to occur at a torus point $a\in [0,1]$, we can write
 $$t^{-1}\mathbb{E}_{\mu}\big[n_{\epsilon,t}^{\prime} \big]\leq t^{-1}\mathbb{E}_{\mu}\big[\sum_{n=1}^{\mathcal{N}_{t}}\chi(|v_{n}|\geq 2^{-1}\epsilon t^{\frac{1}{2}})\big]\leq 2t^{-1}\mathbb{E}_{\mu}[\mathcal{N}_{t}]\sup_{0\leq a\leq 1}\int_{2^{-1}\epsilon t^{\frac{1}{2}} }^{\infty}dv\,\mathcal{P}_{a}(v)=\mathit{O}(t^{-1}).$$
 The decay for the above expression follows by Chebyshev's inequality and by the uniformly bounded second moments for $\mathcal{P}_{a}$, $a\in[0,1]$.  To bound the $t^{-1}n_{\epsilon, t}$, we will finally be able to use the definition of the event $S_{\epsilon,\delta}$ in the second inequality below
$$\mathbb{E}_{\mu}\big[ \chi( S_{\epsilon,\delta}^{(t)}) t^{-1}n_{\epsilon, t}\big]\leq \mathcal{R}\mathbb{E}_{\mu}\big[ \chi( S_{\epsilon,\delta}^{(t)})\int_{0}^{1}ds\,\chi\big( t^{-\frac{1}{2}}|K_{st}|\leq\epsilon \big)     \big]+\mathit{O}(t^{-\frac{1}{2}})\leq \mathcal{R}\delta+\mathit{O}(t^{-\frac{1}{2}}).  $$
Finally, the first inequality follows from~(\ref{Cosmo}).

\end{proof}

\begin{proof}[Proof of Theorem \ref{AbsoluteValue}]
By Lemma~\ref{StochEqn}, $t^{-\frac{1}{2}}|K_{st}|$ obeys the equation
\begin{multline}
t^{-\frac{1}{2}}|K_{st}| =
\int_{0}^{s}S(K_{r\,t})\,dM^{(t)}_{r}+\sup_{0\leq a\leq
s}-\int_{0}^{a}S(K_{rt})\,dM^{(t)}_{r}+\mathcal{E}_{s}^{(t)}\\
=t^{-\frac{1}{2}}\sum_{n=1}^{\mathcal{N}_{st}}w_{n}\,S(K_{t_{n}^{-}})+t^{-\frac{1}{2}}\sup_{0\leq
s\leq 1}-\sum_{n=1}^{\mathcal{N}_{st}}w_{n}\,S(K_{t_{n}^{-} })+\mathcal{E}_{s}^{(t)}
\end{multline}
where the error $\mathcal{E}_{s}^{(t)}$ vanishes in the norm
 $\mathbb{E}\big[\sup_{0\leq s\leq 1}\big| \mathcal{E}_{s}^{(t)}\big| \big]$.
   The martingale $M_{r}^{\prime}=\sum_{n=1}^{\mathcal{N}_{r}}w_{n}\,S(K_{t_{n}^{-}})$
   has the same quadratic variation as the martingale of momentum jumps $M_{r}$.
$$[ M ]_{r}= [ M^{\prime}]_{r} =
\sum_{n=1}^{\mathcal{N}_{r}}w_{n}^{2}.$$
Thus the predictable quadratic variations $\langle M\rangle_{r}$ and $\langle M^{\prime}\rangle$ are also equal.

 However, by the proof of Theorem~\ref{MartBrown},
 $\mathbb{E}_{\mu}\big[\big|t^{-1}
 \langle M\rangle_{st}-\sigma\,s\big|\big] \rightarrow 0 $
 for $s\in [0,\,1]$, and thus $ M^{\prime}_{r}$ converges to a Brownian motion.
  By the same argument as at the end of the proof of
   Lemma~\ref{TimeChange}, we have that
   $t^{-\frac{1}{2}}|K_{st}|$ converges to the absolute value of a Brownian motion.
\end{proof}

\begin{proof}[Proof of Main Result]
It is sufficient to prove that $t^{-\frac{1}{2}}K_{s\,t}$ converges to a Brownian motion, since by a change of integration variable
$$t^{-\frac{3}{2}}X_{s\,t}=t^{-\frac{3}{2}}\big(X_{0}+\int_{0}^{st}dr\,K_{r}\big)=t^{-\frac{3}{2}}X_{0}+\int_{0}^{s}dr\,t^{-\frac{1}{2}}K_{rt}.$$
Hence $(t^{-\frac{3}{2}}X_{s\,t} ,\,t^{-\frac{1}{2}}K_{s\,t})$ converges in distribution to $(\int_{0}^{s}dr\,\mathbf{B}_{r},\mathbf{B}_{s})$ if $t^{-\frac{1}{2}}K_{s\,t}$ converges to $\mathbf{B}_{s}$.

Recall that $K_{t}=K_{0}+M_{t}+\int_{0}^{t}ds\,\frac{dV}{dx}(X_{s})$. By Lemma~\ref{BigDriftLem}, the first moment of
$$\sup_{0\leq s\leq t}\big| t^{-\frac{1}{2}}\int_{0}^{s}dr\,\frac{dV}{dx}(X_{r}) \big| $$
converges to zero.  By Theorem~\ref{MartBrown}, $t^{-\frac{1}{2}}M_{st}$ converges to a Brownian motion with diffusion constant $\sigma$.   Hence $t^{-\frac{1}{2}}K_{st}$ converges to a Brownian motion.

\end{proof}

\section*{Acknowledgments}
We thank Wojciech De Roeck, Karel Neto\v{c}n\'{y} and Frank Redig
for useful discussions. J. C. acknowledges support from the
Belgian Interuniversity Attraction Pole P6/02 and the Marie Curie Research Training Network project MRTN-CT-2006-035651, Acronym CODY, of the European Commission.

\begin{appendix}

\section{A martingale lemma}

Below we formalize a simple lemma for a martingale
$\textbf{M}_{t}$ which makes jumps at discrete times.  For a basic introduction to martingale theory, see~\cite{Chung}.

\begin{lemma}\label{HoldersInequality}
Let $\mathbf{M}_{t}=\sum_{n=1}^{\mathbf{N}_{t}}X_{n}$ be a
right-continuous  martingale adapted to a filtration
$\mathcal{F}_{t}$, and making jumps $X_{n}(\omega) $, $\omega\in
\mathcal{F}_{t_{n}}$ at discrete times according to some random
counter $\mathbf{N}_{t}$ where $t_{n}=\inf\{s\in \R^{+}\big|
\mathbf{N}_{s}=n \} $.  Assume also that
$\mathbb{E}\big[\mathbf{N}_{t}\big]<\infty$ for every $t$ and that
$\mathbb{E}\big[X_{n}^{2}]<\infty$ for every $n$.  Then
$$\mathbb{E}\big[\mathbf{M}_{t}^{2}\big]\leq \mathbb{E}\big[\mathbf{N}_{t}\big]\,
\sup_{n}\mathbb{E}\big[X_{n}^{2}\big| n\leq \mathbf{N}_{t}]$$
\end{lemma}
where $\mathbb{E}\big[X_{n}^{2}\big| n\leq \mathbf{N}_{t}]$ is defined as zero for $n$ such that $\Pr[n\leq \mathbf{N}(t)]=0.$

\begin{proof}

By orthogonality of martingale increments
\begin{multline*}
\mathbb{E}\big[\mathbf{M}_{t}^{2}\big]= \mathbb{E}\Big[ \sum_{n=1}^{\mathbf{N}_{t}}X_{n}^{2}  \Big]= \mathbb{E}\Big[ \sum_{n=1}^{\infty} X_{n}^{2}\chi\big(n\leq  \mathbf{N}_{t}\big)\Big]  =\sum_{n=1}^{\infty}\Pr[n\leq  \mathbf{N}_{t}]\, \mathbb{E}\big[X_{n}^{2}\big| n\leq \mathbf{N}_{t}] \\ \leq  \big( \sum_{n=1}^{\infty}\Pr[n\leq  \mathbf{N}_{t}] \big)\,\sup_{n}\mathbb{E}\big[X_{n}^{2}\big| n\leq \mathbf{N}_{t}]\leq \mathbb{E}\big[\mathbf{N}_{t}\big]\,\sup_{n}\mathbb{E}\big[X_{n}^{2}\big| n\leq \mathbf{N}_{t}],
\end{multline*}
where the inequality is H\"older's.

\end{proof}

\section{Boundary crossing distributions}\label{AppendixBC}

Consider symmetric independent and identically distributed random
variables $X_{1},\,X_{2},...$ with mean zero and density
\begin{align*}
\tilde{P}(v)=\int_{0}^{1}da\,\frac{\kappa(a)}{\bar{\kappa}}\,\mathcal{P}_{a}(v).
\end{align*}
 Construct the random walk
$Y_{n}=\sum_{m=1}^{n}X_{m}$.  We will refer to this as the so called ``averaged random walk''
Let $L\geq 0$.  Our interest in this
section is to understand the probability density
\begin{align}\label{Records}
\pi_{L}(a,\,b)=\sum_{n=1}^{\infty}\textup{Pr}[Y_{n}-L=a,\,X_{n}=b,\,Y_{r}<L \text{ for } r<n   ],\quad a,\,b\leq \R^{+}.
\end{align}
In the present appendix we use the notation  Pr$[\cdot]$ for the
induced density by the random walk, in its obvious meaning for
continuous densities.  To be clear, $\pi_{L}$ describes the
distribution of jumps in excess over point $L$ for the random walk
$Y_{n}$ on the first time that it passes over that $L$.

For $L=0$, we  write $D(v,\,w)=\pi_{L=0}(v,\,w)$ which is in fact
the ``successive record increment'' distribution. Indeed, define
$R_{n}=\sup_{0\leq m\leq n}Y_{m}$, the record in the positive
direction for the random walk up to time $n$, and let $\tau_{m}$
be the time of the $m$th record.  Then, the increments
$R_{\tau_{m}}-R_{\tau_{m-1}}$ are i.i.d. random variables with
density $D(a)=\int_{0}^{\infty}db\,D(a,\,b)$.  We use this fact
in the proof of the following proposition.  In particular
$$\mathbb{E}\big[ R_{\tau_{m}}-R_{\tau_{m-1}}
\big]=\int_{0}^{\infty}\int_{0}^{\infty}da\,db\,a\,D(a,\,b).$$ The
topic of ``record distributions" has a much wider scope and has a
long history in extreme value statistics, see e.g. \cite{hus} for
some pioneering contribution.

\begin{lemma}\label{HighBrow}
Assume (\textit{III}) of List~\ref{AssumpOne} and (\textit{i})-(\text{ii}) of List~\ref{AssumpTwo}.
\begin{enumerate}
\item The Laplace transform $\varphi(q)$ of $D(v)$ satisfies
$$\varphi(q)\leq \frac{\mathcal{C} }{1-e^{-\eta-q}} \quad \text{ for } q> -\eta.$$

\item $D(v)$ is bounded and continuous.

\item The following is a probability density on  $\R^{+}\times\R^{+}$
$$\pi_{\infty}(v,\,w)=\frac{ \int_{v}^{\infty}dx\,D(x,\,w)}{ \int_{0}^{\infty}\int_{0}^{\infty}dx\,dy\, \,x\,D(x,\,y) },$$
which is positive for all $v<w$.

\end{enumerate}

\end{lemma}

\begin{proof} Part (1) follows, since $D(v)$ inherits the property (\textit{i}) of List~\ref{AssumpTwo}.
\end{proof}

\begin{proposition}\label{BC1}
Assume (\textit{III}) of List~\ref{AssumpOne} and (\textit{i})-(\textit{ii}) of List~\ref{AssumpTwo}.  Let the random walk $Y_{n}$, $\pi_{L}(v,\,w)$, and $D(v,\,w)$ be defined as above.

In the limit of large $L$, we have the $L^{1}(\R^{+}\times\R^{+})$ convergence
$$\pi_{L}(v,\,w)\longrightarrow \pi_{\infty}(v,\,w)=\frac{ \int_{v}^{\infty}dx\,D(x,\,w)}{ \int_{0}^{\infty}\int_{0}^{\infty}dx\,dy\, \,x\,D(x,\,y) }
  $$
Moreover, the difference between $\pi_{\infty}$ and
$\pi_L$ converges exponentially and monotonically.

\end{proposition}

\begin{corollary}\label{BC2}
Let  $Y_{n}$ be a random walk as above.  For $L>0$ and $d\in \R$
with $|d|\leq \frac{1}{2}L$ define $S=(-d-L,L)\subset \R$, and let
$\pi_{\uparrow,\,L}(v,\,w),\pi_{\downarrow,\,L}(v,\,w)$ be the
probability densities on $\R^{+}\times \R^{+}$ defined as
\begin{eqnarray}
\pi_{\uparrow,\,L}(v,\,w)&=&\sum_{n=1}^{\infty}\textup{Pr}\big[X_{n}=w,\,Y_{n}-L= v,\, Y_{m}\in S \text{ for } 0\leq m< n \big]1_{Y_{n}>0}. \\
\pi_{\downarrow,\,L}(v,\,w)&=&\sum_{n=1}^{\infty}\textup{Pr}\big[-X_{n}=w,\,Y_{n}+L+d=
-v,\,    Y_{m}\in S\ \text{ for } m< n \big]1_{Y_{n}<0}
\end{eqnarray}

In the limit $L\rightarrow \infty $, we have $L^{1}(\R^{+}\times \R^{+})$ convergence
$$\pi_{\uparrow,t}\longrightarrow p_{\uparrow}\,\pi_{\infty}\quad   \pi_{\downarrow,t} \longrightarrow p_{\downarrow}\,\pi_{\infty}$$
 Moreover, the convergence is uniform for $|d|\leq \frac{1}{2}L$.

\end{corollary}

\begin{proof}
In general, we have the identity
\begin{align}\label{BCID}
\pi_{L}(v,\,w)= \pi_{\uparrow,\,L}(v,\,w)+\int_{\R^{+}\times\R^{+}
}da\,db\, \pi_{\downarrow}(a,\,b)\,\pi_{2L+d+b}(v,\,w),
\end{align}
where $\pi_{L}(v,\,w)$ is the boundary increment distribution from~(\ref{BC1}).  As $L\rightarrow \infty$, $\pi_{L}(v,\,w),\,\pi_{2L+d+b}(v,\,w)$ converge exponentially in $L^{1}(\R^{+}\times\R^{+})$ to $\pi_{\infty}(v,\, w)$ by Proposition~\ref{BC1}.  By the triangle inequality and rearranging~(\ref{BCID}),
$$\|\pi_{\uparrow,\,L}-\pi_{\infty}\|_{1}\leq \int_{\R^{+}\times\R^{+} }da\,db\, \pi_{\downarrow}(a,\,b)\,\| \pi_{2L+d+b}-\pi_{L}\|_{1}\leq 2\,p_{\downarrow}\,\sup_{L^{\prime}\geq L}\| \pi_{L^{\prime} }-\pi_{\infty}\|_{1},
$$
where  where $p_{\uparrow}=\int_{\R^{+}\times
\R^{+}}dv\,dw\,\pi_{\uparrow,\,L}(v,\,w)$ and
$p_{\downarrow}=\int_{\R^{+}\times
\R^{+}}dv\,dw\,\pi_{\downarrow,\,L}(v,\,w)$. Since $|d|\leq
\frac{1}{2}L$, $p_{\uparrow}^{-1}$ is bounded away from zero and
the convergence holds.

The uniformity of the convergence over $|d|\leq \frac{1}{2}L$ is implied by the monotonicity of the convergence from Proposition~\ref{BC1}.

\end{proof}

Next we extend Proposition~\ref{BC1} and Corollary~\ref{BC2} to results about our dynamics in high momentum situations.  We prove an analogue of Corollary~\ref{BC2} and then prove the analogue of Proposition~\ref{BC1} as a corollary.  Since there is also a position variable moving on the torus, the crossing distributions are joint distributions on $[0,1]\times \R^{+}$.   We consider the so called ``averaged random walk''
\begin{align*}
\tilde{P}(v)=\int_{0}^{1}da\,\frac{\kappa(a)}{\bar{\kappa}}\,\mathcal{P}_{a}(v),
\end{align*}
 where $\bar{\kappa}=\int_{0}^{1}da\,\kappa(a)$, and
  its boundary crossing distribution $\pi_{\infty}(v,\,w)$ obtained from the
  previous Proposition B1 gives rise to another density
\begin{align}\label{Jazz}
\phi_{\infty}(a,\,v)=\int_{\R^{+}}  dw\,\pi_{\infty}(v,\,w) \frac{
\kappa(a)\,\mathcal{P}_{a}(w)}{\int_{0}^{1}dx\,\kappa(x)\mathcal{P}_{x}(w)
}.
\end{align}
Intuitively, $\phi_{\infty}(a,\,v)$
is an idealized joint distribution for the position on the torus and the
over-jump for the particle upon crossing the border of some domain
$S\subset [0,1]\times \R^{+} $. $ \phi_{\infty}(a,\,v)$ provides a
good approximation for such a distribution at high energy where
the momentum behaves nearly as the averaged random walk.

The following lemma is similar to Part (1) of Lemma~(\ref{HighBrow}).
\begin{lemma}\label{BoundryVariance}
Assume (\textit{i}) of List~\ref{AssumpTwo}.     Let $S\subset \R $ be an open interval with boundaries $L_{1}< L_{2}$ not both unbounded.  If the particle begins at some point $(x,\,k)$ with $k\in S$, then define
\begin{eqnarray*}
\pi_{\uparrow}(v)&=& \sum_{n=1}^{\infty} \textup{Pr}\big[K_{t_{n}}-L_{2}=v \geq 0,\,\, K_{t_{m}}\in S     \text{ for } 0\leq m< n \big]1_{ K_{t_{n} }\geq S}, \\
\pi_{\downarrow}(v)&=&\sum_{n=1}^{\infty}
\textup{Pr}\big[K_{t_{n}}-L_{1}= -v \leq 0, K_{t_{m}}\in
S\,\text{ for }\, 0\leq m< n \big]1_{ K_{t_{n}}\leq S}.
\end{eqnarray*}
There exists a universal upper bound $J$ such that for all
$L_{1},\,L_{2}$ and $(x,\,k)$
$$ \frac{\int_{\R^{+}}dv\,\pi_{\uparrow}(v)\,v^{2}}{\int_{\R^{+}}dv\,\pi_{\uparrow}(v) },\,\frac{\int_{\R^{+}}dv\,\pi_{\downarrow}(v)\,v^{2}}{\int_{\R^{+}}dv\,\pi_{\downarrow}(v) } \leq J,  $$
when
$\int_{\R^{+}}dv\,\pi_{\uparrow}(v),\,\int_{\R^{+}}dv\,\pi_{\downarrow}(v)\neq
0$. More generally, the Laplace transforms $\varphi_{\uparrow}$,\,
$\varphi_{\downarrow}$ have
$$\varphi_{\uparrow}(q),\,\varphi_{\downarrow}(q)\leq \frac{\mathcal{C} }{1-e^{-\eta-q}} \quad \text{ for } q> -\eta.   $$

\end{lemma}

Let $\textup{P}(v_{1},\cdots,v_{n})$ be the joint probability distribution for the increments $v_{j}=K_{t_{j}}-K_{t_{j-1}}$ where $t_{j}$ is the time of the $j$th momentum jump.

\begin{lemma}\label{JointDen}
Let us make assumptions (\textit{i})-(\textit{ii}) from List~\ref{AssumpTwo}.   Fix $\beta, \gamma > 0$. Let the dynamics begin at a point $(x,k)$ and $v_{1},\dots,v_{n}\in [-2\,t^{\gamma},\,2\,t^{\gamma}]$ be such that $|k+\sum_{r=0}^{m}v_{r}|\geq \frac{1}{2}t^{\beta}$ for all $0\leq m\leq n$.  Then
$$\Big| \frac{\textup{P}(v_{1},\cdots,v_{n}) }{\tilde{P}(v_{1})\cdots \tilde{P}(v_{n})}  -1 \Big|  < c\,n \,t^{\gamma-\beta}e^{c\,n\, t^{\gamma-\beta}},  $$
where $c=8\bar{V}\mu$.

\end{lemma}

\begin{proof}
We start with the one-jump case.  The only difference between Lemma~\ref{JointDen} and Lemma~\ref{SingleJump} is that here the increment $v$ is the sum of the first momentum kick \`and the drift up to the time of the kick.

Let $(x_{s},k_{s})$ be on the trajectory determined by the Hamiltonian $H(x,k)=\frac{1}{2}k^{2}+V(x)$ starting from $(x,\,k)\in [0,1]\times \R$.   For fixed starting point $(x,\,k)$, $\Delta_{s}=k_{s}-k=\Delta(x_{s})$ is function of $x_{s}$.  For the one-jump case, we can use the density $\tilde{r}_{(x,k)}\in L^{1}([0,1])$ from the proof of Lemma~\ref{SingleJump}, to write
$$ \textup{P}(v)=\int_{0}^{1}da\,\tilde{r}_{(x,\,k)}(a)\,\mathcal{P}_{a}\big(v-\Delta(a) \big)$$

Define the density
$\textup{P}^{\prime}(v)=\int_{0}^{1}da\,\frac{\kappa(a)}{\bar{\kappa}}\,\mathcal{P}_{a}\big(v-\Delta(a) \big) $.  Following the same line of reasoning as~(\ref{IG}) in Lemma~\ref{SingleJump},
\begin{multline*}
\sup_{v\in \R }\big| \frac{\textup{P}(v)   }{\textup{P}^{\prime}(v)  }-1\big|\leq \sup_{v\in \R} \frac{1}{\textup{P}^{\prime}(v)}\int_{0}^{1}da\,\big|\tilde{r}_{(x,\,k)}(a)-\frac{\kappa(a)}{\bar{\kappa}}\big| \mathcal{P}_{a}\big(v-\Delta(a)\big)\\
\leq  \sup_{0\leq a\leq 1}\big|\tilde{r}_{(x,\,k)}(a)\frac{\bar{\kappa}}{\kappa(a)}-1\big| \frac{1}{\textup{P}^{\prime}(v)}\int_{0}^{1}da\,\frac{\kappa(a)}{\bar{\kappa}} \mathcal{P}_{a}\big(v-\Delta(a)\big)  =\sup_{a\in [0,1]}\big|\tilde{r}_{(x,k)}(a)\frac{\bar{\kappa}}{\kappa(a)}-1\big|
\end{multline*}
Since, $|k|\geq \frac{1}{2}t^{\beta}$, then by Part (2) of Lemma~\ref{SingleJump},
$$\sup_{a\in [0,1]}\big|\tilde{r}_{(x,k)}(a)\frac{\bar{\kappa}}{\kappa(a)}-1\big|\leq \nu^{-2}4\mathcal{R}\, t^{-\beta}+\mathit{O}(t^{-2\beta}). $$

Moreover,
\begin{multline}\label{TributeBand}
 \big| \frac{\textup{P}^{\prime}(v)}{\tilde{P}(v)}-1 \big| \leq \frac{1}{\tilde{P}(v)}\,\int_{0}^{1}da\,\frac{\kappa(a)}{\bar{\kappa}}\,\big|      \mathcal{P}_{a}(v-\Delta(a))-\mathcal{P}_{a}(v)\big| \\ = \frac{1}{\tilde{P}(v)} \int_{0}^{1}da\,\frac{\kappa(a)}{\bar{\kappa}} |\Delta(a)|\,\Big| \int_{0}^{1}ds  \, \frac{d\mathcal{P}_{a}}{dv}\,\big(v+s\,\Delta(a) \big)\Big| \\ \leq  \frac{\mu\,4\bar{V}(1+|v|)t^{-\beta} }{\tilde{P}(v)}\int_{0}^{1}da\,\frac{\kappa(a)}{\bar{\kappa}}\mathcal{P}_{a}(v)= \mu\,4\bar{V}(1+|v|)t^{-\beta}
\end{multline}
 For the second inequality we have used Lemma~\ref{RealBasics} to get $\sup_{0\leq s\leq 1}|\Delta(a)|\leq 4\bar{V}t^{-\beta}$.  In particular, $|\Delta(a)|<1$ allows us to use (\textit{ii}) of List~\ref{AssumpTwo} to bound $ |\frac{d\mathcal{P}_{a}}{dv}\,\big(v+s\,\Delta(a)\big)|$ by $\mu (1+|v|)\mathcal{P}_{a}(v)$.   Since $1+|v|\leq 1+2\,t^{\gamma}<4t^{\gamma}$ for large enough $t$, we have
$$\sup_{v\in \R} \big| \frac{\textup{P}^{\prime}(v)}{\tilde{P}(v)}-1 \big|< 16\bar{V}\mu\,t^{\gamma-\beta} $$

\begin{multline}
\Big| \frac{\textup{P}(v)}{\tilde{P}(v)}-1\Big|\leq \Big| \frac{\textup{P}(v)}{\textup{P}^{\prime}(v)}-1\Big|\,\frac{\textup{P}^{\prime}(v)}{\tilde{P}(v)}+\Big| \frac{\textup{P}^{\prime}(v)}{\tilde{P}(v)}-1\Big|
\leq (16\bar{V}\mu\, t^{\gamma-\beta})\,\big(1+4\mathcal{R}\,\nu^{-2}\,t^{-\beta} +\mathit{O}(t^{-2\beta})\big) \\+4\mathcal{R}\,\nu^{-2}t^{-\beta}+\mathit{O}(t^{-2\beta})< 20\bar{V}\mu
\,t^{\gamma-\beta}=c\,t^{\gamma-\beta}
\end{multline}
where the last inequality is for $t$ large enough and $c=20\bar{V}\mu$ .

To extend to arbitrary $n$
\begin{align}
\Big| \frac{\textup{P}(v_{1},\dots,v_{n}) }{\tilde{P}(v_{1})\dots \tilde{P}(v_{n})}  -1 \Big| \leq \sum_{m=1}^{n}\Big| \frac{\textup{P}(v_{1},\dots,v_{m}) }{\tilde{P}(v_{1})\dots \tilde{P}(v_{m})}-\frac{\textup{P}(v_{1},\dots,v_{m-1}) }{\tilde{P}(v_{1})\dots \tilde{P}(v_{m-1})}\Big|
\end{align}
Moreover, we can write the summand as
$$\Big|\frac{\textup{P}(v_{1}, \dots,v_{m})}{ \textup{P}(v_{1},\dots ,v_{m-1})\,\tilde{P}(v_{m}) }-1\Big|\, \prod_{r=1}^{m-1}\frac{\textup{P}(v_{1}, \dots,v_{r})}{ \textup{P}(v_{1},\dots ,v_{r-1})\,\tilde{P}(v_{r}) }         $$
The ratio's appearing in the above expression can be written in terms of an expectation:
\begin{align} \label{FracProb}
\frac{\textup{P}(v_{1}, \dots,v_{m})}{ \textup{P}(v_{1},\dots ,v_{m-1})\,\tilde{P}(v_{m})  }  =\mathbb{E}\Big[\frac{\textup{P}_{(X_{t_{m-1}},K_{t_{m-1}})}(v_{m})}{ \tilde{P}(v_{m}) }  \Big| K_{t_{r}}-K_{t_{r-1}}=v_{r} \text{ for } r<m \Big]
 \end{align}
where $\textup{P}_{(X_{t_{m-1}},K_{t_{m-1}})}(v_{m})$ is the one jump distribution starting from $(x,k)=(X_{t_{m-1}},K_{t_{m-1}})$.  The right side of~(\ref{FracProb}) is thus an average of the one-jump case over different starting points $(x,k)=(X_{t_{m-1}},K_{t_{m-1}})$ and thus we can apply the $n=1$ case.

We thus conclude
$$\Big| \frac{\textup{P}(v_{1},\dots,v_{n}) }{\tilde{P}(v_{1})\dots \tilde{P}(v_{n})}  -1 \Big|< c\,n\,t^{\gamma-\beta}\,(1+c\,t^{\gamma-\beta})^{n}< c\,n\,t^{\gamma-\beta}\,e^{c\,n\,t^{\gamma-\beta}}.$$

\end{proof}

For notational convenience, we will take the initial point $(x,k)$ to have positive momentum $k>0$ for the propositions and lemmas below.

\begin{proposition}\label{BC2NHalf}
Assume~(\textit{i})-(\textit{ii}) from List~\ref{AssumpTwo}.  Let
$(X_{s},K_{s})$ evolve from an initial point $(x,k)$ with $k\geq
t^{\beta}$.  Define $S=\{y\in \R \, |\, k-t^{\gamma}-d
<y<k+t^{\gamma}\} $ for $0< \gamma<\frac{1}{4}\beta $ and $|d|\leq
\frac{1}{2}t^{\gamma}$, and the densities
\begin{eqnarray*}
\pi_{\uparrow,\,t}(v,\,w)&=& \sum_{n=1}^{\infty} \textup{Pr}[K_{t_{n}}-K_{t_{n-1}}=w,\,K_{t_{n}}-k-t^{\gamma}=v \geq 0,\,\, K_{t_{m}}\in S \,   \forall( m< n) ]1_{ K_{t_{n} }>S} \\
\pi_{\downarrow,\,t}(v,\,w)&=&\sum_{n=1}^{\infty}
\textup{Pr}[K_{t_{n}}-K_{t_{n-1}}=w,\, K_{t_{n}}-k+t^{\gamma}+d=
-v \leq 0,\,    K_{t_{m}}\in S\,\forall( m< n) ]1_{ K_{t_{n}}<S}
\end{eqnarray*}

In the limit $t\rightarrow \infty$, there is $L^{1}(\R^{+}\times\R^{+})$ convergence
$$\pi_{\uparrow,\,t}\longrightarrow p_{\uparrow}\,\pi_{\infty},\quad \pi_{\downarrow,\,t}\longrightarrow  p_{\downarrow}\,\pi_{\infty},$$
where $p_{\uparrow}$ and $p_{\downarrow}$ are the exit
probabilities for the averaged random walk, and the convergence is
uniform for all $|d|\leq \frac{1}{2}t^{\gamma}$.

\end{proposition}

\begin{proof}

 Define the sets
$$ S_{n}(p,\,w)= \{(v_{1},\dots,v_{n})\in \R^{n}\,|v_{n}=w, \,  p=k+\sum_{r=1}^{n}v_{r}\notin S,\,\, k+\sum_{r=1}^{m}v_{r}\in S \text{ for } m<n \}.$$
and define $S_{n}^{\prime}(p,\,w)$ so that
$S_{n}^{\prime}(p,\,w)=S_{n}(p,\,w)\cap \{ |v_{m}|\leq 2t^{\gamma}
\text{ for } m\leq n\} $  when $n\leq \lfloor t^{3\gamma}\rfloor$
and $S_{n}^{\prime}(p,\,w)=\emptyset$ when $n> \lfloor
t^{3\gamma}\rfloor$.

Then $\pi_{\uparrow,\,t}(v,\,w),\pi_{\downarrow,\,t}(v,\,w)$ satisfy
\begin{align}\label{CrossDef}
 \pi_{\uparrow,\,t}(p-k-t^{\gamma},\,w)+\pi_{\downarrow,\,t}(k-t^{\gamma}-d-p,\,w) =\sum_{n=1}^{\infty}\int_{S_{n}(p,\,w)  }dv_{1}\cdots dv_{n}\,\textup{P}(v_{1},\dots, v_{n})
\end{align}
where $\textup{P}(v_{1},\cdots,v_{n})$ are the joint densities
from Lemma~\ref{JointDen}. The sum of $\pi_{\uparrow,\,t}$ and $\pi_{\downarrow,\,t}$ forms a probability measure, since, by an analogous argument as in the proof of Lemma~\ref{FiniteTimes}, the exit time for $S$ is almost surely finite (and has expectation $\propto t^{2\beta}$).   We will argue that the right side is
close in $L^{1}(\R^{+}\times \R^{+})$ as a function of $(p,\,w)$
to the same expression with $\textup{P}(v_{1},\dots,v_{n})$
replaced by $\tilde{P}(v_{1})\cdots \tilde{P}(v_{n})$:
\begin{align*}
 \pi_{\uparrow,\,t}^{(0)}(p-k-t^{\gamma},\,w)+\pi_{\downarrow,\,t}^{(0)}(k-t^{\gamma}-d-p,\,w) =\sum_{n=1}^{\infty}\int_{S_{n}(p,\,w)  }dv_{1}\cdots dv_{n}\,\tilde{P}(v_{1})\cdots \tilde{P}(v_{n}).
\end{align*}
This expression corresponds to the averaged random walk.
We can
then apply Lemma~\ref{BC2} for an ordinary random walk to get the convergence of $\pi_{\uparrow,\,t}^{(0)}$ and
$\pi_{\downarrow,\,t}^{(0)}$ to $p_{\uparrow} \pi_{\infty}$ and $p_{\downarrow}\pi_{\infty}$.

By definition of $S_{n}(p,\,w),\,S_{n}^{\prime}(p,\,w)$,\, $K+\sum_{n=1}^{m}v_{m}>\frac{1}{2}t^{\beta}$ for all $0\leq m<n$.  First, notice that by Lemma~\ref{JointDen}
\begin{multline}\label{Elephant}
\sum_{n=1}^{\infty}\int_{S_{n}^{\prime}(p,\,w)  }dv_{1}\cdots dv_{n}\,\tilde{P}(v_{1})\cdots \tilde{P}(v_{n})  \Big| \frac{\textup{P}(v_{1},\cdots,v_{n}) }{\tilde{P}(v_{1})\cdots \tilde{P}(v_{n})}  -1 \Big|\\ \leq \sum_{n=1}^{\lfloor t^{3\gamma} \rfloor  }\int_{S_{n}^{\prime}(p,\,w)  }dv_{1}\cdots dv_{n}\,\tilde{P}(v_{1})\cdots \tilde{P}(v_{n})\,c\,n \,t^{\gamma-\beta}e^{c\,n\, t^{\gamma-\beta}}
\end{multline}
The $L^{1}(\R^{+}\times \R^{+})$ norm of the right-side of~(\ref{Elephant}) (by integrating over $p,\,w$) is bounded by $\mathbb{E}\big[q\,(N_{0}\wedge t^{3\gamma})  \,e^{q\,(N_{0}\wedge t^{3\gamma})   }\big]\big|_{q=c\,t^{\gamma-\beta}}$, where $N_{0}$ is the number of steps that a random walk with jumps having density $ c_{t}^{-1}\tilde{P}(v)1_{|v|\leq t^{\gamma}}$, for normalization constant $c_{t}$, takes to leave $S$ starting from $k$.  The right-side of~(\ref{Elephant}) vanishes as
$$\mathbb{E}\big[q\,(N_{0}\wedge t^{3\gamma})  \,e^{q\,(N_{0}\wedge t^{3\gamma})   }\big]\big|_{q=c\,t^{\gamma-\beta}}\leq c\,t^{4\gamma-\beta} \,e^{c\, t^{4\gamma-\beta }   }\approx c\,t^{4\gamma-\beta}. $$

Now we need to show that not much probability was lost by replacing $S_{n}(p,\,w)$ with $S_{n}^{\prime}(p,\,w)$.  Since both  $\pi_{\uparrow,\,t}+\pi_{\downarrow,\,t}$ and $\pi_{\uparrow,\,t}^{(0)}+\pi_{\downarrow,\,t}^{(0)}$ are probability measures, it is enough to show that the probability of the event $ \{ |v_{m}|> 2t^{\gamma} \text{ for some } m\leq N \text{ or } N> t^{3\gamma} \} $ is small for the random walk.
First, note that the expected number of steps $N$ to leave $S$ for the averaged random walk will be smaller than the expected number of the steps $N_{0}$ for the capped random walk.
$$\Pr[N>t^{3\gamma}]\leq \Pr[N_{0}>t^{3\gamma}] \text{ and } \mathbb{E}[N]\leq \mathbb{E}[N_{0}].  $$
This follows since the first jump of size greater than $2t^{\gamma}$ will immediately leave $S$, and the corresponding capped trajectory may have to continue on for more steps before leaving $S$.

Now we argue that  $\mathbb{E}[N_{0}]=\mathit{O}(t^{2\gamma})$.  Let us set $k=0$.  Define the stopping time $N_{0,\mathcal{T}}=N_{0}\wedge \mathcal{T}$.  Reasoning as in the proof of Lemma~\ref{FiniteTimes}, then
$$\mathbb{E}\big[ N_{0,\mathcal{T}}\big]=\zeta^{-1}\mathbb{E}\big[\sum_{n=1}^{N_{0,\mathcal{T}} }  v_{n}^{2}\big]=\zeta^{-1}\mathbb{E}\big[\big(v_{1}+\dots +v_{N_{0,\mathcal{T}}}\big)^{2}  \big]\leq \zeta^{-1}(3\,t^{\gamma}+d)^{2}<16\,\zeta^{-1}t^{2\gamma}, $$
where the expectations are with respect to the statistics for the capped random walk, and $\zeta= c_{t}^{-1}\int_{|v|\leq t^{\gamma}}dv\,\tilde{P}(v)\,v^{2}$.

In the second equality, we used that $v_{1}+\dots
+v_{N_{0,\mathcal{T}}}$ is either inside
$[-t^{\gamma}-d,t^{\gamma}]$ when $\mathcal{T}< N_{0}$ or has
jumped out this interval with a jump smaller than $2t^{\gamma}$.
The bound on the right is independent of $\mathcal{T}$, so
$$\mathbb{E}\big[ N_{0}\big]= \limsup_{\mathcal{T}\rightarrow \infty} \mathbb{E}\big[ N_{0,\mathcal{T}}\big]\leq 16\,\zeta^{-1}t^{2\gamma}.$$
By Chebyshev's inequality
$$\Pr[N_{0}>t^{3\gamma}]\leq t^{-3\gamma} \mathbb{E}\big[ N_{0}\big]\leq 16\,\zeta^{-1}t^{-\gamma}\longrightarrow 0,  $$
since $\zeta$ converges to $\sigma$ for large $t$.

We still need to show that $\Pr\big[ |v_{n}|> 2t^{\gamma} \text{ for some } n\leq N \big]$ is small.
$$\Pr\big[ |v_{n}|> 2t^{\gamma} \text{ for some } n\leq N \big]\leq \mathbb{E}\big[N\big]\, \Pr\big[ |v_{n}  |\geq 2 t^{\gamma} \big].$$
By our remarks above $\mathbb{E}\big[N\big]\leq \mathbb{E}\big[N_{0}\big]=\mathit{O}(t^{2\gamma})$. Using Chebyshev's inequality, and the bound on the fourth moment of a single momentum jump by $\rho$,
\begin{align*}
16\, t^{4\gamma}\Pr\big[ |v_{n} |\geq t^{2\gamma} \big]\leq
\mathbb{E}\big[ v_{n}^{4}\big] \leq  \rho.
\end{align*}
Putting the above inequalities together
$$ \Pr\big[ |v_{n}   |\geq 2 t^{\gamma},\text{ for some } n\leq N\big]\leq \mathbb{E}[N_{0}]\, \Pr\big[ |v_{n}|\geq 2 t^{\gamma} \big]< \rho\,\zeta^{-1}\,  t^{-2\gamma}.$$
Hence the event that the last jump $v_{n}$

 We have shown that
$\pi_{\uparrow,\,t}$, $\pi_{\downarrow,\,t}$ converge to
$\pi_{\infty}$ in $L^{1}(\R^{+}\times \R^{+})$ at $t\rightarrow
\infty$. The convergence of $\pi_{\uparrow,\,t}$ and $\pi_{\uparrow,\,t}$ to $\pi_{\infty,\,t}$ is uniform over $|d|\leq \frac{1}{2}t^{\gamma}$ by the uniformity in Lemma~\ref{BC2} and by the uniformity in the bounds above.

\end{proof}

\begin{proposition}\label{BC3}
Assume~(\textit{i})-(\textit{ii}) from List~\ref{AssumpTwo}.  Let
$(X_{s},K_{s})$ evolve from an initial point $(x,k)$ with $k\geq
t^{\beta}$.  Define $S=\{y\in \R \, |\, k-t^{\gamma}-d
<y<k+t^{\gamma}\} $ for $0< \gamma<\frac{1}{4}\beta $, and the
densities
\begin{eqnarray*}
\phi_{\uparrow,\,t}(a,\,v)&=& \sum_{n=1}^{\infty} \textup{Pr}[X_{t_{n}}=a,\,K_{t_{n}}-k-t^{\gamma}=v \geq 0,\,\, K_{t_{m}}\in S     \text{ for } 0\leq m< n ]1_{ K_{t_{n} }>S}, \\
\phi_{\downarrow,\,t}(a,\,v)&=&\sum_{n=1}^{\infty}
\textup{Pr}[X_{t_{n}}=a,\, K_{t_{n}}-k+t^{\gamma}+d= -v \leq
0,\text{ for }    K_{t_{m}}\in S\,\, 0\leq m< n ]1_{ K_{t_{n}}<S}.
\end{eqnarray*}
In the limit $t\rightarrow \infty$, there is $L^{1}([0,1]\times\R^{+})$ convergence
$$\phi_{\uparrow,\,t}\longrightarrow p_{\uparrow}\, \phi_{\infty}(a,\,v) ,\,\phi_{\downarrow,\,t}\longrightarrow p_{\downarrow}\,\phi_{\infty}(a,\,v),$$
where $p_{\uparrow}$ and $p_{\downarrow}$ are probabilities that
the averaged random walk exits above or below $S$ respectively.

Moreover, the convergence is uniform for $|d|\leq \frac{1}{2}t^{\gamma}$.

\end{proposition}

\begin{proof}
Define the joint density $\Phi_{\uparrow,\,t}(a,\,v,\,w)$ for the
position $a$, increment $v$ for the over-jump of the boundary of
$S$, and the size of the jump $w$ which exits above $S$:
$$
\Phi_{\uparrow,\,t}(a,\,v,\,w)= \sum_{n=1}^{\infty}
\textup{Pr}[X_{t_{n}}=a,\,K_{t_{n}}=w,  \,K_{t_{n}}-k-t^{\gamma}=v
\geq 0,\,\, K_{t_{m}}\in S     \text{ for } 0\leq m< n ]1_{
K_{t_{n} }>S}.
$$
The definition for  $\Phi_{\downarrow,\,t}(a,\,v,\,w)$ is
analogous. Let primed densities be normalized (e.g.
$\Phi_{\uparrow,\,t}^{\prime}=p_{\uparrow}^{-1}\Phi_{\uparrow,\,t}$
and
$\pi_{\uparrow,\,t}^{\prime}=p_{\uparrow,t}^{-1}\pi_{\uparrow,\,t}$,
where $\p_{\uparrow,t}$ is the probability of leaving  $S$ from
the top).  We will show that $\Phi_{\uparrow,\,t}^{\prime}$
converges in $L^{1}$ to
\begin{align} \label{Jass}
\Phi_{\infty}(a,v,w)= \pi_{\infty}(v,\,w) \frac{ \frac{\kappa(a)}{\bar{\kappa}} \mathcal{P}_{a}(w) }{ \tilde{P}(w) }.
\end{align}
Since $\phi^{\prime}_{\uparrow,t}(a,v)= \int_{\R^{+}}
dw\,\Phi_{\uparrow,\,t}(a,\,v,\,w)$, and $p_{\uparrow,\,t}$
converges to the probability that the random walk exits in the up
direction by Proposition~\ref{BC2NHalf}, this would prove the
result.  In particular, since $|d|\leq \frac{1}{2}t^{\gamma}$
neither of the probabilities $p_{\infty}(\uparrow,t)$ or
$p_{\infty}(\downarrow,t)$ will be close to zero.  Also by the proof
of~\ref{BC2NHalf}, the probability that the final momentum increment
$w=K_{t_{N}}-K_{t_{N-1}}$ is greater than $2t^{\gamma}$, where $t_{N}$
is the time of the last momentum jump leaving $S$ and $t_{N-1}$ is the time of
previous momentum jump, decays as $\mathit{O}(t^{-2\gamma})$.  This is true
also for the random walk case.  $\|\Phi_{\infty}\chi(w> 2t^{\gamma})\|_{1} $ and $\|\Phi_{\uparrow,\,t}\chi(w> 2t^{\gamma})\|_{1}$ thus vanish for large $t$.

 Define $\|g\|_{1}^{(t)}=\|g\,\chi(w\leq 2 t^{\gamma},\,v<w)\|_{1}$.  We placed the constraint $v<w$ in the indicator so that by Part (3) of Lemma~\ref{HighBrow}, $\pi_{\infty}^{\prime}(v,w)$ is strictly positive and, in particular, we can divide by it. $\pi_{\uparrow,\,t}^{\prime}(v,\,w)$ is also strictly positive for $v<w<  2 t^{\gamma}$, since there is a non-zero density for jumping from $k$ to $k+t^{\gamma}-w+v$ in one jump and then to $k+t^{\gamma}+v$ on a second jump.

  We now focus on showing that $\|\Phi_{\infty}-\Phi_{\uparrow,\,t}\|_{1}^{(t)}$ tends
 to zero.
\begin{multline}
\frac{\Phi_{\uparrow,\,t}^{\prime}(a,\,v,\,w)}{\pi_{\uparrow,\,t}^{\prime}(v,\,w)
}  =  \mathbb{E}\Big[\Pr\big[X_{t_{N}}=a\, \big| \,
K_{t_{N}}=t^{\beta}+t^{\gamma}+v,\,K_{t_{N}}-K_{t_{N-1}}=w,\,
X_{t_{N-1}}   \big] \Big],
\end{multline}
since
$\pi_{\uparrow,\,t}^{\prime}(v,\,w)=\int_{0}^{1}da\,\Phi_{\uparrow,\,t}^{\prime}(a,\,v,\,w)$.
We may write the conditional probability density in the
expectation above as
\begin{multline}\label{OK}
\Pr\big[X_{t_{N}}=a\, \big| \, K_{t_{N}}=t^{\beta}+t^{\gamma}+v,\,K_{t_{N}}-K_{t_{N-1}}=w,\,  X_{t_{N-1}}   \big]\\ = \frac{\mathcal{P}_{a}\big(w-\Delta(a)\big)\,\tilde{r}_{(X_{t_{N-1}},\,K_{t_{N-1}})}(a) }{\textup{P}_{(X_{t_{N-1}},\,K_{t_{N-1}})}(w)}.
\end{multline}
$\Delta(a)$ is a drift term which was defined in Lemma~\ref{JointDen}. $\textup{P}_{(x,\,k)}(w)$ are defined as in the proof of Lemma~\ref{JointDen} as the difference in momentum (sum of one momentum jump plus some drift) between a starting time with state $(x,k)$ and next time jump time.  $\tilde{r}_{(x,k)}(a)$ a defined as in Lemma~\ref{SingleJump} as the distribution for the position of the particle at the next momentum jump.

Analogously to $\|g\|_{1}^{(t)}$ define the semi-norm $\|g\|_{\infty}^{(t)}=\|g\,\chi(|w|\leq 2\,t^{\gamma},\,v<w) \|_{\infty}$.

Putting together~(\ref{Jass})-(\ref{OK}),
\begin{multline*}
 \| \frac{\Phi_{\uparrow,t}^{\prime}}{\pi^{\prime}_{\uparrow,t} }\frac{\pi_{\infty}}{\Phi_{\infty}}-1 \|_{\infty}^{(t)}\leq  \mathbb{E}\Big[\sup_{w\geq 2t^{\gamma},\, a }\Big| \tilde{r}_{(X_{t_{N-1}},\,K_{t_{N-1}})}(a)\frac{\bar{\kappa}}{\kappa(a)}\,\frac{ \tilde{P}(w) }{\textup{P}_{(X_{t_{N-1}},\,K_{t_{N-1}})}(w) }     -1             \Big|  \Big] \\ \leq \sup_{y\geq t^{\beta},\, a,\,x }\,\big[ \tilde{r}_{(x,\,y)}(a)\frac{\bar{\kappa}}{\kappa(a)}\big]\,\mathbb{E}\Big[\Big|\,\frac{ \tilde{P}(w) }{\textup{P}_{(X_{t_{N-1}},\,K_{t_{N-1}})}(w) }-1\Big]+            \mathbb{E}\Big[\sup_{a }\Big| \tilde{r}_{(X_{t_{N-1}},\,K_{t_{N-1}})}(a)\frac{\bar{\kappa}}{\kappa(a)}   -1             \Big|  \Big] \Big| \\ < \big(1+4\mathcal{R}\nu^{-2}t^{-\beta}+\mathit{O}(t^{-2\beta})\big)\big( 2ct^{\gamma-\beta} \big)   +4\mathcal{R}\nu^{-2}t^{-\beta}+\mathit{O}(t^{-2\beta})=\mathit{O}(t^{\gamma-2\beta}),
\end{multline*}
where in the strict inequality we have used Part (2) Lemma~\ref{SingleJump} and Lemma~\ref{JointDen}:
 $$\sup_{0\leq a\leq 1}\big|\tilde{r}_{(x,k)}(a)\frac{\bar{\kappa}}{\kappa(a)}-1  \big|\leq 8\mathcal{R}\nu^{-2}t^{-\beta},\, \sup_{w\in R^{+} }\big|\frac{ \textup{P}_{(x,k)}(w) }{\tilde{P}(w) }-1\big| \leq 2ct^{\gamma-\beta} , \text{ for } |k|\geq \frac{1}{2}t^{\beta}    $$
(where we doubled the constant factor in front of the higher order term to get ride of the lower), and we used that $|\frac{1}{b}-1|\leq 2\,b$ for $b$ in a small neighborhood around one.

By adding and subtracting by $
\Phi_{\uparrow,\,t}^{\prime}\frac{\pi_{\uparrow,\,\infty}^{\prime}
}{ \pi_{\uparrow,\,t}^{\prime}}   $,   and the triangle inequality
\begin{multline*}
\|\Phi^{\prime}_{\uparrow,\,t}-\Phi_{\uparrow,\,\infty}^{\prime} \|_{1}^{(t)}
 \leq \left \|  \Big(\frac{\Phi_{\uparrow,\,t}^{\prime}}{ \pi_{\uparrow,\,t}^{\prime} }   -\frac{\Phi_{\uparrow,\,\infty}^{\prime}}{ \pi_{\uparrow,\,\infty}^{\prime}} \Big)\pi_{\uparrow,\,\infty}^{\prime}\right \|_{1}^{(t)} +
\left \| \frac{\Phi_{\uparrow,\,t}^{\prime}}{
\pi_{\uparrow,\,t}^{\prime} }\Big(
\pi_{\uparrow,\,t}^{\prime}-\pi_{\uparrow,\,\infty}^{\prime}\Big)\right
\|_{1}^{(t)}
\\
\leq  \|\pi_{\uparrow,\,\infty}^{\prime}\|_{1} \left \|
\frac{\Phi_{\uparrow,\,t}^{\prime}}{ \pi_{\uparrow,\,t}^{\prime} }
-\frac{\Phi_{\uparrow,\,\infty}^{\prime}}{
\pi_{\uparrow,\,\infty}^{\prime} } \right \|_{\infty}^{(t)}  + \|
\pi_{\uparrow,\,t}^{\prime}-\pi_{\uparrow,\,\infty}^{\prime}
\|_{1} \left \| \frac{\Phi_{\uparrow,\,t}^{\prime}}{
\pi_{\uparrow,\,t}^{\prime} }\right \|_{\infty}^{(t)}.
\end{multline*}
$\|\pi_{\uparrow,\,\infty}^{\prime}\|_{1}=1$ and  so with the
analysis above the left term on the second line tends to zero. For
the right term, $\|
\pi_{\uparrow,\,t}^{\prime}-\pi_{\uparrow,\,\infty}^{\prime}
\|_{1}$ vanishes by Proposition~\ref{BC2NHalf} and $\left \|
\frac{\Phi_{\uparrow,\,t}^{\prime}}{ \pi_{\uparrow,\,t}^{\prime}
}\right \|_{\infty}^{(t)}$ is bounded by $\left \| \frac{\Phi_{\infty}}{
\pi_{\infty} }\right \|_{\infty}^{(t)} $ plus a small number by the
analysis above.

\end{proof}

Proposition~\ref{BC3NHalf} states the same results as in
Proposition~\ref{BC3} when there is some conditioning on the
future of the particle's trajectories.  For a particle starting
with momentum in $(t^{\beta},\,2t^{\beta})$,  let
$\theta_{\uparrow}$, $\theta_{\downarrow}$, and $\tau$ be the
first times the the particle has a momentum landing above
$2\,t^{\beta}$, below $-2\,t^{\beta}$ and below $t^{\beta}$
respectively.  By the same argument as in Lemma~\ref{FiniteTimes},
$\theta_{\uparrow}\wedge \theta_{\downarrow}$ has finite
expectation.

\begin{proposition}\label{BC3NHalf}
Assume~(\textit{i})-(\textit{ii}) from List~\ref{AssumpTwo}.
Consider the dynamics conditioned to have momentum jumps bounded
by $t^{\frac{\gamma}{2} }$. Let $(X_{s},K_{s})$ evolve from an
initial point $(x,k)$ with $
\frac{3}{2}t^{\beta}-t^{\frac{\gamma}{2} }\leq  k\leq
\frac{3}{2}t^{\beta}$.  Let  $S=\{y\in \R \, |\,  k-t^{\gamma}
<y<k+t^{\gamma}\} $ for $0< \gamma<\frac{1}{4}\beta $.

The boundary crossing densities $\psi_{\uparrow,\,t}$,\, $\psi_{\downarrow,\,t}$ for trajectories that are conditioned so that $\tau<\theta_{\uparrow}<\theta_{\downarrow}$,
have $L^{1}([0,1]\times\R^{+})$ convergence
$$\psi_{\uparrow,\,t}\longrightarrow \frac{1}{2}\, \phi_{\infty} ,\,\psi_{\downarrow,\,t}\longrightarrow \frac{1}{2}\,\phi_{\infty}.   $$

The convergence is uniform for the allowed range of $k$.  The same statements hold for $\psi_{\uparrow,\,t}$,\, $\psi_{\downarrow,\,t}$ for trajectories which are conditioned so that $\theta_{\downarrow}<\theta_{\uparrow}$.

\end{proposition}

\begin{proof}
Define the joint density $\Phi_{\uparrow,\,t}(a,\,v,\,w)$ for the
position $a$, increment $v$ for the over-jump of the boundary of
$S$, and the size of the jump $w$ which exits above $S$:
$$
\Phi_{\uparrow,\,t}(a,\,v,\,w)= \sum_{n=1}^{\infty}
\textup{Pr}[X_{t_{n}}=a,\,K_{t_{n}}=w,  \,K_{t_{n}}-k-t^{\gamma}=v
\geq 0,\,\, K_{t_{m}}\in S     \text{ for } 0\leq m< n ]1_{
K_{t_{n} }>S}.
$$
The definition for  $\Phi_{\downarrow,\,t}(a,\,v,\,w)$ is
analogous. Let primed densities be normalized (e.g.
$\Phi_{\uparrow,\,t}^{\prime}=p_{\uparrow}^{-1}\Phi_{\uparrow,\,t}$
and
$\pi_{\uparrow,\,t}^{\prime}=p_{\uparrow,t}^{-1}\pi_{\uparrow,\,t}$,
where $\p_{\uparrow,t}$ is the probability of leaving  $S$ from
the top).  We will show that $\Phi_{\uparrow,\,t}^{\prime}$
converges in $L^{1}$ to
\begin{align} \label{Jass}
\Phi_{\infty}(a,v,w)= \pi_{\infty}(v,\,w) \frac{ \frac{\kappa(a)}{\bar{\kappa}} \mathcal{P}_{a}(w) }{ \tilde{P}(w) }.
\end{align}
Since $\phi^{\prime}_{\uparrow,t}(a,v)= \int_{\R^{+}}
dw\,\Phi_{\uparrow,\,t}(a,\,v,\,w)$, and $p_{\uparrow,\,t}$
converges to the probability that the random walk exits in the up
direction by Proposition~\ref{BC2NHalf}, this would prove the
result.  In particular, since $|d|\leq \frac{1}{2}t^{\gamma}$
neither of the probabilities $p_{\infty}(\uparrow,t)$ or
$p_{\infty}(\downarrow,t)$ will be close to zero.  Also by the proof
of~\ref{BC2NHalf}, the probability that the final momentum increment
$w=K_{t_{N}}-K_{t_{N-1}}$ is greater than $2t^{\gamma}$, where $t_{N}$
is the time of the last momentum jump leaving $S$ and $t_{N-1}$ is the time of
previous momentum jump, decays as $\mathit{O}(t^{-2\gamma})$.  This is true
also for the random walk case.  $\|\Phi_{\infty}\chi(w> 2t^{\gamma})\|_{1} $ and $\|\Phi_{\uparrow,\,t}\chi(w> 2t^{\gamma})\|_{1}$ thus vanish for large $t$.

 Define $\|g\|_{1}^{(t)}=\|g\,\chi(w\leq 2 t^{\gamma},\,v<w)\|_{1}$.  We placed the constraint $v<w$ in the indicator so that by Part (3) of Lemma~\ref{HighBrow}, $\pi_{\infty}^{\prime}(v,w)$ is strictly positive and, in particular, we can divide by it. $\pi_{\uparrow,\,t}^{\prime}(v,\,w)$ is also strictly positive for $v<w<  2 t^{\gamma}$, since there is a non-zero density for jumping from $k$ to $k+t^{\gamma}-w+v$ in one jump and then to $k+t^{\gamma}+v$ on a second jump.

  We now focus on showing that $\|\Phi_{\infty}-\Phi_{\uparrow,\,t}\|_{1}^{(t)}$ tends
 to zero.
\begin{multline}
\frac{\Phi_{\uparrow,\,t}^{\prime}(a,\,v,\,w)}{\pi_{\uparrow,\,t}^{\prime}(v,\,w)
}  =  \mathbb{E}\Big[\Pr\big[X_{t_{N}}=a\, \big| \,
K_{t_{N}}=t^{\beta}+t^{\gamma}+v,\,K_{t_{N}}-K_{t_{N-1}}=w,\,
X_{t_{N-1}}   \big] \Big],
\end{multline}
since
$\pi_{\uparrow,\,t}^{\prime}(v,\,w)=\int_{0}^{1}da\,\Phi_{\uparrow,\,t}^{\prime}(a,\,v,\,w)$.
We may write the conditional probability density in the
expectation above as
\begin{multline}\label{OK}
\Pr\big[X_{t_{N}}=a\, \big| \, K_{t_{N}}=t^{\beta}+t^{\gamma}+v,\,K_{t_{N}}-K_{t_{N-1}}=w,\,  X_{t_{N-1}}   \big]\\ = \frac{\mathcal{P}_{a}\big(w-\Delta(a)\big)\,\tilde{r}_{(X_{t_{N-1}},\,K_{t_{N-1}})}(a) }{\textup{P}_{(X_{t_{N-1}},\,K_{t_{N-1}})}(w)}.
\end{multline}
$\Delta(a)$ is a drift term which was defined in Lemma~\ref{JointDen}. $\textup{P}_{(x,\,k)}(w)$ are defined as in the proof of Lemma~\ref{JointDen} as the difference in momentum (sum of one momentum jump plus some drift) between a starting time with state $(x,k)$ and next time jump time.  $\tilde{r}_{(x,k)}(a)$ a defined as in Lemma~\ref{SingleJump} as the distribution for the position of the particle at the next momentum jump.

Analogously to $\|g\|_{1}^{(t)}$ define the semi-norm $\|g\|_{\infty}^{(t)}=\|g\,\chi(|w|\leq 2\,t^{\gamma},\,v<w) \|_{\infty}$.

Putting together~(\ref{Jass})-(\ref{OK}),
\begin{multline*}
 \| \frac{\Phi_{\uparrow,t}^{\prime}}{\pi^{\prime}_{\uparrow,t} }\frac{\pi_{\infty}}{\Phi_{\infty}}-1 \|_{\infty}^{(t)}\leq  \mathbb{E}\Big[\sup_{w\geq 2t^{\gamma},\, a }\Big| \tilde{r}_{(X_{t_{N-1}},\,K_{t_{N-1}})}(a)\frac{\bar{\kappa}}{\kappa(a)}\,\frac{ \tilde{P}(w) }{\textup{P}_{(X_{t_{N-1}},\,K_{t_{N-1}})}(w) }     -1             \Big|  \Big] \\ \leq \sup_{y\geq t^{\beta},\, a,\,x }\,\big[ \tilde{r}_{(x,\,y)}(a)\frac{\bar{\kappa}}{\kappa(a)}\big]\,\mathbb{E}\Big[\Big|\,\frac{ \tilde{P}(w) }{\textup{P}_{(X_{t_{N-1}},\,K_{t_{N-1}})}(w) }-1\Big]+            \mathbb{E}\Big[\sup_{a }\Big| \tilde{r}_{(X_{t_{N-1}},\,K_{t_{N-1}})}(a)\frac{\bar{\kappa}}{\kappa(a)}   -1             \Big|  \Big] \Big| \\ < \big(1+4\mathcal{R}\nu^{-2}t^{-\beta}+\mathit{O}(t^{-2\beta})\big)\big( 2ct^{\gamma-\beta} \big)   +4\mathcal{R}\nu^{-2}t^{-\beta}+\mathit{O}(t^{-2\beta})=\mathit{O}(t^{\gamma-2\beta}),
\end{multline*}
where in the strict inequality we have used Part (2) Lemma~\ref{SingleJump} and Lemma~\ref{JointDen}:
 $$\sup_{0\leq a\leq 1}\big|\tilde{r}_{(x,k)}(a)\frac{\bar{\kappa}}{\kappa(a)}-1  \big|\leq 8\mathcal{R}\nu^{-2}t^{-\beta},\, \sup_{w\in R^{+} }\big|\frac{ \textup{P}_{(x,k)}(w) }{\tilde{P}(w) }-1\big| \leq 2ct^{\gamma-\beta} , \text{ for } |k|\geq \frac{1}{2}t^{\beta}    $$
(where we doubled the constant factor in front of the higher order term to get ride of the lower), and we used that $|\frac{1}{b}-1|\leq 2\,b$ for $b$ in a small neighborhood around one.

By adding and subtracting by $
\Phi_{\uparrow,\,t}^{\prime}\frac{\pi_{\uparrow,\,\infty}^{\prime}
}{ \pi_{\uparrow,\,t}^{\prime}}   $,   and the triangle inequality
\begin{multline*}
\|\Phi^{\prime}_{\uparrow,\,t}-\Phi_{\uparrow,\,\infty}^{\prime} \|_{1}^{(t)}
 \leq \left \|  \Big(\frac{\Phi_{\uparrow,\,t}^{\prime}}{ \pi_{\uparrow,\,t}^{\prime} }   -\frac{\Phi_{\uparrow,\,\infty}^{\prime}}{ \pi_{\uparrow,\,\infty}^{\prime}} \Big)\pi_{\uparrow,\,\infty}^{\prime}\right \|_{1}^{(t)} +
\left \| \frac{\Phi_{\uparrow,\,t}^{\prime}}{
\pi_{\uparrow,\,t}^{\prime} }\Big(
\pi_{\uparrow,\,t}^{\prime}-\pi_{\uparrow,\,\infty}^{\prime}\Big)\right
\|_{1}^{(t)}
\\
\leq  \|\pi_{\uparrow,\,\infty}^{\prime}\|_{1} \left \|
\frac{\Phi_{\uparrow,\,t}^{\prime}}{ \pi_{\uparrow,\,t}^{\prime} }
-\frac{\Phi_{\uparrow,\,\infty}^{\prime}}{
\pi_{\uparrow,\,\infty}^{\prime} } \right \|_{\infty}^{(t)}  + \|
\pi_{\uparrow,\,t}^{\prime}-\pi_{\uparrow,\,\infty}^{\prime}
\|_{1} \left \| \frac{\Phi_{\uparrow,\,t}^{\prime}}{
\pi_{\uparrow,\,t}^{\prime} }\right \|_{\infty}^{(t)}.
\end{multline*}
$\|\pi_{\uparrow,\,\infty}^{\prime}\|_{1}=1$ and  so with the
analysis above the left term on the second line tends to zero. For
the right term, $\|
\pi_{\uparrow,\,t}^{\prime}-\pi_{\uparrow,\,\infty}^{\prime}
\|_{1}$ vanishes by Proposition~\ref{BC2NHalf} and $\left \|
\frac{\Phi_{\uparrow,\,t}^{\prime}}{ \pi_{\uparrow,\,t}^{\prime}
}\right \|_{\infty}^{(t)}$ is bounded by $\left \| \frac{\Phi_{\infty}}{
\pi_{\infty} }\right \|_{\infty}^{(t)} $ plus a small number by the
analysis above.

\end{proof}

In the next corollary, like Proposition~\ref{BC1}, there is only one boundary.  However, we place an optional time constraint $\varrho\,t$, $\varrho\in (0,\infty]$ here on the amount of time the particle is allowed to have before reaching the boundary.

\begin{corollary}\label{BC4}
Assume~(\textit{i})-(\textit{ii}) from List~\ref{AssumpTwo}.  Let $(X_{s},K_{s})$ evolve according to the dynamics from some initial point $(x,k)$ for $2t^{\beta}\leq k \leq 4\,t^{\beta}$, $\frac{1}{4}\leq \beta <\frac{1}{2} $.  Define the density
$$\phi_{t}(a,\,v)= \sum_{n=1}^{\infty}\textup{Pr}[X_{t_{n}}=a,\, t^{\beta}-K_{t_{n}}= v\geq 0,\,\, K_{t_{m}}>t^{\beta}     \text{ for } 0\leq m< n,\,t_{n}< \varrho t ].$$
In the limit $t\rightarrow \infty$,
$$\phi_{t}(a,\,v)\rightarrow \phi_{\infty}(a,\,v).$$
The convergence is uniform over all starting points $(x,\,k)$.

\end{corollary}

\begin{proof}
First, we will argue that the probability $p_{t}=\|\phi_{t}\|_{1}$ that the particle jumps below $t^{\beta}$ in the interval $[0,\varrho\, t]$ approaches one as $t\rightarrow \infty$ for $\varrho\in \R^{+}$.
\begin{multline}
p_{t}=\Pr\Big[ \inf_{0\leq r \leq \varrho t}t^{-\frac{1}{2}}\,K_{r}\geq t^{\beta-\frac{1}{2}}        \Big]\\ \geq  \Pr\Big[ t^{\beta-\frac{1}{2}}+t^{-\frac{1}{2}}k  +\sup_{0\leq r \leq \varrho \,t}      \big|t^{-\frac{1}{2}}\,\int_{0}^{r}ds\,\frac{dV}{dx}(X_{s})  \big| \leq \sup_{0\leq r \leq \varrho t} -t^{-\frac{1}{2}}M_{r}       \Big],
\end{multline}
since $K_{r}=k+M_{r}+\int_{0}^{r}ds\,\frac{dV}{dx}(X_{s})$.  By Chebyshev's inequality and that $k\leq 4\,t^{\beta}$.
\begin{multline}\label{Hitomi}
\Pr\Big[t^{\beta-\frac{1}{2}}+t^{-\frac{1}{2}}k +\sup_{0\leq r \leq \varrho\,t} \big|t^{-\frac{1}{2}}\,\int_{0}^{r}ds\,\frac{dV}{dx}(X_{s})  \big| \geq \epsilon \Big] \\ \leq \epsilon^{-2}\,\mathbb{E}\Big[ 5\,t^{\beta-\frac{1}{2}}+\sup_{0\leq r \leq \varrho\, t} \Big|t^{-\frac{1}{2}}\,\int_{0}^{r}ds\,\frac{dV}{dx}(X_{s})   \Big|^{2} \Big].
\end{multline}
The right-side converges to zero by Lemma~\ref{HEDrift}, and thus the probability that the supremum of the drift is greater than any finite $\epsilon_{1}$ tends to zero.

By~(\ref{MartBrown}), $t^{-\frac{1}{2}}M_{st}$, $s\in[0,1]$ converges to a Brownian motion $\mathbf{B}_{s}$.  Since the supremum over an interval $[0,\,1]$ is a uniformly continuous functional on elements $L^{\infty}[0,1]$, $\sup_{0\leq s \leq st} -t^{-\frac{1}{2}}M_{st}$ converges in distribution to $\sup_{0\leq s\leq 1}-\mathbf{B}_{s}$. In the large $t$ limit,
$$\Pr[\sup_{0\leq s \leq \varrho\,t} -t^{-\frac{1}{2}}M_{st}>2\epsilon ]\longrightarrow \Pr [\sup_{0\leq s \leq \varrho} -\mathbf{B}_{s}>2\epsilon ]. $$
Thus we can pick $\epsilon$ so that $\Pr[\sup_{0\leq s \leq 1} -\mathbf{B}_{s}>2\epsilon ]$ is close to one, and then pick $t$ so that both $\Pr[\sup_{0\leq s \leq \varrho\,t} -t^{-\frac{1}{2}}M_{st}>2\epsilon ]$ is close to $\Pr[\sup_{0\leq s \leq \varrho} -\mathbf{B}_{s}>2\epsilon ]$ and~(\ref{Hitomi}) is close to zero.  By the inclusion exclusion principle
$$ p_{t}\geq 1-2\Pr[\sup_{0\leq s \leq \varrho\,t} -t^{-\frac{1}{2}}M_{st}>2\epsilon ]\wedge \Pr\Big[t^{\beta-\frac{1}{2}}+t^{-\frac{1}{2}}k +\sup_{0\leq r \leq \varrho\, t} \big|t^{-\frac{1}{2}}\,\int_{0}^{r}ds\,\frac{dV}{dx}(X_{s})  \big| \geq \epsilon \Big],$$
and so $p_{t}$ can be made arbitrarily close to $1$.

Now we continue with showing $L^{1}$ convergence of $\phi_{t}$ to $\phi_{\infty}$.  Since $(1-\|\phi_{t}\|)$ is negligible, it is sufficient to show $\phi_{t}\rightarrow (1-\|\phi_{t}\|)\phi_{\infty}$. Let $\gamma<\frac{1}{3}\beta$. We construct a sequence of hitting times
\begin{eqnarray*}
\sigma_{j}&=& \textup{min}\{s\in [0,\infty) \,\big| \, s>\theta_{j-1},\,  K_{s}< t^{\beta}+t^{\gamma} \},\\
\theta_{j}&=& \textup{min}\{s\in [0,\infty) \,\big| \, s>\sigma_{j},\,  K_{s}> t^{\beta}+2t^{\gamma}  \}.
\end{eqnarray*}
where we set $\theta_{0}=0$.  Let $N^{\prime}$ be the number of $\sigma_{j}$'s in the interval $[0,\varrho t]$ before the first time $t_{N}$ that $K_{t_{N}}$ jumps below $t^{\beta}$.  The process will usually have its first jump below $t^{\beta}$ from a point with momentum in the interval $[t^{\beta},\,t^{\beta}+\frac{1}{2}t^{\gamma}]$.  In other words, $\Pr\big[|K_{t_{N-1}}|>t^{\gamma} \big]$ is small.   If, as defined in Lemma~\ref{BoundryVariance}, $\pi_{\downarrow}(v)$ is the distribution for the over-jump for the lower boundary of the set $S=[t^{\beta}+t^{\gamma},\,\infty)$ starting from some point in $S$, then by Lemma~\ref{BoundryVariance}
$\int_{\R^{+}}dv\,\pi_{\downarrow}(v)\,v^{2}<J$.
By Chebyshev's inequality
 $$ \Pr\big[|K_{t_{N-1}}|>t^{\beta}+\frac{1}{2}t^{\gamma} \big]\leq \int_{\R^{+}}dv\,\pi_{\downarrow}(v)\chi(|v|\geq \frac{1}{2}\,t^{\gamma})\leq \frac{J}{4}\,t^{-2\gamma}.  $$

The above comments allow us to write the following
$$\left \|\phi_{t}- \mathbb{E}\Big[\sum_{j=1}^{\infty}\chi\big(N=j;\, K_{\sigma_{j}}-t^{\beta} \geq \frac{1}{2}\,t^{\gamma}  \big)\,\frac{\phi^{(X_{\sigma_{j}},\, K_{\sigma_{j}})}_{\downarrow,\,t} }{p(X_{\sigma_{j}},\, K_{\sigma_{j}}) }   \Big] \right \|_{1}\leq  \Pr\big[|K_{t_{N-1}}|>t^{\gamma} \big]\leq J\,t^{\gamma},
$$
where $\phi^{(X_{\sigma_{j}},\, K_{\sigma_{j}})}_{\downarrow,\,t}$ is the lower boundary crossing distribution as in Proposition~\ref{BC3} starting from the point $(X_{\sigma_{j}},\, K_{\sigma_{j}})$, and $\int da\, dv\, \phi^{(X_{\sigma_{j}},\, K_{\sigma_{j}})}_{\downarrow,\,t}(a,\,v)=p(X_{\sigma_{j}},\, K_{\sigma_{j}})$ is the probability of exiting the domain $[t^{\beta},\,t^{\beta}+2t^{\gamma})$ at the lower boundary.  The difference $d$ between the distance to the upper and lower boundaries  $2t^{\beta}$ and $t^{\beta}+t^{2\gamma}$ respectively is $d=2(t^{\beta}+t^{\gamma}-K_{\sigma_{j}})$.
\begin{multline*}
\left \|\mathbb{E}\Big[\sum_{j=1}^{\infty}\chi\big(N=j;\, K_{\sigma_{j}}-t^{\beta} \geq \frac{1}{2}\,t^{\gamma}  \big)\,\frac{\phi^{(X_{\sigma_{j}},\, K_{\sigma_{j}})}_{\downarrow,\,t}}{p(X_{\sigma_{j}},\, K_{\sigma_{j}})}  \Big]-p_{t}\,  \pi_{\infty} \right \|_{1}\\ \leq \sum_{n=1}^{\infty}\Pr\big[\chi(N^{\prime}=j)\big]\,    \sup_{y-t^{\beta}\geq \frac{1}{2}t^{\gamma},\, a } \left \| \frac{\phi^{(a,\,y)}_{\downarrow,\,t}}{p(a,\, y)}  - \pi_{\infty} \right  \|_{1}\leq \sup_{y-t^{\beta}\geq \frac{1}{2}t^{\gamma},\, a } \left \| \frac{\phi^{(a,\,y)}_{\downarrow,\,t}}{p(a,\, y)}  - \pi_{\infty} \right  \|_{1}
\end{multline*}
By Proposition~\ref{BC3} $p(a,\, y)$, $y\geq
t^{\beta}+\frac{1}{2}t^{\gamma}$ converge uniformly to the the
probabilities $p_{\downarrow}$ for the averaged random walk to
exit in down direction. $p_{\downarrow}$ are bounded away from
zero, since the ratio of the distance to the lower boundary to the
upper boundary is less than or equal to $3$.  Moreover, by
Proposition~\ref{BC4} $\phi^{(a,\,y)}_{\downarrow,\,t}$ converges
uniformly to $p_{\downarrow}\pi_{\infty}$.

\end{proof}

\end{appendix}

\end{document}